\documentclass[11pt]{article}

\usepackage{latexsym} 
\usepackage{amssymb,amsmath}
\usepackage{epsfig}
\usepackage{pstricks}
\usepackage{mathrsfs}

\usepackage{graphicx}
\usepackage{xspace}

\usepackage{fullpage}

\usepackage{tikz}
\usetikzlibrary{automata}
\usetikzlibrary{snakes,arrows,shapes}
\usetikzlibrary{calc}

\newcounter{newstuff}

\newcommand{\support}{\mathtt{support}}
\newcommand{\sort}{\mathtt{sort}}

\newcommand{\children}{{\mathtt{Ch}}}
\newcommand{\child}{{\mathtt{Ch}}}

\newcommand{\Pl}{{\mathscr{P}}}
\newcommand{\Prob}{{\mathbb{P}}}
\newcommand{\Probb}{{\mathcal{P}}}

\newcommand{\Leaves}{{\mathbb{L}}}
\newcommand{\LLeaves}{\vec{\mathbb{L}}}
\newcommand{\Internal}{{\mathbb{W}}}
\newcommand{\Act}{{\mathtt{Act}}}
\newcommand{\height}{{\mathtt{h}}}
\newcommand{\Info}{{\mathcal{I}}}

\newcommand{\Agent}{\mathcal{A}}
\newcommand{\Normal}{\mathcal{N}}
\newcommand{\Normalize}{\mathfrak{D}}
\newcommand{\calG}{\mathcal{G}}
\newcommand{\calA}{\mathcal{A}}
\newcommand{\calE}{\mathcal{E}}
\newcommand{\SysB}{\mathfrak{B}}

\newcommand{\calV}{\mathcal{V}}
\newcommand{\calF}{\mathcal{F}}

\newcommand{\MU}{\ensuremath{\mathsf{K}}}

\newcommand{\PPAD}{\ensuremath{\mathsf{PPAD}}}
\newcommand{\PLS}{\ensuremath{\mathsf{PLS}}}
\newcommand{\TFNP}{\ensuremath{\mathsf{TFNP}}}
\newcommand{\FIXP}{\ensuremath{\mathsf{FIXP}}}
\newcommand{\linFIXP}{\ensuremath{\mathsf{linear\mbox{-}FIXP}}}
\newcommand{\FIXPA}{\ensuremath{\mathsf{FIXP}_a}}

\def\nat{{\mathbb N}}
\def\real{{\mathbb R}}
\def\rat{{\mathbb Q}}

\newcommand{\size}{\mathtt{size}}

\newcommand{\Top}{\mathtt{Top}}

\newenvironment{proof}{\noindent \textbf{Proof. }}{\medskip}

\newtheorem{theorem}{Theorem}

\newtheorem{corollary}[theorem]{Corollary}
 
\newtheorem{lemma}[theorem]{Lemma}
\newtheorem{proposition}[theorem]{Proposition} 

\newtheorem{claim}{Claim}

\newcommand{\commentout}[1]{}

\newcommand{\qed}%
{\penalty 1000 \hfill\penalty 1000\rule{5pt}{5pt}\par\medskip}
\newcommand{\rfotwo}%
{\mbox{\rm FO}^2[<]}
\newcommand{\tl}{\mathrm{TL}}            %

\newcommand{\utl}%
{\mbox{{\rm unary-}}\tl}

\begin{document}

\title{The complexity of computing a (quasi-)perfect equilibrium for an\\  
$n$-player extensive form game
of perfect recall} 
\author{Kousha Etessami\\{\small University of Edinburgh}\\Email: {\tt kousha "at" inf.ed.ac.uk}}

\date{}

\maketitle

\thispagestyle{empty}

\begin{abstract}
We study the complexity 
of computing or approximating refinements of Nash equilibrium
for a given finite $n$-player extensive form game of perfect recall (EFGPR),
where $n \geq 3$.

Our results apply to a number of well-studied refinements:
to sequential equilibrium (SE), which refines both Nash 
and subgame-perfect equilibrium; to 
extensive-form perfect equilibrium (PE), which refines SE;
to normal-form perfect equilibrium (NF-PE);
and to quasi-perfect equilibrium (QPE) which refines both SE and NF-PE.
Of these, the two most refined notions are PE and QPE, which are
incompatible with each other.
By a classic result of Selten (and by a result of van Damme) 
a PE (respectively, a QPE) exists for any EFGPR.  

We show that,
for all these notions of equilibrium,
approximating an equilibrium for a given EFGPR, to within 
a given desired precision,
is $\FIXPA$-complete.
We also consider the complexity of corresponding ``almost'' equilibrium notions,
and show that they are PPAD-complete. 
In particular, we define 
{\em $\delta$-almost $\epsilon$-(quasi-)perfect equilibrium},
and show that
computing one for a given EFGPR, given $\delta > 0$ 
and $\epsilon > 0$, is $\PPAD$-complete.
We show that these notions refine the notion of {\em $\delta$-almost
subgame-perfect equilibrium} for
EFGPRs, which is PPAD-complete.

Thus, approximating one such ($\delta$-almost) equilibrium for
$n$-player EFGPRs, $n \geq 3$, is P-time equivalent to approximating a
($\delta$-almost) NE for a {\em normal form} game with 3 or more
players.  Normal form games are trivially encodable as EFGPRs without
blowup in size. Thus our results extend the celebrated complexity
results for Nash equilibrium in normal form games to various 
refinements of equilibrium in the more general setting
of EFGPRs.

For 2-player EFGPRs, analogous complexity results follow from the algorithms of
Koller, Megiddo, and von Stengel \cite{KolMegvonSteng:1996,vonStengel-efficient-1996}, von Stengel, van den Elzen, and Talman \cite{vonStengel-van-den-elzen-talman:2002}, and Mitersen and S{\o}rensen \cite{MiltSor06,MiltSoren10}.
For $n$-player EFGPRs, an analogous result for Nash equilibrium
and for subgame-perfect equilibrium was given by
Daskalakis, Fabrikant, and Papadimitriou \cite{DFP06}.
No analogous results were known
for the more refined notions of equilibrium for EFGPRs with $3$ or more players.

\commentout{
By contrast to the prior work on 2-player EFGPRs, our results 
make no use of
the {\em sequence form} for EFGPRs. 
By constrast to the work of Daskalakis et. al. \cite{DFP06}, we make no use
of reductions via graphical games.

Instead, we use direct fixed point characterizations
obtained by combining older 
insights from Kuhn and Selten's original {\em agent normal form} for EFGPRs,
and Myerson's alternative definition of PE via $\epsilon$-PEs, 
with insights from new algebraic fixed point functions 
for ($\epsilon$-perfect) equilibria 
of $n$-player normal form games, $n \geq 3$, 
developed recently in \cite{EY07} and \cite{EHMS14}. }

\end{abstract}

\section{Introduction}

{\em Extensive form games} are the fundamental 
mathematical model of  
games that transpire as a sequence
of moves by players over time.
A finite extensive form game is described by a finite tree, 
where each internal node belongs to one of the players (or to chance),
and where each leaf indicates a payoff to every player.  
A ``play'' of the game traces a path in this tree from the root to a leaf,
with each player choosing the child to move to at nodes belonging to it
(the child being chosen randomly at chance nodes, or when players decide
to randomize their moves).  
In general, an extensive form game may be of {\em imperfect information},
meaning roughly that players may need to make moves
without having full knowledge of the current ``state'' 
(i.e., current node of the game tree).
A basic sanity condition for imperfect
information games,  called {\em perfect recall},
requires (roughly) that every player in the game should 
remember all of its own prior moves.
This condition was already put forward
by Kuhn (\cite{Kuhn53}), who showed that games with
perfect recall have nice properties and avoid certain 
pathologies of general extensive form games.
Subsequently, Selten \cite{Selten75}, 
in his seminal work  on 
perfect\footnote{{\em Warning:} 
the word ``perfect'' is over-used as a modifier in game theory: it 
refers to various 
not-necessarily-related concepts, depending on what word it modifies
(  ``information'', ``recall'', ``equilibrium'', $\ldots$  ,  etc. ). }
equilibria,
argued that non-cooperative extensive form games that lack
perfect recall should be rejected as misspecified models.
The assumption of perfect recall has indeed become standard practice in
much of the large literature on
extensive form games. 
Henceforth, we use the abbreviations: EFGPR for ``extensive form game of
perfect recall'',  EFG for ``extensive form game'',
and NFG for ``normal form game''. 

Selten's work made clear that 
Nash equilibrium, and even subgame-perfect equilibrium,
is inadequately refined as a solution concept
for extensive form games.
In particular, there are Nash and subgame-perfect equilibria
of EFGPRs that involve ``non-credible threats'',
rendering them implausible. 
Motivated by this, Selten defined  
a more refined notion of perfect equilibrium,
based on {\em ``trembling hand''} perfection, 
and showed that any EFGPR has at least one perfect equilibrium.
(Selten was awarded a Nobel prize in economics, 
 together with Nash and Harsanyi,
 largely for his work on refinement 
 of equilibria.)
Subsequent work, e.g., by Kreps and Wilson 
on sequential equilibria \cite{Kreps-Wilson:1982},
and by many others, has reaffirmed the imperative
for considering refinements of equilibrium,
especially for extensive form games. 
By now EFGPRs, and equilibrium refinements 
for them, are treated in most standard textbooks
on game theory (see, e.g., \cite{OsRub94,Myerson-book-1997,MasSolZam13,vanDamme91}).

This paper studies the complexity of 
computing or approximating an equilibrium for 
a given EFGPR, with $n \geq 3$ players.  We study 
various important
refinements of NE, 
including:  sequential equilibrium (SE),
extensive form trembling-hand perfect equilibrium (PE),
and quasi-perfect equilibrium (QPE).
All of these notions refine subgame-perfect equilibrium (SGPE).
Of these, PE and QPE
are the most refined notions.\footnote{However, unlike PE and QPE, 
an SE consists not just of 
a suitable behavior profile, but also a {\em system of beliefs}.
We'll see later in what sense PE (and QPE) ``refines'' 
SE (\cite{Kreps-Wilson:1982}).
Our complexity results for SE are also for computing 
its associated belief system.}
Quasi-perfect equilibrium (QPE), defined by van Damme
\cite{vanDamme84}, is incompatible
with PE, meaning that a PE need not be a QPE and a QPE need not be a PE.
Like PE,  QPE also refines NE, SGPE, and SE.
Furthermore, QPE also refines ``normal-form perfect 
equilibrium'' (NF-PE) for EFGPRs, which differs from, and is incompatible with 
(extensive form) PE for EFGPRs.    For the benefit of readers 
confused by all the different mentioned notions of 
equilibrium for EFGPRs,
Figure \ref{fig:refinement-hasse}
of Section \ref{sec:background}
summarizes the refinement relationships that exist (and don't exist) 
between them, by depicting the Hasse diagram of their refinement partial order.

Informally, we show that for all these notions of equilibrium,
approximating an equilibrium for given $n$-player EFGPR
within a given desired precision $\delta > 0$
(or computing an ``$\delta$-almost equilibrium'' for given $\delta > 0$)
is no harder than approximating a ($\delta$-almost) NE for a given
$3$-player {\em normal form game}.
NFGs are trivially encodable as EFGPRs without blowup
in size.  Thus our results extend the celebrated complexity results
for computing/approximating an NE for NFGs to
the much more general setting of EFGPRs,
and furthermore ``perfection comes at 
no extra cost in complexity''.
Before stating our results more precisely, we must first discuss prior related work. 

\vspace*{0.05in}

\noindent {\bf Related work.}
Equilibrium computation,
and  its connection to fixed point computation,
has been studied for decades, both for normal form and extensive
form games.
Papadimitriou \cite{Papa_parity_94} defined the search problem
complexity class $\PPAD$ in order to capture the complexity of
problems related to computing an equilibrium.\footnote{It is
well-known 
that already for 2-player NFGs, 
computing a {\em specific} NE, e.g., that optimizes
total payoff or other objectives,
is NP-hard \cite{GilZem89,ConSan03}. So, in this paper, whenever we speak
of a problem of computing (or approximating) ``an'' equilibrium,
possibly of a {\em refined} kind, we are not more specific than that: 
{\em any} equilibrium
{\em of that kind} will do.}
It follows from the correctness of the Lemke-Howson algorithm 
that computing an NE for 2-player NFGs is in $\PPAD$.
It similarly follows from  Scarf's algorithm that given an 
$n$-player NFG (for any $n$),
and given $\epsilon > 0$,
computing 
a ``$\epsilon$-NE'' (which we call ``$\epsilon$-almost-NE'' in this paper, to avoid confusion with other notions\footnote{We do so 
to avoid confusion when we combine ``$\epsilon$-almost''
with other notions, particularly
Myerson's $\epsilon$-PEs (\cite{Myerson78}).})
is in $\PPAD$;   this is
a strategy profile where  
no player can improve its own payoff by
more than $\epsilon$ by unilaterally deviating from its strategy. 
In a celebrated series of result in 2006, 
Chen and Deng \cite{Chen-Deng06}, and  
Daskalakis et. al. \cite{DasGP09}, showed that
both of these problems are $\PPAD$-complete.
For games with $3$ (or more) players,
specified by an integer payoff table,  
all the NEs
may have irrational numbers (\cite{Nash51}).
Thus, we can not compute an NE {\em exactly} for them (at least not in
the Turing model of computation).
With
Yannakakis in \cite{EY07},
we showed that for games with 3 (or more) players,
an $\epsilon$-NE may in fact
be nowhere near
any actual NE, unless $\epsilon > 0$ is so small that its 
{\em binary} encoding size is exponential in the size of the
game; thus, finding an $\epsilon$-NE may
tell us nothing about the location of any actual NE. 
In \cite{EY07} we considered the complexity of
computing an actual NE
to within a desired number of bits of precision, i.e.,
computing a strategy profile that has
$\ell_\infty$-distance at most $\delta > 0$ to
some NE, for a given $\delta$.
We showed that this problem is complete
for a
natural complexity class which we called 
\FIXPA.\footnote{We also showed in \cite{EY07} that
approximating an actual NE, even within $\ell_\infty$-distance $c < 1/2$
for 3-player NFGs, 
is {\em ``hard''}
in that even placing this in NP would resolve
long standing open problems in arithmetic vs. Turing complexity.}
Informally, \FIXPA\  
is the class of discrete search problems that can
be reduced to approximating, within desired $\ell_\infty$-distance 
$\delta > 0$, a (any) Brouwer fixed point of a continuous function given
by an algebraic circuit using gates $\{+,-,*,/,\max,\min\}$.
(We will later formally define \FIXPA, as well as its {\em real-valued}
 progenitor $\FIXP$, and the piecewise-linear fragment $\linFIXP$ ($= \PPAD$).)
Very recently, in a paper with Hansen, Miltersen, and S{\o}rensen 
\cite{EHMS14}, building on \cite{EY07},
we have shown that for NFGs with $n \geq 3$ players,
approximating a ``trembling-hand perfect equilibrium'' (PE)
within desired precision is also $\FIXPA$-complete.
Since PEs refine NEs, we only had to show containment
in $\FIXPA$.   Interestingly, it was shown previously
in \cite{HMS10} that given a 3-player NFG, 
deciding whether a {\em given} strategy profile is a PE is NP-hard
(unlike for NEs, for which this is easily in P-time).

Research on the computation of equilibria 
for EFGs, with and without perfect recall,  
also has a long and rich history.
Of course for perfect information games 
computing a NE or SGPE
is easily in P-time using Kuhn's classic
``backward induction'' (\cite{Kuhn53}).
On the other hand, for imperfect information games {\em without} perfect
recall,  it was pointed out by Koller and Megiddo 
\cite{koller1992complexity} (and by others, e.g., \cite{blair1993games})
that even for $1$-player games computing or approximating 
a (any) NE is  NP-hard  (it can easily encode 3SAT).
By contrast, for 1-player EFGPRs an equilibrium (i.e., an optimal strategy)
can be computed easily in P-time by dynamic programming, 
as shown by Wilson \cite{wilson1972}.

Of course, one way to compute an equilibrium for an EFGPR (or EFG) is to 
first convert it to an NFG,
and then apply any algorithm applicable to NFGs.
The problem with this approach is that, 
even for EFGPRs, a {\em standard} conversion from extensive to 
normal
form incurs exponential
blowup.\footnote{Even notions of {\em reduced} normal 
form in general incur exponential blowup for EFGPRs.  We will not
elaborate on reduced norm form, but roughly it means redundant strategies
of the EFGPR are not considered in the normal form.} 
  Thus, even a P-time algorithm
for NFGs requires  exponential time if 
applied naively in this way to EFGPRs.
In the other direction, an NFG can trivially be encoded
as an ``equivalent'' EFGPR which is not much bigger,
so that any equilibrium computation problem for NFGs is P-time
reducible to an analogous problem for EFGPRs.

In a series of important works in the 1990s,  Koller, Megiddo, and von Stengel
\cite{koller1992complexity,vonStengel-efficient-1996,
KolMegvonSteng:1996} 
obtained equilibrium algorithms for 2-player
EFGPRs with complexity bounds
that essentially match those of 2-player NFGs.   In particular,  Koller
and Megiddo \cite{koller1992complexity} 
showed that for 2-player zero-sum EFGPRs an
NE (i.e., a minimax profile) in behavior strategies
can be computed in P-time using linear programming.
Furthermore, by using the {\em sequence form} 
(\cite{Romanovskii62,vonStengel-efficient-1996})
of  EFGPRs,  Koller, Megiddo, and von Stengel
(\cite{KolMegvonSteng:1996}) 
showed that one can apply variants of Lemke's algorithm to
certain LCPs associated with 2-player EFGPRs to compute an (exact) NE 
in behavior strategies.
A consequence of their result (when combined 
with Chen and Deng's $\PPAD$-hardness result
for 2-player NFGs \cite{Chen-Deng06}) is
that computing an NE for 2-player EFGPRs is $\PPAD$-complete.
Later, von Stengel, van den Elzen, and Talman \cite{vonStengel-van-den-elzen-talman:2002}, using the sequence form, gave a similar Lemke-like algorithm 
for computing  
a ``normal form perfect equilibrium'' (NF-PE)
\footnote{A {\em normal-form perfect equilibrium} (NF-PE), 
is a (behavior) profile that induces a (mixed profile) 
PE of the {\em standard} NFG associated with the 2-player EFGPR.
In general, this is {\em not} equivalent to  (extensive form) PE
for EFGPRs (see \cite{vanDamme91}, Chapter 6).
In fact, unlike (extensive form) PE, a NF-PE need not 
 even give a subgame-perfect equilibrium of the EFGPR.
We will revisit the distinction
between NF-PE and PE for EFGPRs when we provide formal
definitions. Our results apply to both PE and NF-PE.} for 2-player EFGPRs.
More recently,  Miltersen and S{\o}rensen have used the sequence
form to give related Lemke-like algorithms for computing
both a SE \cite{MiltSor06}  
and a QPE \cite{MiltSoren10}
for 2-player EFGPRs.
As pointed out by Miltersen and
S{\o}rensen in \cite{MiltSoren10}, van Damme's 
existence proof for a QPE in any EFGPR, given
in \cite{vanDamme84},
is somewhat roundabout: it uses the existence
of a {\em proper} equilibrium 
in a NFG (\cite{Myerson78}), and it uses
a relationship established in \cite{vanDamme84}  
between proper equilibrium in NFGs
and QPEs of any EFGPR that has that NFG as its
standard normal form.  Miltersen and
S{\o}rensen  state in \cite{MiltSoren10} 
that ``{As far as we know, no very simple and direct
proof of existence [of QPE] is known.}''
They note that their results furnish a different
proof of existence of QPE for 2-player EFGPRs.
One of the consequences of our results is a 
simple and direct proof, via application of Brouwer's fixed point 
theorem (and Bolzano-Weierstrass), of the existence of a QPE in any $n$-player EFGPR. In a similar
way, our results furnish a direct existence proof for all 
the notions of equilibrium for EFGPRs that we study.

More closely related to our complexity results for $n$-player EFGPRs,
with $n \geq 3$, 
von Stengel  
in \cite{vonStengel-efficient-1996}
used the sequence form of EFGPRs to describe 
an interesting
nonlinear program, associated with a given
$n$-player EFGPR,  such that the optimal solutions to  
the nonlinear program are the NEs of the EFGPR,
where the encoding size of the nonlinear program is polynomial
in the size of the EFGPR.
One can use von Stengel's nonlinear programming formulation,
together with results on decision procedures for the 
theory of reals 
\cite{Ren92,BasuPollackRoy2006},
to show that approximating an NE for a given $n$-player EFGPR, 
to within given $\ell_\infty$-distance $\delta > 0$, is
in PSPACE. 

Even more closely related to our results is a result 
by Daskalakis, Fabrikant, and Papadimitriou in \cite{DFP06}.
Specifically, Theorem 4 of \cite{DFP06} states
that the problem of computing a [$\epsilon$-]Nash equilibrium 
and a [$\epsilon$-almost] subgame-perfect
equilibrium, for an extensive form game [of perfect recall]
is polynomial
time reducible to computing a [$\epsilon$-]Nash equilibrium for a 2-player
normal form game.   
The statement of Theorem 4 in
\cite{DFP06} does not make a distinction between computing
an actual Nash equilibrium (within desired precision $\epsilon > 0$), 
versus computing an $\epsilon$-NE.
Indeed, \cite{DFP06} appeared prior to the publication of
the paper \cite{EY07} where the distinction between
the complexity of these two problems was highlighted,
and where the complexity class $\FIXP$ and $\FIXPA$ were defined.
The proof of Theorem 4 in \cite{DFP06} can be used (\cite{Das14})  
to establish a reduction from the problem of computing an exact
Nash or subgame perfect
equilibrium (within given desired precision $\epsilon > 0$) 
for a given $n$-player EFGPR, to the problem of
computing a Nash equilibrium (within given desired precision $\epsilon > 0$) 
for a 3-player normal form game.
In \cite{DFP06} a brief proof sketch
for Theorem 4 is provided, which  builds on the earlier
PPAD-completeness results in
\cite{DasGP09,Chen-Deng06} and goes via reductions to graphical
games.  However, the sketched
proof provided in \cite{DFP06} contains an error
(\cite{Das14}): it assumes that 
any behavior strategy profile (even when not fully mixed) necessarily
defines a distribution on the nodes of every information set,
but this need not be the case, in particular because some information
sets may be reached with probability $0$.  Thus, the
distributions on information sets described in 
the proof sketch  in \cite{DFP06} are in general ill-defined.
The authors of \cite{DFP06} have communicated
(\cite{Das14})
a fix for this error 
to the author of this paper.
The fix involves defining 
the probability distribution on a given information set
using the most recent common single-node ancestor
of all vertices in that information set.
The authors of \cite{DFP06} will make their fixed proof 
available in some future expanded version of \cite{DFP06}.
We will not elaborate further on their fix, since
our results make
no use of any of the results in \cite{DFP06}.
In particular, we make no use of graphical games.   Instead we
directly provide algebraically-defined functions whose
fixed points give $\epsilon$-perfect equilibria of the given EFGPR.
Our results imply
essentially the same
complexity results for computing Nash and subgame-perfect
equilibrium as those implied by Theorem 4 of \cite{DFP06},
as well as for computing various other important refinements of 
equilibrium.\footnote{Although
it is worth pointing out that, by contrast, our results {\em do not} imply
that computing an exact PE, SE, or QPE, is in $\FIXP$,
only that computing a ($\delta$-almost) 
approximation of these is in $\FIXPA$ (and
$\PPAD$ respectively).}

\vspace*{0.07in}

\noindent {\bf Our results.} 
We consider the complexity of various equilibrium computation problems 
for which an 
input instance consists of 
$\langle \calG, \delta \rangle$, 
where $\calG$ is an $n$-player EFGPR 
(for any $n$: $n$ can be part of the input), and  where the rational 
``error'' parameter $\delta > 0$ is given
in binary representation.   Our main results are the following:

\begin{enumerate}

\item 
Computing a behavior (strategy) profile, $b$, 
such that 
there exists a PE  (or NE, or SGPE)\\  $b^*$ of  $\calG$, with 
$\|b - b^*\|_\infty < \delta$, is  $\FIXPA$-complete.  
(Theorem \ref{thm:main-fixpa-result}, Part 1.)

\item 
Computing a behavior profile, $b$, 
such that 
there exists a QPE (or NF-PE),  $b^*$ of  $\calG$, with 
$\|b - b^*\|_\infty < \delta$, is  $\FIXPA$-complete.  
(Theorem \ref{thm:main-fixpa-result}, Part 2.)

\item Computing an
{\em assessment},
$(b,\mu)$, 
such that there exist an  SE, $(b^*,\mu^*)$ of $\calG$,  with\\
$\| (b,\mu) - (b^*, \mu^*) \|_\infty
< \delta$,  is $\FIXPA$-complete.
(Theorem \ref{thm:main-fixpa-result}, Part 3.)\\
An assessment $(b,\mu)$ consists of both a behavior profile $b$, as well as an
associated\\ {\em system of beliefs},  $\mu$. (We shall define all this
formally later.)

\item Given, additionally,  $\epsilon > 0$  (in binary representation) 
as input, computing a\\ {\em $\delta$-almost $\epsilon$-perfect equilibrium}
($\delta$-almost-$\epsilon$-PE) of $\calG$  is $\PPAD$-complete.
(Theorem \ref{thm:almost-ppad-complete}, Part 1.)\\
A $\delta$-almost-$\epsilon$-PE is a ``$\delta$-almost'' relaxation of 
Myerson's notion  of  $\epsilon$-PE (\cite{Myerson78}) applied
to EFGPRs, 
 $\delta$-almost-$\epsilon$-PE suitably ``refines'' 
{\em $\delta$-almost subgame-perfect equilibrium} ($\delta$-almost-SGPE).
A $\delta$-almost-SGPE of $\calG$ is a behavior profile, $b$, where 
no player can improve its own payoff
{\em in any subgame of} $\calG$ by more than
$\delta$, by unilaterally
changing its strategy in that subgame.
Thus, as a consequence we get (cf. \cite{DFP06}) that
computing a $\delta$-almost-NE  and  $\delta$-almost-SGPE is
PPAD-complete (Theorem \ref{thm:almost-ppad-complete}, Part 3.)

\item Likewise, we define a notion of $\delta$-almost-$\epsilon$-QPE, which is
a relaxation of the notion of $\epsilon$-QPE, defined by 
van Damme in \cite{vanDamme84}, and
we show that computing a $\delta$-almost-$\epsilon$-QPE
of $\calG$, given $\calG$, and given $\delta > 0$ and $\epsilon > 0$,  is $\PPAD$-complete.
(Theorem \ref{thm:almost-ppad-complete}, Part 2.)
\end{enumerate}

\noindent In all the above results, the ``hardness'' result 
follows immediately (already for 3-player games) 
from the prior known hardness results for NFGs 
(\cite{EY07,Chen-Deng06,DasGP09}).
The new results are the upper bounds.

\vspace*{0.07in}

\noindent {\bf Outline of proof ideas.}
By contrast to the prior work on algorithms for 2-player EFGPRs, our results make no
explicit use
of the {\em sequence form} for EFGPRs.
Also, by contrast to \cite{DFP06} we make no use of reductions to graphical
games. 
Instead, we combine older
insights, including Kuhn and Selten's original {\em agent normal form} for EFGPRs,
and Myerson's alternative definition of PE
using $\epsilon$-PEs (both for normal and extensive form), 
with recently developed fixed point functions 
for equilibria of $n$-player normal form games, $n \geq 3$, 
developed in \cite{EY07} and \cite{EHMS14}.

More specifically, a key to our results is this:
in Section \ref{sec:epsilon-PE-char},
we adapt a construction 
in \cite{EHMS14} of a fixed point function 
for ``$\epsilon$-PEs'' of a given NFG (which itself is an
adaptation of a fixed point function for NEs of NFGs given in \cite{EY07})
to show that to any $n$-player EFGPR, $\calG$,  
we can associate a 
continuous function $F^\epsilon_\calG(x)$,
defined 
by a ``small'' algebraic circuit over $\{+,*, \max\}$ (whose encoding
size is polynomial in that of $\calG$), 
where $\epsilon$ in an input parameter to the circuit, 
and such that, for any fixed $\epsilon > 0$, 
the function $F^\epsilon_\calG(x)$ maps the space of behavior
strategy profiles of $\calG$ to itself, such
that the Brouwer fixed points of  $F^\epsilon_\calG(x)$
constitute $\epsilon$-PEs of $\calG$.    This proves that computing an $\epsilon$-PE,
given $\langle \calG, \epsilon \rangle$, is in $\FIXP$, {\em even when $\epsilon > 0$ 
is given  succinctly by an algebraic circuit}.

Also, we similarly define another continuous function,  $H^\epsilon_{\calG}(x)$ using a ``small'' algebraic circuit,
such that, for any fixed $\epsilon > 0$  the function $H^{\epsilon}_\calG(x)$ maps the space of behavior
profiles to itself, and such that every fixed point of $H^{\epsilon}_{\calG}(x)$ is a $\epsilon$-QPE.

The reason why we can construct the functions
$F^{\epsilon}_\calG(x)$ and $H^{\epsilon}_{\calG}(x)$ with a  ``small'' (poly-sized) algebraic circuit 
is related to properties of the {\em agent normal form}
of EFGPRs, and to the fact that the ``realization
probabilities'' and the expected payoff functions for
EFGPRs can be expressed as ``small'' (multilinear) polynomials.
In particular, a simple but important fact is that an EFGPR has exactly the same  ($\epsilon$-)PEs as
its agent normal form.  (It {\em does not} necessarily have the same NEs.)
Even though we can not construct the agent normal form explicitly (because it is
exponentially large),  it turns out that we do not need to: by combining these 
various facts, we can nevertheless 
construct a ``small'' algebraic circuit for $F^{\epsilon}_\calG(x)$, 
by adapting the analogous construction from \cite{EHMS14}.

With the functions $F^{\epsilon}_\calG(x)$ (and $H^\epsilon_{\calG}(x)$) in hand, 
in Section \ref{approx-of-PE} we then use
(similar to \cite{EHMS14}) 
algebraic circuits 
to construct
a ``very very small''  $\epsilon^* > 0$  for which we can prove,
using results from real algebraic geometry (\cite{Ren92,BasuPollackRoy2006}),
that every fixed point of
$F^{\epsilon^*}_\calG(x)$ is $\delta$-close  (in $\ell_\infty$) to an actual
PE.  Likewise, we show that every fixed point of $H^{\epsilon^*}_\calG(x)$ is
$\delta$-close to a QPE.
This allows us to show containment in $\FIXPA$ for approximating a PE,
and for approximationg a QPE.
We furthermore show how to extend the function $F^{\epsilon}_\calG(x)$ to define 
another ``small'' algebraic function
$G^{\epsilon}_\calG(x,z)$ that serves the same purpose for {\em sequential 
equilibrium} (SE), and in particular that additionally includes a corresponding 
{\em system of beliefs} inside its 
fixed points.  This shows containment in $\FIXPA$ for approximating an SE.

Finally, in Section \ref{sec:delta-almost-epsilon-PE},
we observe some properties of the
functions $F^{\epsilon}_\calG(x)$  (they are ``polynomially continuous''
and ``polynomially computable''), which 
when combined with results in \cite{EY07}
imply that
computing a ``$\delta$-almost fixed point''
of
$F^{\epsilon}_\calG(x)$, given $\calG$
and given $\delta > 0$ and $\epsilon > 0$,
is in $\PPAD$.
We then show that a ``$\delta$-almost fixed point''
of $F^{\epsilon}_\calG(x)$ is a $(3 \delta)$-almost-$(\delta+\epsilon)$-PE 
of $\calG$.  
We also show that a ``$\delta$-almost fixed point'' of
$H^\epsilon_{\calG}(x)$ is a  $(3 \delta)$-almost-$(\delta+\epsilon)$-QPE.
Lastly, we show that a $\delta'$-almost-$\epsilon'$-PE,
for ``polynomially small''  $\delta'$ and $\epsilon'$,
is  a $\delta$-almost-SGPE of $\calG$.
These results allow us to show containment in $\PPAD$ 
for the ``$\delta$-almost'' equilibrium notions that we study.

This last part, for establishing $\PPAD$-completeness for ``$\delta$-almost'' equilibria,
is technically one of the more involved
parts of our proofs. 
Also, our proof of $\FIXPA$-completeness 
for computing a QPE involves a novel fixed point characterization.
By comparison to these, our proof of $\FIXPA$-completeness 
for PE is technically
easier, {\em given} the prior results in \cite{EHMS14,EY07},
and given long existing results in the literature on EFGPRs which we exploit.

\noindent {\bf Potential computational applications.}
We believe our results could potentially 
provide a ``reasonably practical'' method for computing 
 $\delta$-almost relaxations of equilibrium refinments 
for $n$-player EFGPRs, 
including $\epsilon$-perfect
and $\epsilon$-quasi-perfect equilibrium as 
well as less refined notions
of  $\delta$-almost equilibrium like SPGE and Nash 
(for which see also the result of \cite{DFP06}),
by applying classic
discrete path following algorithms
for ``almost'' fixed point computation,  such as variants of 
Scarf's algorithm \cite{Scarf67,Scarf73}, 
on the ``small'' algebraic
fixed point functions we associate with EFGPRs.
We believe this is a promising 
approach for 
``almost equilibrium'' computation for EFGPRs
that should be implemented and
explored experimentally.
We note that the well-known software package GAMBIT (\cite{GAMBIT-ref14}),
which provides a variety of state-of-the-art algorithms for solving various classes of games,
does not currently provide any algorithm for computing or approximating
an equilibrium (of any kind) for 
a general $n$-player EFGPR,  for any  $n \geq 3$.  
Indeed, a survey on
equilibrium computation from 1996 
(\cite{MckelveyMcLennan-survey1996}),
by McKelvey and McLennan
who helped to develop GAMBIT,
discusses the algorithms by Koller et. al. 
(\cite{koller1992complexity,vonStengel-efficient-1996,KolMegvonSteng:1996})
for 2-player EFGPRs,
but does not discuss any general algorithms for $n$-player EFGPRs, beyond 
first converting
to (reduced) normal form, 
and using heuristics like iterated elimination of dominated strategies.  
We believe our results 
can potentially be used to remedy this gap 
in the availability of ``practical'' software for
(refined) equilibrium computation for $n$-player EFGPRs.

\section{Definitions and Background}

\label{sec:background}

{\em Dear Reader:}  EFGPRs, and refinements of equilibrium for them, 
are treated in nearly every modern textbook on game theory (see, e.g.,
\cite{MasSolZam13,OsRub94,Myerson-book-1997,vanDamme91}).
Nevertheless,  for us to discuss our problems rigorously,
we can not just point you to a book or paper with relevant definitions.
We must fix
(a considerable amount of) notation and terminology,
and we must describe various essential background results.
This is especially because we will be addressing  
various subtle refinements of equilibrium,
and corresponding notions (in some cases, new) of ``approximate'' 
and ``almost'' equilibrium, where slight differences in 
definitions can  
have major consequences, particularly
for computational complexity.   
We also have to define the relevant
complexity classes like $\FIXP$, $\FIXPA$, and $\PPAD$.
So, we proceed to carefully fix notation
and definitions, and to describe the needed background results.
Readers familiar with EFGPRs, or with other parts of the background,
can skip ahead to subsequent sections 
that contain the new results, and 
return to this section as needed, using it as a ``reference''.
(Although some things are likely to become harder to follow that way.)\\

For a finite set $X$,
we let $\Delta(X)$ denote the set of probability distributions on $X$, i.e.,
the set of functions $f : X \rightarrow
[0,1]$ such that $\sum_{x \in X} f(x) = 1$.
For $f \in \Delta(X)$, we let $\support(f) = \{ x \in X \mid f(x) > 0 \}$ denote its {\em support set}.
For a positive integer $k$, we let $[k] = \{1,\ldots,k\}$.

\noindent {\bf Extensive Form Games.} 
Intuitively, a finite game tree is just a rooted, labeled, finite tree.
We will find it convenient to view 
such a tree 
as a finite, prefix-closed, set of strings over a finite 
alphabet of ``actions''.
Formally,
let $\Sigma$ be a finite set called the {\em action alphabet}.
We shall use the symbols $a, a', a_1, a_2, \ldots$, to denote 
letters in the alphabet $\Sigma$. 
For a string $u \in \Sigma^*$, we use $|u|$ to denote the length of $u$.
A {\em tree}, $T = (V,E)$ over action alphabet $\Sigma$, consists
of a finite set $V \subseteq \Sigma^*$ of  
{\em nodes} (or {\em vertices}),
where 
furthermore $V$ is prefix-closed, 
meaning that if $w \in V$ and $w = u a$, where $a \in \Sigma$, then $u \in V$.
Note that by definition the empty string $\epsilon$ is in $V$.  We refer to $\epsilon$ as 
the {\em root} of the tree.
The directed edge relation $E \subseteq V \times V$, of the tree $T$
(which points ``away from'' the root) is
defined by:  $E = \{ (u,w) \in V \times V \mid  \exists a \in \Sigma:  w = ua \}$.
For two nodes $u, w \in V$,   if $(u,w) \in E$, we say that $w$ is a {\em child} of
$u$, and that $u$ is the (unique) {\em parent} of $w$.  
For $u \in V$, we let
$\children(u) = 
\{  w \in V \mid  (u,w) \in E \}$
denote the set of children of $u$.
Let $\sqsubseteq$ denote the reflexive transitive
closure of $E$.  Thus, $u \sqsubseteq w$ is just the prefix relation 
on the set $V$.  
We use $u \sqsubset w$ to denote the strict prefix relation: $(u \sqsubseteq w  \wedge  u \neq w)$. 
When $u \sqsubset w$, we say that $u$ is a {\em ancestor} of $w$,
and that $w$ is a {\em descendant} of $u$.
For each node $u \in V$, we define $\Act(u) = \{ a \in \Sigma \mid ua \in V \}$
to be the set of actions {\em available} at node $u$.
A {\em leaf} is a node $u \in V$ with no children, i.e., 
where $\children(u) = \emptyset$.
Let $\Leaves = \{ u \in V \mid  \children(u) = \emptyset \}$ denote
the set of leaves of the tree $T$.
A non-leaf node is called an {\em internal node}; 
let $\Internal = V \setminus \Leaves$ denote the set of internal
nodes.
A {\em path} $\psi$ in the tree $T$ is a non-empty sequence 
$\psi = u_0, u_1, u_2, \ldots, u_m$ of nodes, 
where
for all $0 \leq i < m$, 
$(u_i, u_{i+1}) \in E$.
The path $\psi$ is called a  {\em play}
if $u_0 = \epsilon$, and it is called a {\em complete play} 
if additionally  $u_m$ is a leaf.
In other words, a (complete) play is just a path that starts at the root 
(and ends at
a leaf).  Note that a node $u \in V$ is a string in $\Sigma^*$
that encodes all the
information needed to reconstruct the unique path in $T$ from the root to $u$.

A {\em Finite Game in Extensive Form} (EFG), 
$\mathcal{G} = (N, \Sigma, T, P, I, p, r)$,   
is a tuple consisting of:

\begin{enumerate}
\item \underline{\em Players}: A set $N = [n] = \{ 1, \ldots, n \}$ of {\em players}.

\item \underline{\em Action alphabet}: 
a finite set $\Sigma$,
called the {\em action alphabet}.
Let $k_\calG = |\Sigma|$ denote the 
size of $\Sigma$.

\item \underline{\em Game Tree}: A finite tree $T = (V,E)$
over the action alphabet $\Sigma$,
called the {\em game tree}.

\item \underline{\em Player partition}: 
A partition $P = (P_0, P_1, \ldots, P_n)$ of the set $\Internal$
of internal nodes, i.e., 
$P_i \subseteq \Internal$, $\bigcup^n_{k=0} P_k = \Internal$, and 
$P_i \cap P_j = \emptyset$, for  all $i \neq j$,  $i,j \in \{0,\ldots,n\}$.

For $i=1,\ldots,n$, 
the nodes in $P_i$ are the internal nodes ``belonging'' to player $i$:
these are the nodes where player $i$ has to choose the next move.
The set $P_0$ consists of the internal nodes belonging 
to {\em chance} (or {\em nature}).
The next move at a node $u \in P_0$ is chosen randomly,
according to a provided distribution, $p_u$, given in item (6.) below.\\
We define the {\em player map}, $\Pl: V \rightarrow \nat$, by:
for all $i \in \{0,\ldots, n\}$ and $u \in P_i$, $\Pl(u) := i$. 

\item  \underline{\em Information set partition}:
A tuple $I = (I_1,\ldots, I_n)$, such that for each $i \in [n] = \{1,\ldots,n\}$,
$I_i = ( I_{i,1}, \ldots, I_{i,d_i} )$ is a partition of the set 
$P_i$ of vertices belonging to player $i$,  
where each {\em information set}  $I_{i,j} \subseteq P_i$ is non-empty \& 
$\bigcup^{d_i}_{j=1} I_{i,j} = P_i$,  $I_{i,j} \cap I_{i,k} = \emptyset$
for all $j \neq k$, $j,k \in [d_i]$.

It is furthermore assumed that, for every information set $I_{i,j}$,
and for any two nodes $u,v \in I_{i,j}$, $\Act(u) = \Act(v)$.  In
other words, 
the same set of actions is available to player $i$ 
at every node in $I_{i,j}$.
Let $\mathcal{A}_{i,j} := \Act(u)$, where $u \in I_{i,j}$.
By assumption, $\mathcal{A}_{i,j}$ is well-defined.

We define the map $\Info(\cdot)$, 
which maps a node $u$
to the index of the information set to which $u$ belongs.
Thus, if 
$u \in I_{i,j}$, 
then $\Info(u) := j$.  For convenience,
we extend the map $\Info(\cdot)$ to chance nodes $u \in P_0$ as follows: 
for all $u \in P_0$, we define $\Info(u) := u$. 

The extensive form game, $\calG$, is said to have {\em perfect information} if
all information sets $I_{i,j}$ are singleton sets, for 
all $i \in [n]$, $j \in [d_i]$.  
Otherwise, it is called
a game of {\em imperfect information}.  

\item \underline{\em Probability distributions for chance nodes}:
A tuple of probability distributions $p = (p_u)_{u \in P_0}$, 
one for each chance node $u \in P_0$,
 where $p_u :  \Act(u)  \rightarrow (0,1] \cap \rat$ is a positive,
rational\footnote{We restrict the distributions $p_u$
to have rational probabilities for computational
purposes.}, probability distribution
on actions available at $u$.  So,
$p_u(a) >  0$ and $p_u(a) \in \rat$ for all $a \in \Act(u)$, 
and $\sum_{a \in \Act(u)} p_u(a) = 1$.
Let $p^\calG_{0,\min} := \min_{u \in P_0, a \in \Act(u)} 
p_u(a)$.

\item \underline{\em Payoff functions}: An $n$-tuple $r = (r_1,\ldots, r_n)$ of payoff functions.
For each player $i$, the payoff function $r_i: \Leaves \mapsto \nat_{> 0}$,
maps each leaf $u \in \Leaves$ of the tree $T$ to a positive integer payoff 
for player $i$.\footnote{We restrict to positive integer payoffs,
rather than real payoffs, for
computational purposes.  One can of course
also consider rational payoff functions $r_i: \Leaves \mapsto \rat$.
However, as is well-known, restricting to positive integer payoffs is
w.l.o.g. for computational purposes: we can always
 ``clear denominators'' 
by multiplying by their LCM, and then add a large enough
positive value to the resulting integers to get positive payoffs.  
This does not increase by much the encoding size of $\calG$, 
and the resulting game can be shown to be ``suitably isomorphic'' to the original for all our purposes,
including equilibrium approximation within desired precision, and $\delta$-almost equilibrium computation. }
Let $M_{\calG} := \max_{i \in [n], u \in \Leaves} r_i(u)$ denote the largest possible 
(positive integer) payoff.
\end{enumerate}

We denote 
the bit encoding size of an EFG, $\calG$, by $|\calG|$, where
we assume binary encoding for
the integer payoff values at the leaves of $\calG$, as well as the rational probabilities of actions
at chance nodes (with numerator and denominator given in binary).\footnote{We assume natural representations for the various pieces of $\calG$,
 including the tree $T$, player partition, information partition, payoff functions, and the probability distributions at chance nodes
(with rational probabilities encoded in binary) .
The details of the natural encoding
are irrelevant for our purposes, so we do not spell them out.}  
For a rational number $q \in \rat$, we use $\size(q)$ 
to denote its bit encoding size.
Similarly, for a rational vector ${\tt v} \in \rat^m$,
we use $\size({\tt v}) := \sum^m_{i=1} \size({\tt v}_i)$ to denote its encoding size.

For a game $\calG$ with tree $T=(V,E)$, 
let $\height^\calG := \max \{ |u| \mid v \in V \}$ denote the {\em height} of
$T$.
For $u \in V$,  we 
define the {\em subtree} rooted at $u$,   $T_u = (V_u,E_u,u)$,
by:  $V_u = \{ w \in V \mid  u \sqsubseteq w \}$,  and 
$E_u =  \{ (u,w) \in E \mid  u,w \in V_u \}$.
We let $\height^\calG_u := \max \{ |w| - |u| \mid  w \in V_u \}$ 
denote the {\em height} of $T_u$.  (Note that
$\height^\calG = \height^\calG_\epsilon$.)
Consider an EFG, $\mathcal{G} = (N, \Sigma, T, P, I, p, r)$. 
For a node $u$ of the game tree $T$,  if the subtree $T_u$
satisfies the property that for every node $w \in V_u$,
the information set $I_{\Pl(w),\Info(w)}$ is a subset of $V_u$,
then the subtree $T_u$ naturally defines a 
{\em subgame},  $\mathcal{G}_u = (N',\Sigma,T',P',I',p',r')$,  which 
is rooted at the node $u$ instead of at $\epsilon$,
and where the player partition, information set partition,
payoff functions, and probability function for
chance nodes, are all inherited directly from $\mathcal{G}$ by 
restricting them to the subtree $T_u$ in the obvious way.

Note that a node $u \in V$ is a string in $\Sigma^*$ which also encodes 
the unique {\em history}
of actions, starting at the root, which lead to that node in $T$.
For any node $u \in V$, with $|u| = k$,  $u = a_{1} a_{2} \ldots a_{k}$, and 
for any $m \in \{0,1,\ldots,k\}$, let 
$u[m] = a_{1} \ldots a_{m}$ denote the length $m$ prefix of $u$.  
For a node $u$, with $|u| = k$, we define the 
{\em information-action history at $u$}, denoted $Y(u)$, to be the 
following sequence of $k$ triples:
$$Y(u) =  \langle \  ( \Pl(u[m]) , \Info(u[m]) , a_{m+1} ) \  \mid \  
m = 0,\ldots k-1   \ \rangle $$ 

\noindent For each player $i \in [n]$, we 
define the {\em visible history for player $i$ at $u$},
denoted $Y_i(u)$, to be the subsequence of $Y(u)$
obtained by retaining only those triples $(i',j',a')$ in the sequence 
$Y(u)$ for which $i' = i$, and deleting all other triples.
In other words, $Y_i(u)$ records the sequence of information
sets belonging to player $i$ encountered 
along the path from the root $\epsilon$ to $u$ (not including $u$), and 
the actions player $i$ chose at each of those 
information sets, prior to reaching $u$.

An EFG, $\mathcal{G}$, is said to have
{\bf\em perfect recall} if the following condition holds:
for any two nodes $u, v \in V$,  if $\Pl(u) = \Pl(v) = i \in [n]$
and $\Info(u) = \Info(v)$,  then $Y_i(u) = Y_i(v)$.
In other words, during play, players remember their own prior sequence of actions
as well as the information sets they were in when they took those
prior actions.
So, it can not be the case that two nodes $u$ and $v$ are in
the same information set for some player $i$, and yet 
the visible history for player $i$ at $u$ is different from
the visible history for player $i$ at $v$. 
Note that perfect recall implies
there do not exist nodes $u \neq v$ belonging to the same information
set such that $u$ is an ancestor of $v$.  Otherwise, since $Y_i(u)$ is a strict
prefix of $Y_i(v)$, we would have $Y_i(u) \neq Y_i(v)$, 
violating perfect recall.
For a game $\mathcal{G}$ of perfect recall, let us define the visible history
associated with an information set $I_{i,j}$ as follow:  
Let $Y_{i,j} := Y_i(u)$, where $u \in I_{i,j}$.
Note that by perfect recall $Y_{i,j}$ is well-defined.

\vspace*{0.05in}

\noindent {\bf\em Assumption:} {\em Throughout this paper, extensive form games are assumed to have perfect recall.}

\vspace*{0.05in}

As mentioned,
this assumption is standard practice in much of the literature on extensive
form games.  As mentioned, we use EFGPR to refer to an EFG with perfect recall.

\vspace*{0.05in}

\noindent {\bf Strategies.}
For an extensive form game, $\mathcal{G}$, where the
information sets for player $i$ are indexed by the set $[d_i] = \{1,\ldots,d_i\}$,
a {\em pure strategy}, $s_i$, for player $i \in [n]$, 
is a function 
$s_i: [d_i] \rightarrow \Sigma$ that assigns an available action to each
information set belonging to player $i$, so for all $j \in [d_i]$,
$s_i(j) \in {\mathcal A}_{i,j}$.
In other words, when using pure strategy $s_i$,
player $i$ chooses the available action $s_i(j)$ 
at every node in the information set
$I_{i,j}$.
Let $S_i$ denote the set of pure strategies for player $i$.
Let $S = S_1 \times S_2 \times \ldots \times S_n$ denote the
set of {\em profiles} of pure strategies.

A {\em mixed strategy} for player $i$, $\sigma_i \in \Delta(S_i)$, 
is a probability distribution on pure strategies $S_i$ (note: for a finite game
$\mathcal{G}$, $S_i$ is a finite set). 
For a pure strategy $c \in S_i$, we shall use $\pi^c_{i}$ to denote 
this pure strategy  as
an element of $\Delta(S_i)$; so
$\pi^c_{i}(c) =1$, and $\pi^c_{i}$ assigns probability $0$ to other
pure strategies.
We let $M_i = \Delta(S_i)$ denote the set of mixed strategies for player $i$.
Let $M = M_1 \times M_2 \times \ldots \times M_n$ denote the set of {\em profiles} of mixed strategies.
Let $M^{> 0}$ denote the set of {\em fully mixed} profiles of mixed strategies, that is,
$M^{> 0} := \{ \sigma = (\sigma_1,\ldots, \sigma_n) \in M \mid   \sigma_i(c) > 0, 
\ \text{for all}  \ i \in [n]  \text{and} \  c \in S_i \}$.

A {\em behavior strategy}, $b_i$, for player $i$,  is 
a $d_i$-tuple $b_i = (b_{i,1}, b_{i,2}, \ldots, b_{i,d_i})$ 
of probability distributions, such that
for each $j \in [d_i]$,
$b_{i,j} \in \Delta({\mathcal A}_{i,j})$ is a probability distribution 
on the set of actions ${\mathcal A}_{i,j}$ available in information set 
$I_{i,j}$.  
In other words, for all $a \in {\mathcal A}_{i,j}$,  
$0 \leq b_{i,j}(a) \leq 1$,
and $(\sum_{a \in {\mathcal A}_{i,j}} b_{i,j}(a)) = 1$.
We shall find it convenient to sometimes write $b_{i,j,a}$ instead
of $b_{i,j}(a)$, and to view $b_{i,j}$ as a vector of probabilities,
$b_{i,j} =
(b_{i,j,a})_{a \in \calA_{i,j}}$.
Let $B_{i,j} := \Delta({\mathcal A}_{i,j})$.
We call $b_{i,j} \in B_{i,j}$  a {\em local strategy}
at information set $I_{i,j}$. 
For an action $a \in {\mathcal A}_{i,j}$, we shall use $\pi^a_{i,j}$ to denote
the {\em pure} local strategy in $B_{i,j}$, that assigns probability 
$1$ to the action $a$.
Let $B_i = B_{i,1} \times \ldots \times B_{i,d_i}$ 
denote the set of behavior strategies for player $i$.
Let $B = B_1 \times B_2 \times \ldots \times B_n$ denote
the set of profiles of behavior strategies.
Let $B^{> 0}$ denote the set of {\em fully mixed} behavior profiles, that is
$B^{> 0} := \{ b = (b_1, \ldots, b_n) \in B \mid   b_{i,j}(a) > 0 , 
\ \text{for all}  \ i \in [n] ,  \  j \in [d_i], \ \text{and}
\ a \in \calA_{i,j} \}$.

For a behavior strategy $b_i = (b_{i,1}, \ldots, b_{i,d_i}) \in  B_i$, 
for  $j \in [d_i]$ and a local strategy $b'_{i,j} \in B_{i,j}$,
we use $(b_{i} \mid b'_{i,j})$ to denote the
revised behavior strategy $(b_{i,1}, \ldots,b_{i,j-1},b'_{i,j}, b_{i,j+1}, \ldots,b_{i,j})$.
In other words, $(b_{i} \mid b'_{i,j}) \in B_i$ consists
of the same local strategies as $b_i$, except at information set 
$I_{i,j}$
the local strategy is switched from $b_{i,j}$ to $b'_{i,j}$.
Likewise,  for a behavior profile $b \in B$,   and a behavior strategy
$b'_{i} \in B_{i}$,  we let $(b \mid b'_i) = 
(b_{1}, \ldots,b_{i-1},b'_{i},b_{i+1}, \ldots,b_{n})$.
In other words, $(b \mid b'_i) \in B$ consists of the same behavior
strategies as $b$, except for player $i$ the behavior strategy
is switched form $b_i$ to $b'_i$.
Lastly,  for a behavior profile $b = (b_1, \ldots, b_n) \in B$  and a 
local strategy $b'_{i,j} \in B_{i,j}$,  we define the shorthand notation
$(b \mid b'_{i,j}) :=  (b \mid (b_i \mid b'_{i,j}))$.

We  also define a more
general set of strategies, generalizing both $B_i$ and $M_i$, called
{\em mixed-behavior strategies}, $MB_i$. 
A mixed-behavior strategy $\sigma_i \in MB_i$ 
is a probability distribution over a finite subset
of behavior 
strategies in $B_i$.
Clearly, $S_i \subseteq B_i \subseteq MB_i$ and 
$S_i \subseteq M_i \subseteq MB_i$.
We let $MB = MB_1 \times \ldots \times MB_n$ denote the set of profiles of 
mixed-behavior strategies.

Once we fix a strategy
profile, $\sigma = (\sigma_1, \ldots, \sigma_n) \in MB$
for the players,  this determines a {\em realization probability} function,
$\Prob_\sigma(u)$, that assigns to every node $u \in V$ 
the probability
of reaching $u$ starting from the root, when players 
use their respective strategies in the profile $\sigma$.
Then the {\em expected payoff}, $U_i(\sigma)$, to player $i$ under
the strategy profile $\sigma$ is:

\begin{equation}
\label{eq:exp-payoff-formula}
U_i(\sigma) =  \sum_{z \in \Leaves} \Prob_\sigma(z) \cdot r_i(z)
\end{equation}

For any profile $\sigma$, and a strategy $\sigma'_i$ for player $i$,
we use $(\sigma \mid \sigma'_i)$ to denote the revised profile
$(\sigma_1, \ldots,\sigma_{i-1},\sigma'_{i}, \sigma_{i+1}, \ldots, \sigma_n)$,
where everyone's strategy remains the same, except player $i$'s 
strategy switches
to $\sigma'_i$.
We call two strategies $\sigma'_i$ and $\sigma''_i$ for player $i$ 
{\em realization equivalent},  denoted by $\sigma'_i \approx \sigma''_i$,
 if for all $u \in V$  and for all strategy profiles $\sigma \in MB$,
$\Prob_{(\sigma \mid \sigma'_i)}(u) = \Prob_{(\sigma \mid \sigma''_i)}(u)$.
Note that if $\sigma'_i \approx \sigma''_i$, then 
$U_i(\sigma \mid \sigma'_{i}) =  U_i(\sigma \mid \sigma''_i)$ for all 
$\sigma \in MB$.
For games
of perfect recall, we have:

\begin{proposition}[\cite{Kuhn53},
 \cite{Selten75}]
\label{prop:kuhn-mix-behave-lemma}
For every EFGPR, $\mathcal{G}$,
every mixed-behavior
strategy $\sigma_i \in MB_i$ is realization equivalent to a
behavior strategy $b_i \in B_i$, i.e.,  such that 
$\sigma_i \approx b_i$.
\end{proposition}

Thus, w.l.o.g., we can confine our attention to behavior strategies in $B_i$
for all EFGPRs.

Note that also for every behavior strategy $b_i \in B_i$ there exists a 
realization equivalent mixed strategy, $\sigma^{b_i}_i \in M_i$.
Here's how.
Define $\chi(x,y)$ by:
$\chi(x,y) := 1$ if $x=y$, and otherwise $\chi(x,y) := 0$.
We define the mixed strategy $\sigma^{b_i}_i$ as follows. For every $c \in S_i$:
$$\sigma^{b_i}_i(c) :=  \prod_{\{ \; (j , a) \; \mid \; j \in [d_i] \; \& \;  
a  \in \calA_{i,j} \; \} }  \chi(c(j),a) \cdot  b_{i,j}(a).$$
The mixed strategy $\sigma^{b_i}_i$ is realization 
equivalent to behavior strategy $b_i$.\footnote{Of course,  in general, the support size of $\sigma^{b_i}_i$
can be exponential in the dimension of the vector $b_i$, so it
is not in general efficient to work explicitly with $\sigma^{b_i}_i$
instead of $b_i$.}
For a behavior profile $b \in B$,  we will use the notation
$\sigma[b] := (\sigma^{b_1}_1, \ldots, \sigma^{b_n}_n) \in M$
to denote the (realization equivalent) mixed profile {\em induced} by $b$.

For a EFGPR, $\mathcal{G}$,
for any node $u \in V$,  and 
any behavior profile $b \in B$,
we can define the realization probability 
$\Prob_b(u)$ 
as a multi-variate polynomial $F_u(x)$
(in fact, a multilinear monomial) whose ``variables'' $x$ 
correspond to the coordinates of 
a behavior strategy profile in $B$, and such that for all $b \in B$,
$F_u(b) = \Prob_b(u)$.  
Specifically, for all $u \in V$, where $|u| = k$ and 
$u = a_1 a_2 \ldots a_k$,  we associate the variable
$x_{i,j,a}$ with the probability $b_{i,j,a} = b_{i,j}(a)$ in a behavior profile $b$,
and $F_u(x)$ is given by: 
\begin{equation*}
F_u(x) \equiv
\left( \prod_{\{ m \in \{0,\ldots,k-1\} \; \mid \; u[m] \in P_0 \}} p_{u[m]}(a_{m+1}) \right)  \ 
\cdot \ 
\prod_{ \{ m \in \{0,\ldots,k-1\}  \; \mid \; u[m] \in \Internal \setminus P_0\} } 
x_{{\small \Pl(u[m])} , {\small \Info(u[m])} , a_{m+1}}
\end{equation*}

\noindent Note that, for any $u \in V$, 
the total degree of $F_u(x)$
is at most $\height^\calG$.
More generally, for a subset $V' \subseteq V$ of nodes, 
let $\Top(V') := \{ u \in V' \mid \neg \exists v \in V': \: v \sqsubset u \}$.
(Note: for any information set $I_{i,j}$, $\Top(I_{i,j}) = I_{i,j}$.)
We define the 
{\em realization probability}, $\Prob_b(V')$, of ({\em some} node in) $V' \subseteq V$,  
under (behavior) profile $b$, as follows:
$\Prob_b(V') \doteq \sum_{u \in \Top(V')} \Prob_b(u)$.  
Thus we can also define the multilinear polynomial: 
$F_{V'}(x) \equiv \sum_{u \in {\Top(V')}} F_u(x)$, such that for all $b \in B$,
$F_{V'}(b) =  \Prob_b(V')$.

Also, 
using equation (\ref{eq:exp-payoff-formula}), we have 
that the expected payoff function is given by the polynomial:
\begin{equation} 
\label{eq:poly-for-expectation}
U_i(x) \equiv \sum_{z \in \Leaves} F_z(x) \cdot r_i(z)
\end{equation} 
Thus, restating all this, we have:
\begin{proposition}
\label{prop:expected-poly}
Given a EFGPR, $\calG$,  and given any subset $V' \subseteq V$
of nodes of the game tree,
there is a multi-variate multilinear polynomial
$F_{V'}(x)$ in the vector of 
variables $x$, with total degree bounded by $\height^\calG$,
 such that 
for all $b \in B$,  $F_{V'}(b) = \Prob_b(V')$
defines
the realization probability of $V'$  
under behavior profile $b$ in $\calG$.
Moreover, 
there is a multilinear polynomial $U_i(x)$, with 
total degree bounded by $\height^\calG$,
such that
for all $b \in B$,
 $U_i(b)$
is the expected payoff of player $i$ under behavior profile $b$ in $\calG$,
and moreover, the polynomials  $F_{V'}(x)$ 
and $U_i(x)$ can be expressed
(as a weighted sum of multilinear monomials) with 
an encoding size that is polynomial in $|\calG|$.
\end{proposition}

For a fixed $b_i \in B_i$, we shall use the notation  $U_k(x \mid b_i)$
to denote the polynomial 
obtained by fixing the values of the variables $x_i$,
by assigning to them their corresponding values in $b_i$,  in the polynomial $U_k(x)$.
Likewise,  for a fixed local strategy $b_{i,j} \in B_{i,j}$, we shall use
$U_k(x \mid b_{i,j})$ to denote the polynomial obtained
by fixing the variables $x_{i,j}$ by assigning to them their corresponding values in 
$b_{i,j}$ in the polynomial $U_k(x)$.

\noindent {\bf Information Set Forest.}
We shall need the concept of the {\em information set forest}
associated with each player in a EFGPR.
Specifically, for a EFGPR, $\calG$,
for each player $i \in [n]$, we define a directed, edge-labeled, graph, $\mathcal{F}_i = (V^{\calF_i},E^{\calF_i})$,
whose nodes are $V^{\calF_i} = [d_i]$, i.e., the (indices of) information sets belonging to player $i$,
and whose $\Sigma$-labeled directed edges, 
$E^{\calF_i} \subseteq V^{\calF_i} \times \Sigma \times V^{\calF_i}$,
are defined as follows:  $(j,a,j') \in E^{\calF_i}$ if and only if
the last triple in 
the (non-empty) sequence $Y_{i,j'}$ is $(i,j,a)$.
It follows immediately from this definition that $\calF_i$ 
is a  directed (edge-labeled) forest, for all $i$.
The source nodes (roots) of the forest $\calF_i$ are those information sets 
which are the first belonging to player $i$ 
to be encountered along some complete play of the game $\calG$.
The sink nodes (leaves) of this forest are the last information set for
player $i$ encountered along some complete play.
The action $a$ labeling the edge $(j,a,j') \in E^{\calF_i}$ is the
action that player $i$ must take at information set $I_{i,j}$ in
order to enable the possibility of reaching information set $I_{i,j'}$
(but whether or not this happens with positive probability can depend on the strategies of
other players).
We henceforth refer to $\calF_i$ as the {\em information set forest} associated 
with player $i$.
We shall say that a node $j' \in V^{\calF_i}$ is a 
{\em descendant} of a node $j$ in $\calF_i$ if there is a path in $\calF_i$ from
$j$ to $j'$ (in other words, if $j'$ is in the subtree rooted at $j$).

We let $\height^{\calF_i}$ denote the {\em height} of the forest $\calF_i$, 
i.e., the length of
the longest path in $\calF_i$.   For $j \in [d_i]$, 
we let $\height^{\calF_i}_j$ denote the height of information set $j$ in the
forest $\calF_i$, i.e., the length of the longest path from vertex $j$ to a leaf of 
the forest $\calF_i$. 
For a node $u \in P_i$ of the game tree $T$, we will sometimes abuse
notation and use $\height^{\calF_i}_u$ instead of $\height^{\calF_i}_{\Info(u)}$.
Note that $\height^{\calF_i} \leq \height^\calG$, for all $i \in [n]$.

For a behavior strategy $b_i \in B_i$ for player $i$, 
for any information set $j \in [d_i]$, and for any
(other) profile $b'_i \in B_i$, 
we use the notation $(b_i \mid_{(i,j)} b'_i)$
to denote a new behavior strategy $b''_i := (b_i \mid_{(i,j)} b'_i) \in B_i$
which is defined as follows. For every
information set $j' \in [d_i]$,
the local strategy $b''_{i,j'}$ is defined as follows:  
if 
$j'$ is a descendant of $j$
in the information forest $\calF_i$, or if $j'$ is equal to $j$,  then $b''_{i,j'} :=  b'_{i,j'}$.
Otherwise,  $b''_{i,j'} := b_{i,j'}$.
We also use the notation $(b \mid_{(i,j)} b'_i)  :=   (b \mid (b_i \mid_{(i,j)} b'_i))$
to denote a behavior profile which is identical to $b$ except that
player $i$'s  behavior strategy
$b_i$ is replaced by $(b_i \mid_{(i,j)} b'_i)$. 
In other words, $(b \mid_{(i,j)} b'_i)$ is the profile which is identical
to $b$ for all players other than player $i$,  and where for player $i$,
the local strategy at information set $j'$ agrees with $b'_i$ if the
information set $I_{i,j'}$ is reachable from $I_{i,j}$,  and otherwise
it agrees with $b_i$.

We shall also use $\calF_i$ in another way to alter behavior strategies of player $i$.
For the information set forest $\calF_i$ of player $i$, and for 
integer $m$ such that $0 \leq m \leq \height^{\calF_i}$,
let $\calF^m_i$ denote the sub-forest  of $\calF_i$
induced by all vertices $j$ in $\calF_i$ that have height $\height^{\calF_i}_j  \leq m$.
Let   $\calV^m_i$ denote the vertices of $\calF^m_i$.  

For a behavior strategy $b_i \in B_i$ for player $i$, 
for $0 \leq m \leq h^{\calF_i}$,
and for any other behavior strategy, $b'_i \in B_i$,
we use $(b_i \mid_{m}  b'_i)$ to denote the behavior strategy that is 
given by local strategy $b'_{i,j}$
for every $j \in \calV^m_i$, and by the original local strategy $b_{i,j}$,
for all other $j \in [d_i] \setminus \calV^m_i$.
We also use the notation $(b \mid_{m} b'_i)  :=   (b \mid (b_i \mid_m b'_i))$
to describe a profile that is identical to $b$, except that 
behavior strategy $b_i$ for player $i$ is replaced by $(b_i \mid_m b'_i)$.

Recall $U_k(x)$ is the polynomial representing
the expected payoff function
to player $k$ under a behavior profile $x$.
For fixed $b_i \in B_i$, we will use the notation $U_k( x \mid_{(i,j)} b_i)$
to denote the polynomial obtained from $U_k(x)$ as follows: 
for any $j' \in [d_i]$, if information set 
$I_{i,j'}$ is reachable from information set $I_{i,j}$, then
the associated variables $x_{i,j'}$ are fixed to their values in the local strategy $b_{i,j'}$.
Likewise,  for $0 \leq m \leq h^{\calF_i}$,
$U_k(x \mid_m b_i)$ denotes the polynomial
obtained from $U_k(x)$ as follows: for every $j' \in \calV^m_i$,
the variables $x_{i,j'}$ are fixed to their values in $b_{i,j'}$.

\noindent {\bf Normal Form.}
A finite {\em normal form game} (NFG), $\Gamma = (N,(S_i)^n_{i=1},(u_i)^n_{i=1})$,
consists of a finite set $N = \{1,\ldots, n\}$ of players, a finite set 
$S_i$ of 
pure strategies
for each player $i$, and a payoff function 
$u_i: S \rightarrow \nat_+$ for each player\footnote{Again, we restrict
w.l.o.g. to positive integer payoffs, for computational purposes.} 
$i$,
where $S = S_1 \times \ldots \times S_n$.
For every finite $n$-player EFG(PR), $\mathcal{G}$,  
there is an associated {\em standard normal form} game, 
$\Normal(\calG) = (N,(S_i)^n_{i=1},(u_i)^n_{i=1})$, where the set of pure
strategies $S_i$ for player $i$ in $\Normal(\calG)$ is the 
set of pure strategies for
player $i$ in $\mathcal{G}$,
and where the payoff function, $u_i(\cdot)$, for each player $i$ 
is defined by $u_i(s) := U_i(s)$ for all $s \in S$,
where $U_i(s)$ is the expected payoff in $\mathcal{G}$ to player $i$ under pure profile $s$.
For NFGs we use the same notations ($\sigma_i$, $\sigma$, $U_i(\sigma)$, etc.)
for mixed strategies, mixed profiles, and their 
expected payoffs, etc.,  as we do for EFGPRs.
Note that the encoding size $|\Normal(\calG)|$ of the 
NFG $\Normal(\calG)$ is in 
general exponential in $|\calG|$, because already when there
are  two actions available at each information set, the number of strategies $|S_i|$ 
of player $i$ is $2^{d_i}$, where $d_i$ is the number of information
sets belonging to player $i$.

In the other direction, we can easily convert any NFG 
$\Gamma = (N,(S_i)^n_{i=1},(u_i)^n_{i=1})$
to an ``equivalent'' EFGPR, $\calE(\Gamma)$, which is
not much bigger in terms of encoding size than $\Gamma$.   
Specifically, let the action alphabet $\Sigma$ of $\calE(\Gamma)$ be 
the disjoint union of pure strategies of $\Gamma$,
$\Sigma = \dot\bigcup^n_{i=1} S_i$, and let the nodes $V$
of the game tree of $\calE(\Gamma)$ be
$V :=  \{ s_1 s_2 \ldots s_k \mid k \leq n \ \text{and, for all } j \in [k]: \  
s_j \in S_j \}$.   The player partition is given as follows: 
$P_0 = \emptyset$
and for all $i \in [n]$:  $P_i := \{ u \in V \mid  |u| = i-1 \}$.
There is only one information set for each player $i \in [n]$:
namely $I_{i,1} := P_i$.    Finally, the leaves are the nodes 
$\Leaves := \{ u \in V \mid |u| = n\}$, and the payoff
functions $r_i$ are defined as follows, for all $i \in [n]$:  for any leaf 
$s_1 s_2 \ldots s_n \in \Leaves$,   
$r_i(s_1 s_2 \ldots s_n) :=  u_i(s_1, s_2, \ldots, s_n)$.
Note that $\calE(\Gamma)$ clearly has perfect recall since
``there is nothing to remember'':
for any player $i \in [n]$ and any nodes $u , v \in P_i$, 
the visible histories  $Y_i(u)$ and $Y_i(v)$ are both the 
empty sequences, and thus equal, because there is no
ancestor of $u$ or $v$ belonging to $P_i$.
The encoding size of $\calE(\Gamma)$ is certainly polynomial
 in the encoding size of $\Gamma$  (and with judicious
encoding of the various parts of $\calE(\Gamma)$ it could be made essentially linear).  
It is not hard to see that the games $\Gamma$ and $\calE(\Gamma)$
are essentially ``equivalent'' in every respect that matters 
to us (including for computational purposes).  
Note,
in particular, that
there is a one-to-one correspondence,  which respects payoffs,
between the mixed strategies of $\Gamma$ and the behavior strategies
of $\calE(\Gamma)$.

\noindent {\bf Equilibrium.}
For a NFG, $\Gamma = (N,(S_i)^n_{i=1},(u_i)^n_{i=1})$, a mixed strategy $\sigma'_i$ for
player $i$ is called a {\em best response} to a mixed
profile $\sigma = (\sigma_1, \ldots, \sigma_n)$ 
if $U_i(\sigma \mid \sigma'_i) \geq U_i(\sigma \mid \sigma''_i)$ for all
mixed strategies $\sigma''_i$.
Note that $\sigma'_i$ is a best response to $\sigma$ if and only if,
for every pure strategy $c \in \support(\sigma'_i)$, and
for every strategy $c' \in S_i$,  $U_i(\sigma \mid \pi^c_i) \geq 
U_i(\sigma \mid \pi^{c'}_i)$.
A mixed profile $\sigma$ is called a {\em Nash equilibrium} (NE) for $\Gamma$
if $\sigma_i$ is a best response to $\sigma$ for all $i$.
Nash \cite{Nash51} showed every (finite) NFG has an 
NE.
It follows that the standard normal form game $\Normal(\calG)$ associated 
with an EFGPR, $\mathcal{G}$, has a
mixed NE, $\sigma^* \in M$, which by definition is also a 
mixed Nash equilibrium 
of $\mathcal{G}$.  
We can say more. 
In light of Proposition \ref{prop:kuhn-mix-behave-lemma},
a behavior strategy $b'_i \in B_i$ for player $i$  is called
a {\em best response} to 
a behavior profile $b \in B$
if for all $b''_i \in B_i$,  
$U_i(b \mid b'_i) \geq U_i(b  \mid b''_i)$.
A profile $b = (b_1, \ldots, b_n) \in B$ 
is call a {\em Nash equilibrium (NE)} in behavior strategies
if  for all players $i$, $b_i$ is a best response to $b$. 
Combining Proposition \ref{prop:kuhn-mix-behave-lemma} and
Nash's theorem applied
to the standard normal form $\Normal(\calG)$, it follows 
that a NE 
in behavior strategies
exists for any EFGPR, $\mathcal{G}$.

A profile $b \in B$ is called a {\em subgame-perfect equilibrium} (SGPE) 
if $b$ induces a Nash equilibrium on every subgame $\mathcal{G}_u$ of
$\mathcal{G}$.  In other words, for every subgame $\mathcal{G}_u$,
if we confine the behavior profile $b$ to the subtree $T_u$ rooted at $u$,
it induces a Nash equilibrium $b^u$ for the subgame $\mathcal{G}_u$.
Again, a SGPE in behavior strategies exists
for any EFGPR \cite{Selten65},
and of course subgame-perfection is a {\em refinement} of NE:
the SGPEs form a subset of the NEs. 

We now discuss several notions of ``approximate'' and ``almost'' 
equilibrium for normal form and extensive form games.
The well known notion of a ``$\epsilon$-NE'' for
a NFG is a profile where, informally, no player can improve
its own payoff by more than $\epsilon$ by switching its strategy unilaterally.
This of course can be defined analogously for EFGs and EFGPRs.
However, to avoid confusion in terminology between this notion and the very 
different notion (introduced by Myerson \cite{Myerson78})
of $\epsilon$-perfect equilibrium ($\epsilon$-PE), which we define shortly, 
we will use the different terminology ``$\delta$-almost-NE''
to refer to what would usually be called a ``$\delta$-NE'' in the
literature.

Formally, for $\delta > 0$, we call a behavior strategy $b'_i \in B_i$
for player $i$ a {\em $\delta$-almost best response}
to a profile $b \in B$ 
if  for all $b''_i \in B_i$,  
$U_i(b \mid b'_i) \geq U_i(b  \mid b''_i) - \delta$.
We call a profile $b = (b_1,\ldots,b_n) \in B$
a {\em $\delta$-almost Nash equilibrium}
($\delta$-almost-NE),   
if for all players $i$, $b_i$ is a $\delta$-almost best response
to $b$.
For $\delta > 0$, we define a 
{\em $\delta$-almost subgame-perfect equilibrium} ($\delta$-almost-SGPE), to be a profile $b \in B$ which induces
a $\delta$-almost-NE, $b^u$, on every subgame 
$\mathcal{G}_u$ of $\mathcal{G}$.
Note that ``$\delta$-almost-SGPE'' is a refinement
of ``$\delta$-almost-NE''.

As mentioned, Selten \cite{Selten75} pointed out that SGPE
has inadequacies as a refinement
of NE.
For this reason, Selten defined a  
more refined notion of perfect equilibrium,
based on {\em ``trembling hand''} perfection.
Two distinct notions emerge from this:
{\em normal form perfect equilibrium} (NF-PE) 
and  {\em extensive form perfect equilibrium}  (PE). 
We shall find it very useful to provide
Myerson's \cite{Myerson78} alternative definitions for
these notions,
going via the notion
of ``$\epsilon$-perfect equilibrium''. 
Myerson originally defined $\epsilon$-PE for NFGs, but his
definition adapts readily to EFGPRs (see, e.g., \cite{vanDamme91,vanDamme84}).
Although Myerson's definition of PE via $\epsilon$-PEs (adapted to EFGPRs) differs from
the original definition of (extensive form) PE given by Selten \cite{Selten75},
it is equivalent; see, e.g. \cite{Myerson78,vanDamme91,vanDamme84}.
(The key reason
for the equivalence was already pointed out
by Selten himself in (\cite{Selten75}, Lemma 7 \& 8), as we shall highlight later.)

For an NFG\footnote{For example, but not necessarily,
for the standard normal form $\Normal(\calG)$ of 
an extensive form game $\mathcal{G}$.},
 $\Gamma = (N,(S_i)^n_{i=1},(u_i)^n_{i=1})$,
and for $\epsilon > 0$, 
a mixed profile $\sigma \in M$
is called a {\em  $\epsilon$-perfect equilibrium}
($\epsilon$-PE) of $\Gamma$ if it is both
(a):  {\em fully mixed} meaning $\sigma \in M^{> 0}$,
and (b):  for every player $i$ and pure strategy $c \in S_i$, if $\sigma_{i}(c) > \epsilon$,
then the pure strategy $\pi^c_{i}$ is a best response for player $i$ 
to $\sigma$,  in other words, $U_i(\sigma \mid \pi^c_{i}) \geq
U_i(\sigma \mid \pi^{c'}_i)$ for all $c' \in S_i$. 
Likewise, we call $\sigma$ 
a {\em $\delta$-almost $\epsilon$-perfect equilibrium}
($\delta$-almost-$\epsilon$-PE) of $\Gamma$ if $(a)$ holds and, instead of condition $(b)$, 
$\sigma$ satisfies the following condition
$(b')$: for every player $i$ and pure strategy $c \in S_i$, if $\sigma_{i}(c) > \epsilon$,
then the pure strategy $\pi^c_{i}$ is a $\delta$-almost best response for player $i$
to $\sigma$,  in other words, $U_i(\sigma \mid \pi^c_{i}) \geq
                                         U_i(\sigma \mid \pi^{c'}_i) - \delta$,
for all $c' \in S_i$.

We call a mixed profile $\sigma^*$,  a  (trembling hand) 
{\em perfect equilibrium} (PE) of $\Gamma$
if it is a {\em limit point} of a sequence of $\epsilon$-PEs of $\Gamma$  
(with $\epsilon \rightarrow 0$).
In other words,  
$\sigma^*$ is a PE iff there is a sequence $\epsilon_k > 0$, $k \in \nat$,
such that $\lim_{k \rightarrow \infty} \epsilon_k = 0$, 
and such that for all $k \in \nat$ there is an $\epsilon_k$-PE,
$\sigma^{\epsilon_k}$ of $\Gamma$, with
$\lim_{k \rightarrow \infty} \sigma^{\epsilon_k} = \sigma^*$.
Every NFG, $\Gamma$, has a PE, and every PE is both a NE and a SGPE (\cite{Selten75}).

For a EFGPR, $\mathcal{G}$,  
a local strategy $b'_{i,j} \in B_{i,j}$ 
is called a {\em local best response} to a profile $b \in B$  if
for all local strategies $b''_{i,j} \in B_{i,j}$, 
$U_i(b \mid  b'_{i,j}) \geq  
U_i(b \mid b''_{i,j})$.  
It is not hard to show that $b'_{i,j}$ is a local best response iff 
$U_i(b \mid b'_{i,j}) \geq U_i(b \mid \pi^a_{i,j})$ for all $a \in \calA_{i,j}$.
For $\delta > 0$, 
a local strategy $b'_{i,j} \in B_{i,j}$ is called a {\em $\delta$-almost local best response}
to a profile $b \in B$  if for all $b''_{i,j} \in B_{i,j}$,   
$U_i(b \mid  b'_{i,j}) \geq  
U_i(b \mid b''_{i,j}) - \delta$.
Again, $b'_{i,j}$ is a $\delta$-almost local best response to $b$ if and only if for all actions 
$a \in \calA_{i,j}$,  $U_i(b \mid b'_{i,j}) \geq U_i(b \mid \pi^a_{i,j}) - \delta$.

For an EFGPR, $\mathcal{G}$,  
and for $\epsilon > 0$,
a behavior profile $b \in B$ is called a {\em $\epsilon$-perfect equilibrium} 
($\epsilon$-PE), if it is (a): {\em fully mixed}, meaning $b \in B^{> 0}$,
and (b):  for all $i$, $j$, and all
$a \in {\mathcal A}_{i,j}$, 
if $b_{i,j}(a) > \epsilon$, then $\pi^a_{i,j}$ is
a local best response to $b$.
It other words, if a local strategy $b_{i,j}$ places probability
greater than $\epsilon$ on action $a$,
then unilaterally switching the local strategy $b_{i,j}$ to  
pure action $a$ is a local best response to $b$.

For $\delta > 0$, and $\epsilon > 0$, a behavior profile $b \in B$ is called a {\em $\delta$-almost $\epsilon$-perfect
equilibrium}  ($\delta$-almost-$\epsilon$-PE) of $\calG$,  if it is (a.):
fully mixed, $b \in B^{> 0}$, and (b.):  for all $i$, $j$, and all $a \in \calA_{i,j}$ 
if $b_{i,j}(a) > \epsilon$,   then 
$\pi^a_{i,j}$ is a $\delta$-almost local best response to $b$. 

We call a behavior profile $b^* \in B$ a 
{\em extensive form 
perfect equilibrium} (PE) of $\mathcal{G}$ if it is a limit point of $\epsilon$-PEs
of $\mathcal{G}$  (where $\epsilon \rightarrow 0$). 
Selten \cite{Selten75} showed that every EFGPR, $\calG$, has a PE, and 
that every PE is also a SGPE of $\mathcal{G}$  (so, PE refines both SGPE 
and NE).\footnote{Please note that we have overloaded the
``($\epsilon$-)PE'' terminology to apply to both ($\epsilon$-)PE for
NFGs and {\em extensive form} ($\epsilon$-)PE for EFGPRs.
The reason for this overloading will become clear when we discuss
{\em agent normal form}.\\
We remark that it is easier to see why (extensive form) PE 
refines SGPE via
Selten's original definition of PE (via {\em perturbed} games). 
But Myerson's definition, via $\epsilon$-PEs,   
has particular advantages for our purposes, as we'll see.}

A different refinement of equilibrium for a EFGPR, $\calG$, 
is a  {\em normal form perfect equilibrium}
(NF-PE).  This is, by definition, a behavior profile
$b \in B$  such that the (realization equivalent)
mixed profile $\sigma[b]$ induced by $b$ 
is a PE of the standard normal form game, $\Normal(G)$.
We note that even a pure PE of an EFGPR, $\calG$, is not
necessarily a NF-PE (i.e., does not necessarily induce 
a PE of $\Normal(G)$)), and nor is a 
pure NF-PE (i.e., a pure PE of $\Normal(G)$)
necessarily a PE of $\calG$ (see \cite{vanDamme91}, Chapter 6).
So, for EFGPRs, the two notions of PE and NF-PE are incompatible.
In fact,  a NF-PE of $\calG$ is not necessarily even a SGPE
(there are examples where it is not), 
and note that Selten's purpose for defining PE was to refine
subgame-perfect equilibrium.
So, it is not unreasonable to argue that PE is the more relevant
notion for EFGPRs.   Our results apply to approximating both 
a PE and a NF-PE for EFGPRs. 
(By contrast, the results of \cite{vonStengel-van-den-elzen-talman:2002} 
apply only to computing NF-PE for 2-player EFGPRs.)

We next define {\em quasi-perfect equilibrium} (QBE), and the associated notions: $\epsilon$-QPE.
For an EFGPR, $\mathcal{G}$,  
and for $\epsilon > 0$,
a behavior profile $b \in B$ is called a {\em $\epsilon$-quasi-perfect equilibrium} 
($\epsilon$-QPE), if it is (a.): {\em fully mixed}, $b \in B^{> 0}$, and (b.):  for all players $i$, 
all $j \in [d_i]$, and all
actions $a, a' \in {\mathcal A}_{i,j}$, 
if  $(\max_{b'_{i} \in B_{i}} U_i(b \mid_{(i,j)} (b'_{i} \mid \pi^{a}_{i,j}))) < 
(\max_{b'_{i} \in B_{i}} U_i(b \mid_{(i,j)} (b'_{i} \mid \pi^{a'}_{i,j})))$   then  $b_{i,j}(a) \leq \epsilon$.

(We shall delay the analogous definition of ``$\delta$-almost $\epsilon$-quasi-perfect
equilibrium'' until Section \ref{sec:delta-almost-epsilon-PE},
because it will require further definitions. )

We call a behavior profile $b^* \in B$ a 
{\em quasi-perfect equilibrium} (QPE) of $\mathcal{G}$ if it is a limit point of $\epsilon$-QPEs
of $\mathcal{G}$  (where $\epsilon \rightarrow 0$). 
It was shown by van Damme \cite{vanDamme84} that every EFGPR has at least one QPE.  Furthermore, as noted by van Damme in \cite{vanDamme84},
QPE  refines NF-PE.  (We will highlight this again
in Proposition \ref{prop:kreps-wilson} below.)

Finally, we define the notion of {\em sequential equilibrium} due to Kreps and Wilson  \cite{Kreps-Wilson:1982}.
We need the notion of a {\em system of beliefs}.   For a EFGPR, $\calG$, 
with game tree $T =(V,E)$,  a {\em system of beliefs} (or {\em belief system}) 
is a map $\mu: (\Internal \setminus P_0) \rightarrow [0,1]$ such that
that for all players $i \in [n]$ and all $j \in [d_i]$, 
we have $\sum_{u \in I_{i,j}} \mu(u) = 1$.
Let $\SysB$ denote the set of all belief systems 
(associated with the game $\calG$).
An {\em assessment} is a pair $(b,\mu) \in B \times \SysB$, 
where $b$ is a behavior strategy profile,
and $\mu$ is a belief system.   Intuitively, in assessment $(b, \mu)$, for a node $u \in I_{i,j}$,
the belief $\mu(u)$ 
represents the probability that player $i$ assigns to the play hitting node 
$u$ assuming profile $b$ is played, if player $i$ finds out that the 
play has hit information set $I_{i,j}$.  
For any node $u \in I_{i,j}$, let $\Prob_b(u \mid I_{i,j}) = \Prob_b(u)/\Prob_b(I_{i,j})$ denote
the conditional realization probability of reaching node $u$, under 
profile $b$, conditioned on reaching (i.e., realizing)
information set $I_{i,j}$.  This is well-defined whenever $\Prob_b(I_{i,j}) > 0$.

We will call a belief system $\mu$  {\em suitable for behavior profile $b$}
if for all information sets $I_{i,j}$ such that $\Prob_b(I_{i,j}) > 0$,
for all nodes $u \in I_{i,j}$, $\mu(u) = \Prob_b(u \mid I_{i,j})$.
Note that if $b$ is a fully mixed profile then there
is a {\em unique} belief system suitable for $b$, 
because $\Prob_b(I_{i,j}) > 0$ for
all information sets $I_{i,j}$.
Accordingly, when $b$ is a fully mixed behavior profile, 
we denote the unique belief system suitable for $b$ by $\mu^b$, and
we say that $\mu^b$ is {\em the belief system generated by} $b$.
Note that given an EFGPR, $\calG$, and given a fully mixed (rational) profile $b \in B^{> 0}$,
we can easily compute the belief system $\mu^b$ generated by $b$
in time polynomial in $|\calG| + \size(b)$, because
the conditional probability $\mu^b(u) = \Prob_b(u \mid I_{i,j}) = \Prob_b(u)/\Prob_b(I_{i,j})$ is
easy to compute given $\calG$, $b$, and $u$. (By
Proposition \ref{prop:expected-poly} the numerator and denominator
are defined by multilinear polynomials, whose value can be
easily evaluated at $b$, given $\calG$ and $b$, in time polynomial in 
$|\calG| + \size(b)$.)

For any node $u \in V$, and for any leaf $z \in \Leaves$, let $\Prob^u_b(z)$
denote the probability that leaf $z$ is reached if the game is 
started at node $u$
and the profile $b$ is played.   
For any information set $I_{i,j}$, define the probability distribution 
$\Prob^{i,j}_{b,\mu}(z)$ on leaves
by: $\Prob^{i,j}_{b,\mu}(z) := \sum_{u \in I_{i,j}} \mu(u) \cdot \Prob^u_b(z)$, for all $z \in \Leaves$.
Then the {\em expected payoff with respect to assessment $(b,\mu)$, starting in information set $I_{i,j}$},
is defined by $U^{\mu,j}_{i}(b) = \sum_{z \in \Leaves} \Prob^{i,j}_{b,\mu}(z) \cdot r_i(z)$.
A behavior strategy $b'_i$ for player $i$ is called a {\em best reply at information set $I_{i,j}$ 
against assessment $(b,\mu)$}  if $U^{\mu,j}_{i}(b \mid b'_i) = \max_{b''_i \in B_i} U^{\mu,j}_{i}(b \mid b''_i)$.
We say that profile $b$ is a {\em sequential best reply against assessment $(b,\mu)$} if
for all players $i$, and all information sets $I_{i,j}$,   $b_i$ is a best reply at information set
$I_{i,j}$ against assessment $(b,\mu)$.
An assessment $(b,\mu)$ is called a {\em sequential equilibrium} (SE)  of $\calG$ 
if: 
there exists a 
sequence $\langle (b^k,\mu^{b^k}) \mid k \in \nat \rangle$ of assessments, such
that for all $k \in \nat$,
$b^k$ is fully mixed and $\mu^{b^k}$ is the belief system generated by $b^k$,
and $\lim_{k \rightarrow \infty} (b^k,\mu^{b^k}) = (b,\mu)$ (this conditioned is usually
called {\em consistency} of $(b,\mu)$), 
and furthermore $b$ 
is a sequential best reply against $(b,\mu)$.
Kreps and Wilson (\cite{Kreps-Wilson:1982}) 
showed the following facts about sequential equilibrium
(the facts relating QPE to SE and NF-PE were shown later by van Damme \cite{vanDamme84}):

\begin{proposition}[\cite{Kreps-Wilson:1982}; \cite{vanDamme84}]
\label{prop:kreps-wilson}
For any EFGPR, $\calG$: 
\begin{enumerate}
\item (\cite{Kreps-Wilson:1982}) An SE,  $(b',\mu')$, exists for $\calG$.
\item  (\cite{Kreps-Wilson:1982}) 
For every SE, $(b',\mu')$, of $\calG$,  the behavior profile 
$b'$ is a SGPE of $\calG$.

\item (\cite{Kreps-Wilson:1982}) For every PE, $b^*$, of $\calG$, 
there is a system of beliefs $\mu^*$ such that $(b^*,\mu^*)$ is a SE. 
In this sense, we say ``every PE is a sequential equilibrium''.\footnote{The 
converse is false: there are EFGPRs with an SE, $(b',\mu')$,
such that $b'$ is
far from any PE.  See, e.g., \cite{Kreps-Wilson:1982,vanDamme91}.}\\
In fact, for every PE,  $b^*$, of $\calG$,  if $\langle (b^k,\mu^{b^k}) \rangle_{k \in \nat}$ denotes any 
sequence where, for all $k \in \nat$, $b^k$
is a fully mixed behavior profile 
which is a $(1/k)$-PE
for $\calG$, and
$\mu^{b^k}$ is the belief system generated by $b^k$, 
and where $\lim_{k \rightarrow
\infty} b^k = b^*$ and $\lim_{k \rightarrow \infty} \mu^k = \mu^*$,  then 
$(b^*,\mu^*)$ is a SE of $\calG$.

\item (\cite{vanDamme84}) For every QPE, $b^*$, of $\calG$, 
there is a system of beliefs $\mu^*$ such that $(b^*,\mu^*)$ is a SE. 
In this sense, we again say ``every QPE is a sequential equilibrium''.\footnote{The 
converse is again false: there are EFGPRs with an SE, $(b',\mu')$,
such that $b'$ is
far from any QPE.  See \cite{vanDamme84}.}\\
In fact, for every QPE,  $b^*$, of $\calG$,  if $\langle (b^k,\mu^{b^k}) \rangle_{k \in \nat}$ denotes any 
sequence where, for all $k \in \nat$, $b^k$
is a fully mixed behavior profile which is a $(1/k)$-QPE for $\calG$,
$\mu^{b^k}$ is the belief system generated by $b^k$, 
and where $\lim_{k \rightarrow
\infty} b^k = b^*$ and $\lim_{k \rightarrow \infty} \mu^k = \mu^*$,  then 
$(b^*,\mu^*)$ is a SE of $\calG$.

\item (\cite{vanDamme84}) Every  QPE, $b^*$, of $\calG$  is a NF-PE.\\
(Recall: for $b^*$ is a NF-PE of $\calG$ means that the mixed profile 
 $\sigma[b^*] = (\sigma^{b^*_1}_1, \ldots, \sigma^{b^*}_n)$
induced by $b^*$ is
a PE of the standard NFG, $\Normal(\calG)$.)  
\end{enumerate}
\end{proposition}

\begin{figure*}

\begin{center}
\begin{tikzpicture}[scale=0.85]
[->,>=stealth',shorten >=1pt,auto,node distance=2.8cm, semithick,
 every state/.style={draw=blue,fill=blue!20,text=black}]
  \node[state]  (NE)         at (3,0)   {{\Large NE}};
  \node[state] (PE) at (1,5.5) {{\Large PE}};
  \node[state]          (NFPE) at (4,2.3) {{\small NF-PE}};
  \node[state]          (SE) at (3,4) {{\Large SE}};
  \node[state]          (SGPE) at (2,2) {{\small SGPE}};
  \node[state]          (QPE) at (5,5.5) {{\large QPE}};
  \path[->] (NE) edge      node {} (SGPE)
             edge               node {} (NFPE)
        (SGPE) 
            edge   node {} (SE)
        (SE) edge  node {} (PE)
            edge               node {}  (QPE)
        (NFPE) edge  node {} (QPE);
\end{tikzpicture}

\end{center}

\caption{\label{fig:refinement-hasse} 
Hasse diagram of the mentioned equilibrium refinements for EFGPRs.}
\end{figure*}
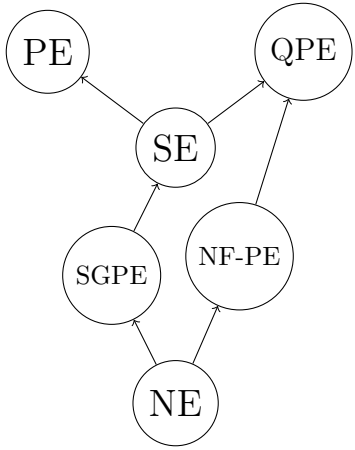

Figure \ref{fig:refinement-hasse}
summarizes the mentioned refinement relationships 
between the various equilibrium notions that we have defined
for EFGPRs: it depicts the Hasse diagram 
of the refinement partial order.
In the diagram, a directed edge 
$X \rightarrow Y$ means that equilibrium notion $Y$ refines notion $X$,
i.e., that every $Y$-equilibrium is also a $X$-equilibrium.
Moreover, whenever there is no directed path in this Hasse diagram from a 
node $X$
to a node $Y$, that means there exist known examples of EFGPRs where
 a $Y$-equilibrium is {\em not} an $X$-equilibrium.   (So, this is a partial order not
because we lack knowledge of an underlying richer (total) order: 
no other refinement relationships exist for general EFGPRs,
other than those implied by this Hasse diagram.)

It is noteworthy that there {\em can not exist}
some more refined equilibrium notion that 
refines {\em both} PE and QPE, {\em and} exists in every EFGPR. 
In particular,  Mertens \cite{Mertens95} 
has given a simple example of a 2-player EFGPR 
whose set of PEs  is {\em disjoint} from its set of NF-PEs (and 
whose NF-PEs consist of just one dominant strategy equilibrium). 
Thus, since QPE refines NF-PE, the set of PEs of Mertens' EFGPR is 
also disjoint from its set of QPEs.
Mertens argues, partly based on this example, that  QPE is preferable to PE
 as a refinement for EFGPRs:  a dominant strategy equilibrium, when it exists, 
is generally prized, and it is always a QPE, but it is not necessarily a PE as shown by Mertens's example.
Mertens's example  shows we can not hope for some (as yet unknown) 
``most refined''
notion of equilibrium for EFGPRs, which always exists, and which 
refines all the refinements we have mentioned.
It is worth mentioning however
 that the results of \cite{BlumeZame94} and \cite{PimientaShen2014} combined
show that
if a EFGPR is suitably {\em ``generic''}\footnote{Here ``generic'' means 
the EFGPR has some  
``structure'' $\Psi$ (which excludes the payoff information) and has a vector of payoff functions
 $r=(r_1,\ldots,r_n) \in \real^m$ such that  $r \not\in R[\Psi]$;   where $R[\Psi] \subseteq \real^m$ is 
a certain (semi-algebraic) ``forbidden'' set of 
{\em dimension} strictly less than $m$.},
then its set of PEs, QPEs, and SEs are all the same.
However, many natural games that we might encounter may not be 
``generic'' in this sense, as illustrated by the various 
simple and natural examples of games provided in, e.g., \cite{vanDamme91,Mertens95,vanDamme84,MasSolZam13}, where PE, SE, and 
QPE do not coencide.

\noindent {\bf Agent Normal Form.}   Kuhn \cite{Kuhn53} and Selten \cite{Selten75}  considered an alternative way to associate 
a normal form game with
a given EFGPR, $\mathcal{G}$,
which they called the {\em agent normal form}.
The agent normal form game, $\Agent\Normal(\calG)$, is defined as follows. 
$\Agent\Normal(\calG)$ has a player, called an {\em agent}, associated with each
information set $I_{i,j}$ of the EFGPR, $\calG$.  
Thus if $\calG$ has $n$ players and player $i$ has $d_i$ information sets, then
the total number of agents in $\Agent\Normal(\calG)$ is $d = \sum^n_{i=1} d_i$, which is the total
number of information sets in $\calG$.
We refer to each agent in $\Agent\Normal(\calG)$ by its index: $(i,j)$.
The set of pure strategies 
for agent $(i,j)$ in $\Agent\Normal(\calG)$ is given
by the set ${\mathcal A}_{i,j}$ of actions available to player $i$ of $\mathcal{G}$
in the information set $I_{i,j}$.
Thus, note that the set of mixed strategies for agent $(i,j)$ in $\Agent\Normal(\calG)$ 
is in one-to-one correspondence with the set of local strategies $B_{i,j}$ 
for player $i$ at information set $I_{i,j}$
in the EFGPR, $\calG$.
Thus also, the set of profiles of mixed strategies in $\Agent\Normal(\calG)$ 
is in one-to-one correspondence with the set $B$ of behavior strategy profiles in $\calG$.
Moreover, the set of pure strategy profiles of the agents in $\Agent\Normal(\calG)$
is in one-to-one correspondence with the set of pure strategy profiles $S$ in $\calG$.
Thus, hereafter, we use $S$ interchangeably, to denote both the sets of pure profiles
for $\calG$ and for $\Agent\Normal(\calG)$, and we also use $B$
interchangeably, to denote both the set of {\em behavior} profiles of $\calG$ and
the set of {\em mixed} profiles of $\Agent\Normal(\calG)$.

We define
the payoff functions, $u_{(i,j)}(s)$, of $\Agent\Normal(\calG)$ as follows: 
given a pure profile $s \in S$ for the $d$ agents, the payoff to agent
$(i,j)$ is given by $u_{(i,j)}(s) := U_i(s)$.   In other words, 
the payoff for every agent $(i,j)$ in  $\Agent\Normal(\calG)$
under profile $s$ is the expected payoff of player $i$ in $\calG$
under the same profile $s$.
Thus, the goal
of all the agents $(i,j)$ who are ``acting on behalf of'' player $i$,
is aligned exactly with the goal of player $i$.
It follows that also 
the {\em expected} payoff, $U_{(i,j)}(b)$, to agent $(i,j)$ under any mixed profile
$b \in B$ in $\Agent\Normal(\calG)$ is equal to the expected payoff $U_i(b)$
of player $i$ under the same (behavior) profile $b \in B$ of $\calG$.

A simple but important fact,
that follows immediately from the definitions we have given for ($\epsilon$-)PEs,
is that
the set of ($\epsilon$-)PEs of $\calG$ is equal to the set
of ($\epsilon$-)PEs of $\Agent\Normal(\calG)$.\footnote{This is why we overload the 
``($\epsilon$-)PE'' terminology 
for the corresponding notions of both NFGs and EFGPRs.}

\begin{proposition}
[cf. \cite{Selten75} Lemma 7, \& \cite{Myerson78}; see also \cite{vanDamme91}]
\label{prop:local-opt-global-opt}
For a EFGPR, $\calG$,
and $\epsilon > 0$,   
a behavior profile $b \in B$ is a $\epsilon$-PE of $\calG$ if and only if
$b$ is a mixed $\epsilon$-PE of $\Agent\Normal(\calG)$
(this is true by definition).
\ \ Thus, a profile $b \in B$ is a PE of $\calG$
\ \ iff \ \ $b$ is a PE of $\Agent\Normal(\calG)$.
\end{proposition}

Note, firstly, that it is {\em not} true in general that the set of Nash equilibria
of $\calG$ and $\Agent\Normal(\calG)$ are the same.  There are simple (even 1-player) examples showing this.
This is because even though a profile $b \in B$ might consist entirely
of ``local best responses'' in $\calG$,  some information sets may be 
reached with
probability $0$ under profile $b$, and therefore ``local best responses''
together do not necessarily constitute a ``global'' best response in $\calG$.

Note also that, as mentioned already, no such relationship holds in general between the PEs of $\calG$
and the PEs of its standard normal form $\Normal(\calG)$, in either direction.

Proposition \ref{prop:local-opt-global-opt} holds 
by definition 
because we have used
Myerson's \cite{Myerson78} alternative definition of PEs, via $\epsilon$-PEs.
We remark that the reason why Myerson's definition is equivalent to Selten's
original definition (which we will not give formally)
was shown already by Selten himself. Namely,
Selten defined a PE as a limit point of NEs of a sequence of {\em perturbed} games
(with positive ``perturbations'' going to zero).
In a {\em perturbed} EFGPR, there is a minimum positive probability 
specified for each action available in each information set, and that action 
 must be played with at least that probability in any behavior strategy.
Selten (\cite{Selten75}, Lemma 7) showed that 
for perturbed EFGPRs, a behavior strategy that consists
entirely of ``local best responses'' is also necessarily a ``global'' 
best response. 
As explained already, this {\em does  not} hold in general when the
game is not perturbed. 

We shall need the following ``almost'' 
variant of  Proposition \ref{prop:local-opt-global-opt}, which also follows
immediately from our definitions.

\begin{proposition}
\label{prop:almost-local-opt-global-opt}
For all $\delta > 0$ and $\epsilon > 0$,
for any EFGPR, $\calG$,  
a (behavior) 
strategy profile $b \in B$
is a $\delta$-almost-$\epsilon$-PE
of $\calG$ iff $b$ is a (mixed) $\delta$-almost-$\epsilon$-PE
of $\Agent\Normal(\calG)$.
\end{proposition}

Note that if the agent normal form 
$\Agent\Normal(\calG)$ is represented in the usual way, by providing its
table of payoffs for all possible pure strategy profiles of all the agents,
then just as was the case for standard normal form,
the encoding size $|\Agent\Normal(\calG)|$ 
is also exponential in $|\calG|$,
because the number $|S|$ of pure profiles of $\Agent\Normal(\calG)$ is 
exponential in $|\calG|$.
Nevertheless,  we shall find $\Agent\Normal(\calG)$ very useful
for our computational purposes.

\subsection*{The complexity classes \FIXP, \FIXPA, and \linFIXP ( = \PPAD)}

We shall now define the search problem complexity 
classes \FIXP, \FIXPA, and \PPAD,
which we shall use to characterize the complexity of computing
an equilibrium (of various kinds) for a EFGPR.

A {\em $\{+, -, *, /, \max, \min\}$-circuit} has inputs consisting of variable
$x_1, x_2, \ldots, x_n$, as well as rational constants,
and has a finite number of (binary) computation gates taken from 
$\{+, - , *, /, \max, \min\}$, 
with a subset of the computation gates labeled $\{o_1, o_2, \ldots, o_m\}$ 
and called output 
gates.\footnote{The set of gates $\{+,-,*,/,\max,\min\}$
is of course redundant, e.g., using
rational constants the gates $\{-,\min\}$ can
be simulated by the other gates. }
The class of {\em $\{+,\max\}$-circuits} are the restricted class of 
$\{+,-,*,/,\max,\min\}$-circuits, where
the only allowed gates are $\{+,\max\}$ in addition to gates for
{\em multiplication  by a rational constant}.

When a circuit in this paper is a  
general $\{+,-,*,/,\max,\min\}$-circuit, we shall often just refer
to it simply as 
``circuit", when it is clear from the context. 
We shall also refer to $\{+,\max\}$-circuits as {\em piecewise-linear} circuits.   
A circuit (of either kind) computes a continuous function 
from $\real^n \rightarrow \real^m$ (and $\rat^n \rightarrow \rat^m$) in the natural way. Abusing 
notation slightly, we shall often identify the circuit with the function it computes.

By a (total) {\em multi-valued function}, $f$, with domain $A$ and co-domain $B$,
we mean a function that maps each $a \in A$ to a non-empty subset 
$f(a) \subseteq B$.  
We use $f: A \twoheadrightarrow B$ to denote such a function. Intuitively, when considering a multi-valued function as a computational problem, we are interested in producing just one of the elements of $f(a)$ on input $a$, so we refer to $f(a)$ as the set of {\em allowed outputs}.

A multi-valued function  $f: \{0,1\}^* \twoheadrightarrow \real^*$ is said to be in $\FIXP$ 
if there is a polynomial time 
computable map, $r$, that maps each instance $I \in \{0,1\}^*$ of $f$ to $r(I)= \langle1^{k^I}, 1^{d^I}, P^I, C^I, \phi^I, a^I, b^I \rangle$, where
\begin{itemize}
\item{}$k^I$ and $d^I$ are positive integers.

\item{} $P^I$ is a convex polytope in $\real^{k^I}$, given as a set of 
linear inequalities with rational coefficients.
\item{}$C^I$ is a circuit, with $k^I$ inputs and $k^I$ outputs, which maps $P^I$ to itself.
\item{} $\phi^I :[d^I] \rightarrow [k^I]$ is a finite function, given by its table.
\item{}$a^I, b^I \in \rat^{d^I}$.
\item{}$f(I) = \{(a^I_i y_{\phi^I(i)} + b^I_i)_{i=1}^{d^I} \mid y \in P^I \: \wedge \: C^I(y) = y \}$. 
Note that $f(I) \not = \emptyset$, by Brouwer's fixed point theorem.
\end{itemize}
The above is one of many equivalent characterizations of $\FIXP$ \cite{EY07}.
In particular, it was shown in \cite{EY07} that the gates 
$\{+,*,\max\}$ together with rational constants suffice 
for functions computed by the corresponding circuits
to characterize $\FIXP$,
and furthermore adding other gates such as $k$'th-root gates for any fixed $k$
does not increase the power of $\FIXP$.

A multi-valued function  $f: \{0,1\}^* \twoheadrightarrow \real^*$
is said to be in $\linFIXP$ if it satisfies the same
definition
as for $\FIXP$, except that the circuit $C^I$ must be a $\{+,\max\}$-circuit (recall: with multiplication
by rational constants allowed).  

Informally, $\FIXP$ are those real vector multi-valued functions, with discrete inputs, 
that can be cast as Brouwer fixed point computations for algebraically defined functions,
and $\linFIXP$ is the restriction of those to functions 
that are piecewise-linear.
A multi-valued function $f:  \{0,1\}^* \twoheadrightarrow \real^*$ is said to be $\FIXP${\em -complete} 
(respectively, $\linFIXP${\em -complete}) if:
\begin{enumerate}
\item{}$f \in \FIXP$  (respectively, $f \in \linFIXP$), and
\item{} [$f$ is {\em $\FIXP$-hard} (respectively, $f$ is {\em  $\linFIXP$-hard})]:
for all $g \in \FIXP$  (respectively, $g \in \linFIXP$), there is a polynomial time computable map, mapping instances $I$ of $g$ to 
$\langle y^I, 1^{k^I},
\phi^I, a^I, b^I \rangle$, 
where $y^I$ is an instance of $f$,  where $f(y^I) \subseteq \real^{k^I}$,
$\phi^I:[d^I] \rightarrow [k^I]$ is a function (given by its table), 
$d^I \geq 1$,
and $a^I$ and $b^I$ are 
$d^I$-tuples with rational entries, so that 
$g(I) \supseteq  \{ (a^I_i z_{\phi^I(i)} + b^I_i)^{d^I}_{i=1} \mid 
z \in f(y^I) \}$.  In other words, for any allowed output $z$ of $f$ on input $y^I$, 
the vector $(a^I_i z_{\phi^I(i)} + b^I_i)^{d^I}_{i=1}$ is an allowed output of $g$ on input $I$. \end{enumerate}

In \cite{EY07} it was shown that the multi-valued function
which maps normal forms games, with $n \geq 3$ players,
to their Nash equilibria is $\FIXP$-complete.\footnote{To view the Nash equilibrium problem as a total multi-valued function, $f_{\mbox{\rm \tiny Nash}}: \{0,1\}^* \twoheadrightarrow \real^*$, we can view all strings in $\{0,1\}^*$ as encoding some game, by viewing ``ill-formed" input strings as encoding a fixed trivial game.}

Since the output of a $\FIXP$ function consists of real-valued vectors, and since there exist
circuits whose fixed points are all irrational, a $\FIXP$ function is not directly computable by a Turing machine, and 
the class is therefore not directly comparable with standard complexity classes of {\em discrete} total search problems 
(such as \PPAD, \PLS, or \TFNP). 

Even though we phrased $\linFIXP$ as a class of real-valued
search problems, it can also be viewed as class of {\em discrete}
search problems, because the nature of the functions defined
by $\{+, \max\}$-circuits (with multiplication by rational constants),
over a convex polytope domain $P^I$,
implies that they always have at least one  {\em rational}-valued
fixed point, with encoding size polynomial in that of the circuit.\footnote{Technically, 
to view $\linFIXP$ 
as a {\em discrete} search problem class, comparable to $\PPAD$, etc., we  
likewise close (discrete) $\linFIXP$ under polynomial time (search problem) reductions.}   
In fact, it was shown in \cite{EY07} that $\linFIXP = \PPAD$.
(So, $\linFIXP$ can serve as our definition of $\PPAD$ in this paper.
We will not need the original definition.)

It was shown by Chen and Deng \cite{Chen-Deng06} that the
multi-valued function that maps 2-player NFGs to their NEs is $\PPAD$-complete,
and by 
Daskalakis {\em et al.} \cite{DasGP09} that
the multi-valued function that maps NFGs (with any number of players), and a given rational $\epsilon > 0$,
to their $\epsilon$-NEs is $\PPAD$-complete.

We now define the discrete class $\FIXPA$, also from \cite{EY07}.
A multi-valued function $f: \{0,1\}^* \twoheadrightarrow \{0,1\}^*$ (a.k.a. a totally defined discrete search problem) 
is said to be in $\FIXPA$ if there is a function $f' \in \FIXP$,
and polynomial time computable maps $\delta: \{0,1\}^* \rightarrow \rat_+$ and $g: \{0,1\}^* \rightarrow \{0,1\}^*$, such that for all instances $I$, 
\[ f(I) \supseteq \{ \: g(\langle I, y \rangle) \mid  y \in \rat^* \: \wedge \:  
\exists y' \in f'(I): \:  \|y-y'\|_ \infty \leq \delta(I)  \: \}. \] 
Informally, $\FIXPA$ are those totally defined discrete search problems that reduce to approximating exact Brouwer fixed points.
A multi-valued function $f: \{0,1\}^* \twoheadrightarrow \{0,1\}^*$ is said to be $\FIXPA$-{\em complete} if:
\begin{enumerate}
\item{}$f \in \FIXPA$, and
\item{}[$f$ is {\em $\FIXPA$-hard}]: For all $g \in \FIXPA$, there are polynomial time computable maps 
$r_1, r_2:\{0,1\}^* \rightarrow \{0,1\}^*$,
such that 
$g(I) \supseteq \{ \: r_2(\langle I,z \rangle)  \mid   z \in f(r_1(I)) \: \}$.
\end{enumerate}
In \cite{EY07} it was shown that the multi-valued function 
that maps pairs $\langle \Gamma, \delta \rangle$,
where $\Gamma$ is a NFG and $\delta > 0$, to the set of rational $\delta$-approximations   
(in $\ell_\infty$-distance) of Nash equilibria of $\Gamma$, is $\FIXPA$-complete.

\section{Computing a (extensive form) $\epsilon$-PE, and a $\epsilon$-QPE, is in $\FIXP$}

\label{sec:epsilon-PE-char}

Given a EFGPR, $\calG$, 
we now construct
an algebraically defined  
function, $F^\epsilon_\calG(x)$,
whose Brouwer fixed points 
(for each fixed $\epsilon > 0$), constitute 
$\epsilon$-PEs of $\calG$.  
We likewise construct a function, $H^{\epsilon}_\calG(x)$
whose Brouwer fixed points (for each fixed $\epsilon > 0$),
constitute $\epsilon$-QPEs of $\calG$.
The functions $F^\epsilon_\calG(x)$ and $H^{\epsilon}_\calG(x)$ are
both defined using an algebraic
$\{+,*,\max\}$-circuit whose encoding size is polynomial in $|\calG|$,
and where $\epsilon > 0$ is an input of the algebraic circuit.
Our construction of $F^\epsilon_\calG(x)$ essentially amounts to the same 
construction as given for $\epsilon$-PEs of {\em normal form games} 
in \cite{EHMS14}, 
except when it is applied to the {\em agent normal form}, $\Agent\Normal(\calG)$.
Of course the problem is that we can not afford to actually construct 
$\Agent\Normal(\calG)$, because it is exponentially large.
However, it turns out we do not need to construct $\Agent\Normal(\calG)$
in order to construct $F^\epsilon_{\Agent\Normal(\calG)}(x)$.
We instead  
exploit the fact (Proposition \ref{prop:expected-poly}) that the expected
payoff functions $U_{(i,j)}(x) := U_i(x)$ for agents $(i,j)$
in $\Agent\Normal(\calG)$ are expressible as polynomials
whose encoding size is polynomial in $|\calG|$.
This allows
us to construct  $F^\epsilon_\calG(x) = F^\epsilon_{\Agent\Normal(\calG)}(x)$ 
with encoding
size polynomial in $|\calG|$, 
avoiding the explicit construction of  $\Agent\Normal(\calG)$.

Our construction of the function $H^{\epsilon}_\calG(x)$ for $\epsilon$-QPEs
is based on some similar ideas, but is more involved, and does not
make direct use of the relationship with $\Agent\Normal(\calG)$.

Given a $n$-player EFGPR, $\calG$,  
the space $B$ of behavior 
strategy
profiles for $\calG$ is clearly a compact convex polytope in euclidean
space, $\real^m$,  where $m$ is the dimension of the vectors $b \in B$ that denote behavior
profiles.
Moreover, $B$ can clearly be expressed efficiently using a system of less 
than $3m$ linear inequalities (which define $B$ to be the set of vectors 
$b \in \real^m$ 
in which each local strategy 
$b_{i,j}$ forms a probability distribution on $\calA_{i,j}$).
For $\epsilon > 0$, let $B^\epsilon \subseteq B$ denote the 
polytope of behavior profiles defined by:
$$B^{\epsilon} = \{ b \in B \mid b_{i,j}(a) \geq \epsilon, \  
\mbox{for all $i \in [n]$,
$j \in [d_i]$ and $a \in \calA_{i,j}$} \}.$$

\begin{theorem}
\label{fixp-no-division}
For any EFGPR, $\calG$:
\begin{enumerate}
\item There is a function, $F^{\epsilon}_{\calG}(x): B \rightarrow B^{\epsilon}$,  given by a 
$\{+,*,\max\}$-circuit
computable in polynomial time from $\calG$, with the circuit having both $x$ 
and $\epsilon > 0$ as its inputs, such that for all fixed $0 < \epsilon < 1/m$
(where $m$ is the dimension of vectors $b \in B$), 
every Brouwer fixed point of the function $F^{\epsilon}_{\calG}(x)$ is a $\epsilon$-PE
of $\calG$. In particular, the problem of computing an 
extensive form $\epsilon$-perfect equilibrium 
for a given EFGPR is in \FIXP.

\item 
There is a function, $H^{\epsilon}_{\calG}(x): B \rightarrow B^{\epsilon}$,  given by a 
$\{+,*,\max\}$-circuit
computable in polynomial time from $\calG$, with the circuit having both $x$ 
and $\epsilon > 0$ as its inputs, such that for all fixed $0 < \epsilon < 1/m$
(where $m$ is the dimension of vectors $b \in B$), 
every Brouwer fixed point of the function $H^{\epsilon}_{\calG}(x)$ is a $\epsilon$-QPE
of $\calG$. In particular, the problem of computing a $\epsilon$-QPE
for a given EFGPR is in \FIXP.
\end{enumerate}
\end{theorem}

As mentioned, the proof of Part (1.) of Theorem \ref{fixp-no-division}
is very similar to the proof of
the analogous result for
$\epsilon$-PEs of NFGs given in \cite{EHMS14},
which itself builds on a fixed point characterization of Nash equilibria
from \cite{EY07}.
By Proposition \ref{prop:local-opt-global-opt},
to prove Theorem \ref{fixp-no-division} it suffices to 
find $\epsilon$-PEs of the agent normal form $\Agent\Normal(\calG)$,
because these are the same as $\epsilon$-PEs of $\calG$.
We can not ``construct''  $\Agent\Normal(\calG)$, because it
has size exponential in $\calG$, but we  do not need to.   
We now give the detailed proof for both parts.
Although the proof of Part (1.) is very similar to 
the analogous proof in \cite{EHMS14}, the proof of Part (2.) 
also involves additional constructions and does not appeal to the relationship with $\Agent\Normal(\calG)$.
To facilitate our proof of Part (2.), 
we need some definitions,
and an alternative characterization of $\epsilon$-QPE.

Note that for any fully mixed profile $b \in B^{> 0}$, for any player $i$, 
$j \in [d_i]$, 
and any node $u \in I_{i,j}$, the conditional probability 
$\Prob_b(u | I_{i,j})$  is well-defined, because $\Prob_b(I_{i,i}) > 0$.
Furthermore, importantly, given that $\Prob_b(I_{i,j}) > 0$, 
$\Prob_b(u | I_{i,j})$ is otherwise
``independent'' of $b_i$. It
only depends on the behavior strategies $b_{-i}$ of players other than $i$, 
because, by perfect recall, for
all nodes $u \in I_{i,j}$ the visible history for player $i$ is the same: $Y_{i,j}$.
For $b \in B^{> 0}$, for $i \in [n]$, and for $j \in [d_i]$, we use $U^j_i(b)$ to denote the {\em conditional expected payoff}
to player $i$, conditioned on reaching information set $I_{i,j}$, under profile $b$. 
Again, 
this conditional expectation is well-defined, since $b \in B^{> 0}$. 
Furthermore, again, except for the fact that $\Prob_b(I_{i,j}) > 0$,
the conditional expectation $U^j_i(b)$ is independent of those 
local strategy $b_{i,j'}$ in $b_i$ for information 
sets $I_{i,j'}$ such that the node
$j' \in V^{\calF_i}$ of the information set forest $\calF_i$ is not in 
the subtree of 
$\calF_i$ rooted at node $j \in V^{\calF_i}$. It only
depends on those local strategies $b_{i,j''}$ where $j'' \in V^{\calF_i}$ 
is a node in the subtree of 
$\calF_i$ rooted at $j$.
For $i \in [n]$, $j \in [d_i]$ and $a \in \calA_{i,j}$, and for 
$b \in B^{> 0}$, we define
$$\MU^{j,a}_i(b) := \max_{b'_i \in B_i} U^j_i(b \mid_{(i,j)} (b'_i \mid \pi^a_{i,j})).$$
Thus $\MU^{j,a}_i(b)$ denotes the maximum conditional expected payoff to player $i$, 
conditioned
on reaching information set $I_{i,j}$  using $b$, where 
player $i$ switches to action $a \in \calA_{i,j}$ 
at $I_{i,j}$, and chooses the rest of its strategy $b'_i$ (below information set $I_{i,j}$ in $\calF_i$) so as to maximize
$U^j_i(b \mid_{(i,j)} (b'_i \mid \pi^a_{i,j}))$.
Note that, since $b \in B^{> 0}$,   $\MU^{j,a}_i(b)$ is both well defined and 
``independent'' of $b_i$: it only matters that $\Prob_b(I_{i,j}) > 0$.
Now, observe that, for any $b \in B^{> 0}$, 
for any $i \in [n]$, $j \in [d_i]$, and 
for any $a , a' \in \calA_{i,j}$,    we have:
\begin{equation}
\label{eq:equiv-def-of-qpe}
(\  \MU^{j,a}_i(b)  < \MU^{j,a'}_i(b) \ )  \quad  \Longleftrightarrow   \quad
(\  (\max_{b'_{i} \in B_{i}} U_i(b \mid_{(i,j)} (b'_{i} \mid \pi^{a}_{i,j}))) < 
 (\max_{b'_{i} \in B_{i}} U_i(b \mid_{(i,j)} (b'_{i} \mid \pi^{a'}_{i,j})))
\end{equation}
This equivalence 
holds because the profiles 
$(b \mid_{(i,j)} (b'_{i} \mid \pi^{a}_{i,j})))$  and  
$(b \mid_{(i,j)} (b'_{i} \mid \pi^{a'}_{i,j})))$ differ only within 
player $i$'s local strategies within $b_i$ at information sets $j'$
in the subtree of $\calF_i$ rooted at $j \in V^{\calF_i}$.
Thus, since $\Prob_b(I_{i,j}) > 0$, the strict inequality on the left
of (\ref{eq:equiv-def-of-qpe}) holds if and only if the strict
inequality on the right of (\ref{eq:equiv-def-of-qpe}) holds.
Thus, an alternative definition for a profile $b$ to be a 
{\em $\epsilon$-quasi-perfect equilibrium} 
($\epsilon$-QPE),  is this:  (a.) $b \in B^{> 0}$, and (b.)  
for all $i \in [n]$, $j \in [d_i]$, and $a, a' \in {\mathcal A}_{i,j}$, 
if  $\MU^{j,a}_i(b)  < \MU^{j,a'}_i(b)$,   then  $b_{i,j}(a) \leq \epsilon$.
We will exploit this alternative 
definition.\footnote{Indeed, this is one of the  
equivalent characterizations of $\epsilon$-QPE that was originally given by van Damme 
in \cite{vanDamme84}.   We used a different definition for clarity, 
and for compatibility
with the way we defined $\epsilon$-PE.
In fact, similarly van Damme \cite{vanDamme84} used
a similar equivalent characterization of  $\epsilon$-PE
for an EFGPR, defined as follows:  
(a.) $b \in B^{> 0}$, and $(b.)$, for all $i \in [n]$, $j \in [d_i]$, and $a, a' \in \calA_{i,j}$,
if $U^j_i(b \mid \pi^a_{i,j}) < U^j_i(b \mid \pi^{a'}_{i,j})$ then $b_{i,j}(a) \leq \epsilon$.
Again, it is clear that this is equivalent to the definition we have given 
for $\epsilon$-PE.}

Consider a EFGPR, $\calG$, and
let $b \in B$ have dimension $m$ as vectors in Euclidean space.
Suppose we are given $0 < \epsilon < 1/m$.
For a vector $x$ of variables corresponding to the coordinates of a 
behavior strategy $b \in B$, 
we let $v(x)$ 
be a $m$-vector such that for all $i \in [n]$, $j \in [d_i]$, and $a \in \calA_{i,j}$
$v(x)_{i,j,a} = U_i(x \mid \pi^{a}_{i,j}) = U_{(i,j)}(x \mid \pi^a_{i,j})$.
In other words, for all behavior profiles $b \in B$,  $v(b)_{i,j,a}$ 
is the expected payoff to agent $(i,j)$ in the agent normal form game $\Agent\Normal(\calG)$,
if all agents play according to $b$, except that agent $(i,j)$ switches 
to pure strategy $\pi^a_{i,j}$.
Note that by Proposition \ref{prop:expected-poly},  $v(x)_{i,j,a}$ can be 
expressed 
as a polynomial in the variables $x$ whose encoding size is polynomial  in 
$|\calG|$.

Likewise, let us define $v'(x)_{i,j,a} :=  \MU^{j,a}_i(x)$.
We shall show, in  Lemma \ref{lem:max-expected-poly} below,
that the function $\MU^{j,a}_i(x)$,  defined over $B^{> 0}$, 
can indeed be expressed 
as a $\{+, - , *, /, \max , \min \}$-formula in the variables $x$, where 
the encoding size of the formula is polynomial  in 
$|\calG|$.

\begin{lemma}
\label{lem:max-expected-poly}
Given a EFGPR, $\calG$, 
for all players $i \in [n]$, all information sets $j \in [d_i]$,
and all actions $a \in \calA_{i,j}$,  
there is a 
$\{+, - , *, / , \max\}$-formula 
 $v'(x)_{i,j,a}$ 
(i.e.,
a $\{+,-, *, / , \max , \min \}$-circuit with 
no re-use of subcircuits),  such that
the encoding size of $v'(x)_{i,j,a}$ is polynomial in $| \calG |$, 
and each $v'(x)_{i,j,a}$ can be constructed from $\calG$ in P-time,
and
such that for all fully mixed $b \in B^{> 0}$,   
$v'(b)_{i,j,a} = \MU^{j,a}_i(b)$.
\end{lemma}

\begin{proof}
The basic idea of the proof is that, given $b \in B^{> 0}$, one can
compute $\MU^{j,a}_i(b)$ using dynamic programming,
by working ``bottom up'' on the
information set forest $\calF_i$ for player $i$.
Then the key observation is that
this dynamic program can actually be described by a $\{+,-, *,/, \max\}$-formula
which has encoding size only polynomial in $\calG$.

We next describe the dynamic program, and the resulting formula, in detail.
(We will later need to use facts about the detailed structure of the formula.)
Consider the information set forest $\calF_i$ for player $i$.
Let $\Leaves_{\calF_i}$ denote the set of leaves of $\calF_i$.
Let $\Internal_{\calF_i}$ denote the set of internal nodes of $\calF_i$.
For a node $j \in [d_i] = V^{\calF_i}$,  
and for $a \in \calA_{i,j}$, let us denote
the set of $a$-children of $j$ in $\calF_i$ by:
$\child^a_{\calF_i}(j) = \{ j' \in V^{\calF_i} \mid  (j, a , j') \in E^{\calF_i} \}$.
For an internal node $u \in \Internal$, and for $a \in \Act(u)$,
let $\LLeaves^{u,a} = \{ z \in \Leaves \mid   u a  \sqsubseteq  z \ \& \
\forall m \ \text{such that}  \ ua \sqsubseteq z[m],   \ z[m] \not\in P_{\Pl(u)} \}$.  
In other words, $\LLeaves^{u,a}$ denotes the set of
leaves $z$ of the game tree $T$ that are in the subtree rooted at $ua$,
and such that there is no node on the path from $ua$ to $z$ which
belongs to the same player $\Pl(u)$ that $u$ belongs to.

For $u, v \in V$,
let $\Prob_{b}(v \mid u)$ 
denote
the probability that, using profile $b$,
conditioned on reaching node $u$, 
the play eventually thereafter hits node $v$.
For $i \in [n]$ and $j, j' \in [d_i]$, 
let $\Prob_{b}( I_{i,j'} \mid I_{i,j})$ 
denote conditioned probability of 
reaching information set $I_{i,j'}$,
conditioned on reaching $I_{i,j}$,  when
using profile $b$.

We can define $v'(x)_{i,j,a} := \MU^{j,a}_i(x)$ inductively in a ``bottom up''
fashion based on the forest $\calF_i$, based on the
{\em height}, $\height^{\calF_i}_j$, of the subtree rooted at node $j \in V^{\calF_i} = [d_i]$
of $\calF_i$.
Recall that $\Prob_x(u \mid I_{i,j}) =  \frac{\Prob_x(u)}{\Prob_x(I_{i,j})}$,
is defined for all $x \in B^{> 0}$,
and by Proposition \ref{prop:expected-poly}
 both the numerator and denominator are given by polynomials in $x$
with ``small'' encoding size (polynomial in $|\calG|$).
Note that likewise, for $a \in \calA_{i,j}$,
$\Prob_{(x \mid \pi^a_{i,j})}(v \mid u)$ is easily defined by 
a weighted monomial
over the variables $x$ whose encoding size is polynomial in $|\calG|$.
Furthermore if the node $j' \in V^{\calF_i}$ is 
a child of the node 
$j \in V^{\calF_i}$
in the forest $\calF_i$, then 
$$\Prob_{(x \mid \pi^a_{i,j})} ( I_{i,j'} \mid I_{i,j}) =   
\sum_{ u \in I_{i,j}} \Prob_x(u \mid I_{i,j}) \cdot
\sum_{v \in I_{i,j'}} \Prob_{(x \mid \pi^a_{i,j})}(v \mid u).$$
Thus $\Prob_{(x \mid \pi^a_{i,j})}(I_{i,j'} \mid I_{i,j})$ is also 
described by a formula over the variables $x$ with
encoding size polynomial in $|\calG|$.
We can now describe a dynamic program for computing 
$\MU^{j,a}_i(x)$, for all $i \in [n]$, $j \in [d_i]$,
and $a \in \calA_{i,j}$:

\begin{equation}
\label{eq:dyn-program}
 \MU^{j,a}_i(x) := \left\{  \begin{array}{ll} 
\sum_{u \in I_{i,j}} \Prob_{x}(u \mid I_{i,j}) \cdot 
\sum_{z \in \LLeaves^{u,a}}   \Prob_{(x \mid \pi^a_{i,j})}(z \mid u) \cdot
r_i(z) \ ,  &    \mbox{if  $j \in \Leaves_{\calF_i}$}\\
(\ \sum_{j' \in \child^a_{\calF_i}(j)}  \Prob_{(x \mid \pi^a_{i,j})}( I_{i,j'} \mid I_{i,j})
\cdot (\max_{a' \in \calA_{i,j'}}  \MU^{j',a'}_i(x)) \ )  +\\    
\quad \quad \sum_{u \in I_{i,j}} \Prob_{x}(u \mid I_{i,j}) \cdot 
\sum_{z \in \LLeaves^{u,a}} 
\Prob_{(x \mid \pi^a_{i,j})}(z \mid u) \cdot
r_i(z)  \ , 
& \mbox{if  $j \in \Internal_{\calF_i}$}
\end{array} \right. 
\end{equation}
It is clear that (\ref{eq:dyn-program}) both defines a
dynamic program for computing
$\MU^{j,a}_i(b)$,  given $b \in B^{> 0}$,
and at the same time  $\MU^{j,a}_i(x)$ defines
a $\{+,-,*,/,\max \}$-formula with variables $x$,
which when evaluated at $b \in B^{> 0}$ yields $\MU^{j,a}_i(b)$.
Furthermore, the encoding size of the formulas $\MU^{j,a}_i(x)$ is
polynomial in $|\calG|$.   This can be seen by noting, firstly,
that all the constituant parts of the inductively defined formula for
$\MU^{j,a}_i(x)$ are given by formulas with 
encoding size polynomial in $|\calG|$,
and furthermore since the inductive definition 
works ``bottom up'' on the forest $\calF_i$, there is no
re-use of subformulas in this inductive definition, i.e., it indeed
defines a formula, not a circuit, and the size of the formula is
polynomial in $|\calG| \times |V^{\calF_i}| \leq |\calG|^2$.
(Later, in Section \ref{sec:delta-almost-epsilon-PE}, for ``almost'' approximation of a QPE,  
 we will also use the fact that the only use of division
gate in this formula is in cases where the denominator evaluates
to $\Prob_b(I_{i,j})$ for some information set $I_{i,j}$.) 
\qed
\end{proof}

Let $h(x) = x+v(x)$,   and let $h'(x) = x + v'(x)$.
For each agent $(i,j)$, and for fixed $x \in B$, consider the function 
$f_{i,j,x}(t) = \sum_{a \in \calA_{i,j}} \max (h_{i,j,a}(x) -t,\epsilon)$.
Likewise, for $x \in B^{> 0}$, 
consider the function $f'_{i,j,x}(t)  = \sum_{a \in \calA_{i,j}} \max (h'_{i,j,a}(x) -t,\epsilon)$.
Clearly, both $f_{i,j,x}(t)$ and $f'_{i,j,x}(t)$ are continuous, piecewise linear function of $t$.
The functions are strictly decreasing as $t$ ranges from $-\infty$,
where $f_{i,j,x}(t) = + \infty$ (respectively, $f'_{i,j,x}(t) = +\infty$), up to 
$\max_{a \in \calA_{i,j}} h_{i,j,a}(x) - \epsilon$  (respectively,
 $\max_{a \in \calA_{i,j}} h'_{i,j,a}(x) - \epsilon$),
where $f_{i,j,x}(t) = |\calA_{i,j}| \cdot \epsilon$  (respectively, $f'_{i,j,x}(t) =
 |\calA_{i,j}| \cdot \epsilon$).
Since we have $|\calA_{i,j}| \cdot \epsilon \leq m \cdot \epsilon < 1$,
there is a unique value of $t$, which depends on $x$, call it $t_{i,j}(x)$
(call it, $t'_{i,j}(x)$, respectively) , where $f_{i,j,x}(t_{i,j}(x)) =1$
(where $f'_{i,j,x}(t'_{i,j}(x)) = 1$).

The functions $F^{\epsilon}_{\calG} : B  \rightarrow B^{\epsilon}$  and 
$H^\epsilon_{\calG}: B \rightarrow B^{\epsilon}$ are defined as 
follows.  First we define $F^{\epsilon}_{\calG}$:
\begin{equation}
\label{eq:fixed-point-functions}
F^{\epsilon}_\calG(x)_{i,j,a} = \max (h_{i,j,a}(x) -t_{i,j}(x),\epsilon) 
\end{equation}
for every $i=1,\ldots,n$, and $j \in [d_i]$, and $a \in \calA_{i,j}$.

To define $H^\epsilon_{\calG}: B \rightarrow B^{\epsilon}$,
care is needed since $v'(x)_{i,j,a}$ is only defined for 
$x \in B^{> 0}$.  To address this, we use an auxiliary normalizing function.
For $\epsilon > 0$,
$\Normalize^{\epsilon}:  B \rightarrow B^{> 0}$, defined as follows:
$$  \Normalize^\epsilon(x)_{i,j,a} =   \frac{\max ( x_{i,j,a}, \epsilon)}{\sum_{a' \in \calA_{i,j}} 
\max ( x_{i,j,a'} , \epsilon)}$$
$\Normalize^{\epsilon}$ clearly does map $B$ to $B^{> 0}$.
Furthermore, importantly, note that for all $b' \in B^{\epsilon}$,   $\Normalize^{\epsilon}(b') = b'$.
We only use $\Normalize^{\epsilon}$ as a tool to ensure the function
$H^{\epsilon}_G$ is defined for all $b \in B$.    
The range, and thus
the fixed points, of $H^{\epsilon}_G$
lies within $B^{\epsilon}$, and on $B^{\epsilon}$  the 
function $\Normalize^{\epsilon}(x)$ 
is the trivial identity function.
We define $H^{\epsilon}_G : B \rightarrow B^{\epsilon}$ as follows:

\begin{equation}
\label{eq:qpe-fixed-point-functions}
H^{\epsilon}_\calG(x)_{i,j,a} = \max (h'_{i,j,a}(\Normalize^{\epsilon}(x)) -t'_{i,j}(\Normalize^{\epsilon}(x)),\epsilon) 
\end{equation}
for every $i=1,\ldots,n$, and $j \in [d_i]$, and $a \in \calA_{i,j}$.

From our choice of  $t_{i,j}(x)$ and $t'_{i,j}(\Normalize^{\epsilon}(x))$, it follows that
$\sum_{a \in \calA_{i,j}} F^{\epsilon}_{\calG}(x)_{i,j,a} =1$
and also that $\sum_{a \in \calA_{i,j}} H^{\epsilon}_{\calG}(x)_{i,j,a} = 1$, 
for all $i \in [n]$ and $j \in [d_i]$.
Thus, for any  behavior profile, $x \in B$, we have 
$F^{\epsilon}_{\calG}(x) \in  B^{\epsilon}$ and $H^{\epsilon}_{\calG}(x) \in B^{\epsilon}$.
So both $F^{\epsilon}_{\calG}$ and $H^{\epsilon}_{\calG}$ indeed map $B$ to $B^{\epsilon}$, and since they are 
clearly also continuous maps, by Brouwer's theorem, they both 
have a fixed point in $B^{\epsilon}$.\footnote{The reason we  specify
the domain of these functions as $B$ instead of $B^{\epsilon}$ is technical.
To place the approximation problems for PE and QPE in $\FIXP_a$, we 
shall need 
make $\epsilon > 0$  {\em very very} small, and we do so by using 
a polynomial sized algebraic circuit to define it.  
However, we shall also need the function domains to be definable by 
linear inequalities having encoding size only polynomial in $|\calG|$.
Both can be achieved by retaining the domain $B$.}

\begin{lemma} 
\label{new-nash}
For $0 < \epsilon < 1/m$:

\begin{enumerate} 
\item Every fixed point of the function 
$F^{\epsilon}_{\calG} : B \rightarrow B^{\epsilon}$
is an $\epsilon$-PE of $\Agent\Normal(\calG)$, and thus also of $\calG$.

\item Every fixed point of the function
$H^{\epsilon}_{\calG} : B \rightarrow B^{\epsilon}$ is a $\epsilon$-QPE of $\calG$.
\end{enumerate}
\end{lemma}
\begin{proof}  The proof is essentially the same in both cases:
\begin{enumerate}
\item If $x$ is a fixed point of $F^{\epsilon}_\calG$, 
then $x \in B^{\epsilon}$ and
$x_{i,j,a}= \max (x_{i,j,a}+ v(x)_{i,j,a} -t_{i,j}(x),\epsilon)$
for all $(i,j,a)$. 
Recall that $v(x)_{i,j,a}= U_i(x \mid \pi^a_{i,j}) = U_{i,j}(x \mid \pi^a_{i,j})$ is the expected payoff for 
agent $(i,j)$ under profile $(x \mid \pi^a_{i,j})$.

Note that
the equation $x_{i,j,a}= \max (x_{i,j,a}+ U_i(x \mid \pi^a_{i,j})-t_{i,j}(x),\epsilon)$ 
implies that
$U_i(x \mid \pi^a_{i,j})= t_{i,j}(x)$ for all $i,j,a$ such that $x_{i,j,a}> \epsilon$,
and  that $U_i(x \mid \pi^a_{i,j}) \leq t_{i,j}(x)$ for all $i,j,a$ such that 
$x_{i,j,a}=\epsilon$.
Consequently, by definition, $x$ constitutes an $\epsilon$-PE.

\item  If $x$ is a fixed point of $H^{\epsilon}_\calG$, 
then $x \in B^{\epsilon}$, and thus $\Normalize^\epsilon(x) = x$.
Thus,  we have will
$x_{i,j,a}= \max (x_{i,j,a}+ v'(x)_{i,j,a} -t'_{i,j}(x),\epsilon)$
for all $(i,j,a)$, where $v'(x)_{i,j,a}=  \MU^{j,a}_i(x)$.

Note, again, that
the equation $x_{i,j,a}= \max (x_{i,j,a}+ \MU^{j,a}_i(x) -t'_{i,j}(x),\epsilon)$ 
implies that
$\MU^{j,a}_i(x) = t'_{i,j}(x)$ for all $i,j,a$ such that 
$x_{i,j,a} > \epsilon$,
and  that $\MU^{j,a}_i(x) \leq t'_{i,j}(x)$ for all $i,j,a$ 
such that 
$x_{i,j,a}= \epsilon$.
Consequently, by definition, $x$ constitutes an $\epsilon$-QPE.
\end{enumerate}

\vspace*{-0.2in}

\qed
\end{proof}

The following Lemma shows that
we can implement the functions $F^{\epsilon}_\calG(x)$ 
and $H^{\epsilon}_{\calG}(x)$ by a
circuit
which has $x$ and $\epsilon$ as inputs, by using
sorting networks.

\begin{lemma}
\label{new-nash-circuit}
Given $\calG$,
we can construct in polynomial time a  $\{+,*,\max\}$-circuit that computes the function $F^{\epsilon}_\calG(x)$,
where $x$ and $\epsilon > 0$ are inputs to the circuit.
Likewise, we can construct in P-time a $\{+,*,/, \max \}$-circuit that
computes the function $H^{\epsilon}_\calG(x)$, where $x$ and $\epsilon > 0$
are inputs to the circuit.
\end{lemma}
\begin{proof}

We define the circuits for both $F^{\epsilon}_\calG(x)$ and
$H^{\epsilon}_{\calG}(x)$ together, since they are defined very similarly.

Given a vector $x \in B$, and $\epsilon > 0$ as inputs,
the respective circuits first compute $y=h(x) = x+v(x)$,
and $y' = \Normalize^{\epsilon}(x) +  v'(\Normalize^{\epsilon}(x))$.
It follows from the definition of $v(x)$, $\Normalize^{\epsilon}(x)$, and 
$v'(x)$, and from Lemma \ref{lem:max-expected-poly}, that both
$y$ and $y'$ can be computed by a circuit 
using $\{+, *, /, \max \}$-gates which has size polynomial in $|\calG|$.
For each agent $(i,j)$, let $y_{i,j}$ be the corresponding subvector of $y$
induced by the (local) strategy of agent $(i,j)$. 
Likewise, let $y'_{i,j}$ be the corresponding subvector of $y'$.
Sort the vector $y_{i,j}$  (the vector $y'_{i,j}$) 
in decreasing order, and let $z_{i,j}$  (respectively, $z'_{i,j}$)
be the  resulting sorted vector, i.e. the components of 
$z_{i,j}=(z_{i,j,a_{1}},\ldots,z_{i,j,a_{|\calA_{i,j}|}})$ are the same
as the components of $y_{i,j}$, but they are sorted
(likewise for $z'_{i,j} = (z'_{i,j,a'_1}, \ldots, z'_{i,j,a'_{|\calA_{i,j}|}})$).
In other words,  
we are assuming 
for convenience that $\calA_{i,j} = \{a_1, \ldots, a_{|\calA_{i,j}|}\}$
and that
$z_{i,j,a_{1}} \geq z_{i,j,a_{2}} \geq \ldots \geq z_{i,j,a_{|\calA_{i,j}|}}$,
and likewise that $\calA_{i,j} = \{ a'_1, \ldots, a'_{|\calA_{i,j}|} \}$
and that $z'_{i,j,a'_{1}} \geq z'_{i,j,a'_{2}} \geq \ldots \geq z'_{i,j,a'_{|\calA_{i,j}|}}$,
To obtain the sorted lists $z_{i,j}$ and $z'_{i,j}$, the respective circuits 
use a 
polynomial sized sorting network, 
for each $(i,j)$ 
(see e.g. Knuth \cite{Knuth} for
background on sorting networks). For each comparator gate
of the sorting network we use a $\max$ and a $\min$ gate.

Using this, for each agent $(i,j)$,
we compute $t_{i,j}(x)$  and $t'_{i,j}(\Normalize^{\epsilon}(x))$ 
as the following expressions: 

\begin{equation}
\label{eq:t-i-j}
t_{i,j}(x) := \max \{ (1/l) \cdot  ( (\sum_{k=1}^l z_{i,j,a_{k}}) +  
(|\calA_{i,j}| - l) \cdot \epsilon - 1 ) \mid  l=1,\cdots,|\calA_{i,j}|\}
\end{equation}

\begin{equation}
\label{eq:t-prime-i-j}
t'_{i,j}(\Normalize^{\epsilon}(x)) := \max \{ (1/l) \cdot  
( (\sum_{k=1}^l z'_{i,j,a'_{k}}) +  
(|\calA_{i,j}| - l) \cdot \epsilon - 1 ) \mid  l=1,\cdots,|\calA_{i,j}|\}
\end{equation}

We will show below that this expression does indeed give the correct value of 
$t_{i,j}(x)$.   The proof for $t'_{i,j}(\Normalize^{\epsilon}(x))$ is
virtually identical, so we omit it.

We output 
$F^{\epsilon}_{\calG}(x)_{i,j,a} = \max(y_{i,j,a}-t_{i,j}(x),\epsilon)$,
and $H^{\epsilon}_{\calG}(x)_{i,j,a} =  \max(y'_{i,j,a}-t'_{i,j}(\Normalize^{\epsilon}(x)),\epsilon)$,
 for 
each $i=1,\ldots,n$, $j\in [d_i]$,
and $a \in \calA_{i,j}$.

We now have to establish 
that $t_{i,j}(x)$, defined above, is the correct value.
(Again, we forgo the proof for $t'_{i,j}(\Normalize^{\epsilon}(x))$,
which is virtually identical.)
Consider the function $f_{i,j,x}(t) = \sum_{a \in \calA_{i,j}} \max (z_{i,j,a} -t,\epsilon)$ as $t$ decreases
from $z_{i,j,a_{1}}-\epsilon$ where the function 
value is at its minimum of $|\calA_{i,j}| \cdot \epsilon$, down until the function reaches the value $1$.
In the first interval from $z_{i,j,a_{1}}-\epsilon$ to $z_{i,j,a_{2}}-\epsilon$ the function is
$f_{i,j,x}(t) = z_{i,j,a_{1}}-t + (|\calA_{i,j}|-1)\cdot\epsilon$; in the second 
interval from  $z_{i,j,a_{2}}-\epsilon$ to 
$ z_{i,j,a_{3}}-\epsilon$ it is
$f_{i,j,x}(t) = z_{i,j,a_{1}}+z_{i,j,a_{2}}- 2t + (|\calA_{i,j}|-2)\cdot\epsilon$, and so forth.
In general, in the $l$-th interval, 
$f_{i,j,x}(t) = \sum_{k=1}^l (z_{i,j,a_{k}}- t) + (|\calA_{i,j}| - l) \cdot \epsilon = \sum_{k=1}^l z_{i,j,a_{k}}- lt + 
(|\calA_{i,j}| - l) \cdot \epsilon$.
If the function reaches the value 1 in the $l$'th interval, 
then clearly $t_{i,j}(x) =  ((\sum_{k=1}^l z_{i,j,a_{k}})  + (|\calA_{i,j}| - l) \cdot \epsilon -1)/l$.

In that case, furthermore
for $k'<l$, we have $\sum_{k=1}^{k'} (z_{i,j,a_{k}} -t_i) + (|\calA_{i,j}| - k')\cdot \epsilon 
\leq \sum_{k=1}^l (z_{i,j,a_{k}} -t_{i,j}(x)) + (|\calA_{i,j}| - l)\cdot\epsilon = 1$, 
because in that case we know $(z_{i,j,a_{k}}-t_{i,j}(x)) \geq \epsilon$ for every $a \in \{1,\ldots,l\}$.
Therefore, in this case
$((\sum_{k=1}^{k'} z_{i,j,a_{k}}) + (|\calA_{i,j}|- k')\cdot \epsilon -1)/k' \leq t_{i,j}(x)$.
On the other hand, if $l < |\calA_{i,j}|$, then for $k' > l$ we have $t_i \geq z_{i,j,a_{{k'}}}-\epsilon$, 
i.e., $z_{i,j,a_{{k'}}}- t_i \leq \epsilon$, and thus for all $k' >l$, $k' \leq |\calA_{i,j}|$,
we have
$\sum_{k=1}^{k'} (z_{i,j,a_{k}} -t_{i,j}(x)) + (|\calA_{i,j}| - k')\cdot\epsilon 
\leq \sum_{k=1}^l (z_{i,j,a_{k}} -t_{i,j}(x)) + (|\calA_{i,j}| - l)\cdot\epsilon = 1$.
Thus again $((\sum_{k=1}^{k'} z_{i,j,a_{k}})+ (|\calA_{i,j}| -k')\cdot \epsilon -1)/k' \leq t_{i,j}(x)$.
Therefore, $t_{i,j}(x)= \max \{ (1/l) \cdot  ( (\sum_{k=1}^l z_{i,j,a_{k}}) + (|\calA_{i,j}|-l)\cdot \epsilon -1) | l=1,\cdots,|\calA_{i,j}|\}$. 
\qed
\end{proof}

Lemma \ref{new-nash} and Lemma \ref{new-nash-circuit} together immediately imply Theorem \ref{fixp-no-division}.

\section{Approximating an SE,  PE, and QPE is $\FIXPA$-complete}

\label{approx-of-PE} 

In this section we exploit the algebraically defined function $F^{\epsilon}_{\calG}(x)$
and $H^{\epsilon}_{\calG}(x)$
for a EFGPR, $\calG$,  with input parameter $\epsilon > 0$,
devised in the previous section for $\epsilon$-PEs and $\epsilon$-QPEs, 
and we construct a 
``small enough''  $\epsilon^* > 0$ (using an algebraic circuit, given $\delta > 0$) 
such that any fixed point of $F^{\epsilon^*}_{\calG}(x)$ is
a $\epsilon^*$-PE which is also
$\delta$-close to an actual PE of $\calG$ (in $\ell_\infty$ distance),
and likewise any fixed point of $H^{\epsilon^*}_\calG(x)$ is a $\epsilon^*$-QPE
which is also $\delta$-close to an actual QPE.
In this way, we show that approximating a PE, and a QPE, 
to within given desired precision, $\delta > 0$,  for a given EFGPR is $\FIXPA$-complete.
Since PE  constitutes a refinement of NE and of
SGPE, this of course immediately implies
that approximating a NE or SGPE is also 
$\FIXPA$-complete  (cf. \cite{DFP06}).
Likewise, since QPE constitutes a refinement of NF-PE, this
also implies that approximating a NF-PE is $\FIXPA$-complete.

For SEs, we then also show that for any such $\epsilon^*$-PE, $b''$,  
if $\mu^{b''}$ is
the unique belief system generated by $b''$  then $(b'',\mu^{b''})$ is
$\delta$-close to an actual SE of 
$\calG$  (again in $\ell_\infty$).
Furthermore, using $F^{\epsilon^*}_\calG(x)$,
we define an auxiliary fixed point function 
$G^{\epsilon^*}_\calG(x,z)$ with domain $B \times \SysB$,
such that the Brouwer fixed points 
of $G^{\epsilon^*}_\calG$ are pairs $(b'',\mu^{b''})$,
where $b''$ is a $\epsilon^*$-PE and $\mu^{b''}$ is the belief
system that it generates.    In this way, we show that 
 approximating a SE (including its belief
system)
to within given desired precision $\delta > 0$, for a given EFGPR, is also $\FIXPA$-complete.

\begin{theorem}
\label{thm:main-fixpa-result}
Given as input a EFGPR, $\calG$, and a rational $\delta > 0$:
\begin{enumerate}
\item The problem of computing 
a vector $b' \in B$ such that there is a PE (or NE or SGPE), $b^*$, of 
$\calG$, with $\|b' - b^* \|_\infty < \delta$,
is $\FIXPA$-complete.

\item The problem of computing 
a vector $b' \in B$ such that there is a QPE (or NF-PE), $b^*$, of 
$\calG$, with $\|b' - b^* \|_\infty < \delta$,
is $\FIXPA$-complete.

\item The problem of computing a vector $b' \in B$ and a belief system $\mu'$ 
such that
there is a SE, $(b^*, \mu^*)$ of $\calG$, with $\|(b',\mu') - (b^*,\mu^*)\|_\infty < \delta$,
is $\FIXPA$-complete. 
\end{enumerate}
\end{theorem}

Note that $\FIXPA$-hardness for these problems follows from the fact that we can encode
any NFG, $\Gamma$, as an EFGPR, $\calE(\Gamma)$, with not much larger 
encoding size, and 
from the fact that approximating a NE
within desired precision for $n$-player NFGs is $\FIXPA$-hard, as shown in \cite{EY07}.
The $\FIXPA$-hardness of approximating a SPGE, PE, QPE, NF-PE, and SE,
then follows because we know that these constitute
refinements of NE.
Thus, we only need to prove containment in $\FIXPA$.
Our proofs follow closely some of the proofs in
\cite{EHMS14} used for characterizing 
the complexity of approximating a PE for NFGs.  Although very similar,
our proof differs in some details (especially for sequential equilibrium).
So, both for clarity and in order to 
be self-contained,
we provide detailed proofs.

\newcommand{\transpose}{\ensuremath{\mathsf{T}}}
\newcommand{\abs}[1]{\ensuremath{\mathopen\lvert #1 \mathclose\rvert}}
\newcommand{\norm}[1]{\ensuremath{\mathopen\lVert #1 \mathclose\rVert}}
\newcommand{\Abs}[1]{\ensuremath{\left| #1 \right|}}
\newcommand{\Norm}[1]{\ensuremath{\left\| #1 \right\|}}
\newcommand{\RR}{\mathbb{R}}

\newcommand{\epsPE}{\ensuremath{\operatorname{EPS-PE}}}
\newcommand{\epsPEBS}{\ensuremath{\operatorname{EPS-PE-BS}}}
\newcommand{\epsQPE}{\ensuremath{\operatorname{EPS-QPE}}}
\newcommand{\PE}{\ensuremath{\operatorname{PE}}}
\newcommand{\QPE}{\ensuremath{\operatorname{QPE}}}
\newcommand{\PESE}{\ensuremath{\operatorname{PE-SE}}}
\newcommand{\PEbound}{\ensuremath{\operatorname{PE-bound}_\delta}}
\newcommand{\QPEbound}{\ensuremath{\operatorname{QPE-bound}_\delta}}
\newcommand{\PESEbound}{\ensuremath{\operatorname{PE-SE-bound}_\delta}}

Before we prove Theorem \ref{thm:main-fixpa-result}, we need some Lemmas.
The following is a special case of a general paradigm 
noted by
Anderson~\cite{TAMS:Anderson86}.
\begin{lemma}
\label{LEM:AlmostNear}
  For any fixed EFGPR, $\calG$, and any $\delta > 0$, there is an
  $\epsilon > 0$, so that any $\epsilon$-(Q)PE, $b'$, of $\calG$
  has $\ell_\infty$-distance at most $\delta$ from some (Q)PE of $\calG$,
 and furthermore, if $\mu^{b'}$ denotes the 
belief system generated by $b'$, then $(b',\mu^{b'})$ has
$\ell_\infty$-distance  at most $\delta$ from some SE of $\calG$.
\end{lemma}
\begin{proof}
Assume to the contrary that there is a EFGPR, $\calG$, and a $\delta > 0$
so that for all $\epsilon > 0$, 
there is an $\epsilon$-(Q)PE, $b^\epsilon$ 
 of $\calG$ so that there is 
no (Q)PE in the $\delta$-neighborhood (with respect to $\ell_\infty$) of $b^\epsilon$
or that there is no
SE in the $\delta$-neighborhood (with respect to $\ell_\infty$)
of $(b^\epsilon,\mu^{b^\epsilon})$, where $\mu^{b^\epsilon}$
is the belief system generated by $b^{\epsilon}$. 

Consider the sequence of assessments $(b^{1/n},\mu^{b^{1/n}})_{n \in \nat}$.
Since this is a sequence in a compact space (namely, the direct
product of the space of behavior profiles and the space of belief systems),
it has a limit point $(b^*,\mu^*)$.  But then $b^*$ is
a (Q)PE of $\calG$, by definition, since each $b^{1/n}$
is a $1/n$-(Q)PE. 
But this contradicts the statement that
  there is no (Q)PE in a $\delta$-neighborhood of any of
  the behavior profiles $b^{1/n}$.    
Furthermore, it follows from Proposition \ref{prop:kreps-wilson}
(Part 3.) that  $(b^*, \mu^*)$ is a SE.
But this contradicts the statement that there is no SE
in a $\delta$-neighborhood of any of the assessments $(b^{1/n},\mu^{b^{1/n}})$.
\qed
\end{proof}

A priori, we have no bound on $\epsilon$, but we can 
use results in real algebraic geometry 
\cite{BasuPollackRoy2006,BasuPollackRoy2011} 
to obtain 
a specific bound.   We first do this for PE and SE: 
\begin{lemma}
\label{LEM:AlmostVeryNear}
There is a constant $c$, so that for all integers $n,m,k,M \in \nat$ and $\delta \in \rat_+$, 
the following holds. Let $\epsilon \leq \min(\delta,1/(M^{\height^\calG + 1}))^{m^{c m^3}}$. For any $n$-player 
EFGPR, $\calG$, with a combined total of $m$ pure local strategies for all players in the game,
with game tree $T$ having height $\height^\calG$, 
and with $M$ a positive integer
which is at least as large as any (by assumption, necessarily positive) integer payoff of $\calG$
and such that $p_u(a) > 1/M$, for every $u \in P_0$ and every $a \in \Act(u)$. 
Then any $\epsilon$-PE,  $b^\epsilon$, 
of $\calG$ has $\ell_\infty$-distance at most $\delta$ from some PE of $\calG$, and
furthermore  if $\mu^{b^\epsilon}$ is the belief system
generated by $b^\epsilon$, then $(b^\epsilon,\mu^{b^\epsilon})$ has
$\ell_\infty$-distance at most $\delta$ from some SE of $\calG$.\end{lemma}
\begin{proof}  
The proof involves constructing formulas in the first order theory of real numbers, which 
formalize the statement of Lemma \ref{LEM:AlmostNear}, 
with $\delta$ being ``hardwired" as a constant and $\epsilon$ being the only free variable. 
Then, we apply {\em quantifier elimination} to these formulas. This leads to a quantifier free statement to which we can apply standard theorems bounding the size of an instantiation of the free variable $\epsilon$ making the formula true. We shall apply and refer to theorems in the monograph of Basu, Pollack and Roy \cite{BasuPollackRoy2006,BasuPollackRoy2011}. Note that we specifically refer to theorems and page numbers of the online edition \cite{BasuPollackRoy2011}; these are in general different from the printed edition \cite{BasuPollackRoy2006}.

\paragraph{First-order formula for an extensive form $\epsilon$-perfect equilibrium and for the belief system it
generates:} 

Let $\epsPEBS(x,z,\epsilon)$ be the quantifier-free first-order formula, with free
variables $x \in \RR^m$, $z \in \RR^{|\Internal \setminus P_0|}$, and $\epsilon\in\RR$, defined by the conjunction
of the following formulas, which together express the fact that $x$ is
a behavior profile that is
an extensive form $\epsilon$-PE of the given EFGPR, $\calG$,
and that $z$ is  the (unique) belief system generated by $x$:
\begin{gather*}
x_{i,j,a} > 0,  \quad \text{for } i \in [n] , j \in [d_i] \text{, and }  a \in \calA_{i,j} \enspace ,\\
\sum_{a \in \calA_{i,j}} x_{i,j,a} = 1, \quad \text{for } i \in [n] \text{ and }  j \in [d_i] \ \enspace ,\\
\left(U_i(x \mid \pi^a_{i,j}) \geq U_i(x  \mid  \pi^{a'}_{i,j}) \right) \vee \left(x_{i,j,a} \leq \epsilon \right) ,
\quad \text{for } i \in [n] \text{, } j \in [d_i], \text{ and }   a, a' \in \calA_{i,j} \enspace ,\\
z_{u} \cdot \Prob_x(I_{i,j}) =  \Prob_x(u) ,   \quad \text{for all } u \in V \text{ where } u \in I_{i,j}
\text{ for } i \in [n]  \text{ and }  j \in [d_i] .  
 \enspace 
\end{gather*}

\noindent Note that by Proposition \ref{prop:expected-poly},
$\Prob_x(I_{i,j})$ and $\Prob_x(u)$ are expressible as multilinear polynomials
in the variables $x$ 
(whose encoding size is polynomial in $|\calG|$).

\paragraph{First-order formula for  perfect equilibrium and sequential equilibrium:}
Let $\PESE(x,z)$ denote the following first-order formula with free
variables $x\in \RR^m$, and $z \in \RR^{|\Internal \setminus P_0|}$, 
expressing that $x$ is a behavior profile this a PE of $\calG$,
and that $z$ is a belief system such that $(x,z)$ is a SE
of $\calG$:
\begin{gather*}
\forall \epsilon >0 \: \exists x' \in \RR^m  \: \exists z' \in \RR^{|\Internal \setminus P_0|}
: \epsPEBS(x',z',\epsilon) \wedge \norm{x-x'}^2 < \epsilon \wedge \norm{z - z'}^2 < \epsilon\enspace .
\end{gather*}

\paragraph{First-order formula for ``almost implies near'' statement:}

Given a {\em fixed} $\delta >0$ let $\PESEbound(\epsilon)$ denote the following
first-order formula with free variable $\epsilon\in\RR$, denoting that
any $\epsilon$-perfect equilibrium, $x$, of $\calG$ is $\delta$-close to a
PE (in $\ell_2$-distance, and therefore also in $\ell_\infty$-distance),
and likewise that if $z$ is the belief system generated by $x$, then
$(x,z)$ is $\delta$-close to a SE:
\begin{gather*}
\forall x \in \RR^m  \: \forall z \in \RR^{|\Internal \setminus P_0|} \: \; \exists x^* \in \RR^m \:
\exists z^* \in \RR^{|\Internal \setminus P_0|}: \\
(\epsilon>0) \wedge \left(\neg \epsPEBS(x,z,\epsilon) \vee 
\left(\PESE(x^*,z^*) \wedge \norm{x-x^*}^2 < \delta^2 
\wedge \norm{z - z^*}^2 < \delta^2 \right)\right) \enspace .
\end{gather*}

Suppose $\delta^2 = 2^{-k}$ and that $M=2^\tau$
is a positive integer that satisfies 
the conditions in the statement of the Lemma. Then for this formula we have
\begin{itemize}
\item The total degree of all involved polynomials is at most $\max(2,m)$.
\item The bitsize of coefficients is at most $\max(k,\tau \cdot (\height^\calG + 1))$.
\item The number of free variables is $1$.
\item Since $|\Internal \setminus P_0| \leq m$, 
converting to prenex normal form, the formula has 4 blocks of
  quantifiers, of sizes at most $2m$, $2m$, $1$, $2m$, respectively.
\end{itemize}
 
We now apply quantifier elimination
\cite[Algorithm 14.6, page 555]{BasuPollackRoy2011} to the formula $\PESEbound(\epsilon)$, converting it into an
equivalent quantifier free formula $\PESEbound'(\epsilon)$ with a single free variable $\epsilon$. This is simply a Boolean formula whose atoms are sign conditions on various polynomials in $\epsilon$. The bounds given by 
\cite{BasuPollackRoy2011}
in association with Algorithm 14.6 imply that for this formula:
\begin{itemize}
\item The degree of all involved polynomials (which are univariate polynomials in $\epsilon$) is:\\ \centerline{$\max(2,m)^{O(m^3)} 
= m^{O(m^3)}$.}
\item The bitsize of all coefficients is at most:\\ 
\centerline{$\max(k,\tau \cdot 
(\height^\calG+1))\max(2,m)^{O(m^3)} =
  \max(k,\tau \cdot (\height^\calG + 1))m^{O(m^3)}$.}
\end{itemize}

By Lemma~\ref{LEM:AlmostNear}, we know that there exists an
$\epsilon>0$ so that the formula $\PESEbound'(\epsilon)$ is true. We
now apply Theorem 13.14 of \cite[Page 521]{BasuPollackRoy2011} to the
set of polynomials that are atoms of $\PESEbound'(\epsilon)$ and
conclude that $\PESEbound'(\epsilon^*)$ is true for some $\epsilon^*
\geq 2^{-\max(k,\tau \cdot (\height^\calG+1))m^{\Omega(m^3)}}$. By the
semantics of the formula $\PESEbound(\epsilon)$, we also have that
$\PESEbound(\epsilon')$ is true for all $\epsilon' \leq \epsilon^*$,
and the statement of the lemma follows. \qed
\end{proof}

\noindent {\bf Proof of Theorem \ref{thm:main-fixpa-result}, parts (1.) and (3.)}.
We shall combine the proofs of parts (1.) and (3.) of the Theorem together.
To do so, we shall first define an auxiliary fixed point function
$G^{\epsilon}_\calG(x,z)$ defined in terms of
$F^{\epsilon}_\calG(x)$,  such that the Brouwer fixed points 
of $G^{\epsilon}_\calG$ are pairs $(b'',\mu^{b''})$,
where $b''$ is a $\epsilon$-PE and $\mu^{b''}$ is the belief
system that it generates.  
Specifically,  we define $G^{\epsilon}_\calG : B \times \SysB  \rightarrow   B^\epsilon \times \SysB$ as follows:
For all $(b,z) \in B \times \SysB$,  $G^{\epsilon}_\calG(b,z) := (b',z')$ where $b'_{i,j,a} := F^{\epsilon}_{\calG}(b)$,
for all $i \in [n]$, $j \in [d_i]$ and $a \in \calA_{i,j}$;
and furthermore where $z'_u :=  \frac{\Prob_{b'}(u)}{\Prob_{b'}(I_{i_u,j_u})}$ 
for all $u \in \Internal \setminus P_0$,
and where $u \in I_{i_u,j_u}$.   Note in particular that,
for all $u \in \Internal \setminus P_0$, we can express $z'_u$ as a (efficiently algebraically encodable) rational function of $b$
because, recalling from Proposition \ref{prop:expected-poly} that for all 
$V' \subseteq V$,
there is a efficiently encodable polynomial $F_{V'}(x)$ such that for all $b \in B$ 
$F_{V'}(b) = \Prob_b(V')$ represents the realization probability of $V'$,  we have 
 $z'_u := \frac{\Prob_{b'}(u)}{\Prob_{b'}(I_{i_u,j_u})} =  \frac{F_u(F^{\epsilon}_{\calG}(b))}
{F_{I_{i_u,j_u}}(F^{\epsilon}_{\calG}(b))}$ .

Thus $G^{\epsilon}_{\calG}:  B \times \SysB  \rightarrow 
B^\epsilon \times \SysB$ is a continuous map, and notably $G^{\epsilon}_{\calG}$ is defined in the entire compact
domain $B \times \SysB$, because $b' := F^{\epsilon}_{\calG}(b) \in B^{\epsilon}$ 
and thus
the ratio $\frac{\Prob_{b'}(u)}{\Prob_{b'}(I_{i_u,j_u})}$ is always well defined (we never divide
by $0$, because all nodes have positive realization
probability under a profile $b' \in B^\epsilon$, for all $\epsilon > 0$).
Moreover, by definition
of $G^{\epsilon}_{\calG}$, for all $\epsilon > 0$,  for any Brouwer fixed point $(b'',\mu'') \in B^\epsilon \times \SysB$ of 
$G^{\epsilon}_{\calG}$,  $b''$ must be a $\epsilon$-PE
of $\calG$ and $\mu''$ must be the unique belief system $\mu^{b''}$ generated
by $b''$.

We now prove that computing a PE to within
desired precision is $\FIXPA$-complete, and 
that computing a SE to within desired precision is
$\FIXPA$-complete.
Let $\calG$ be the $n$-player EFGPR given as input. Let $m$ be the combined total number of pure strategies for 
all players. 
Let $M'$ be the minimum positive integer 
such that $p_u(a) > 1/M'$, for every $u \in P_0$ and every $a \in \Act(u)$.
Let $M \in \nat$ be a positive integer
which is the maximum of $M'$ and any (by assumption, necessarily positive) integer payoff of $\calG$.
By the definition of $\FIXPA$, our task is the following. 
Given a parameter $\delta > 0$, we must construct a polytope $P$, 
a circuit $C: P \rightarrow P$, and a number $\delta'$,  so that a
$\delta'$-approximation to a fixed point of $C$ can be 
efficiently transformed into $\delta$-approximation of 
a PE of $\calG$, and a $\delta'$-approximation of a fixed point
of $C$ can also be efficiently transformed into a $\delta$-approximation
of a SE of $\calG$. In fact, we shall let $\delta' = \delta/2$ and 
ensure that $\delta'$-approximations to fixed points of $C$ yield both 
a $\delta$-approximation of a PE and a $\delta$-approximation of a 
SE of $\calG$. 
The polytope $P$ is simply the polytope $B \times \SysB$, i.e., the 
cartesian product of the space of 
behavior profiles of $\calG$ and the space of belief systems; clearly we can 
output the inequalities defining this polytope in polynomial time. 
The circuit $C$ is the following: We construct the circuit for the function 
$G^\epsilon_\calG$ above. 
Then, we construct a circuit for the number $\epsilon^* = 
\min(\delta/2, M^{-h^{\calG}})^{{2^{\lceil c m^3 \lg m \rceil }}} \leq \min(\delta/2,M^{-h^{\calG}})^{m^{c m^3}}$, where $c$ is the 
constant of Lemma \ref{LEM:AlmostVeryNear}: The circuit simply repeatedly 
squares the number $\min(\delta/2, M^{-h^{\calG}})$ (which is a rational constant
that can be computed in P-time given the input $\calG$) 
and thereby consists of exactly $\lceil c m^3 \lg m \rceil$ multiplication 
gates, 
i.e., a polynomially bounded number. We then plug in the circuit for 
$\epsilon^*$ for the parameter $\epsilon$ in the circuit for $G^\epsilon_\Gamma$, 
obtaining the circuit $C$, which is obviously a circuit for $G^{\epsilon^*}_\Gamma$. Now, by the above, any fixed point $(b'',\mu'')$ of $C$ on $P$ is an 
$\epsilon^*$-PE of $\calG$. Therefore, by Lemma \ref{LEM:AlmostVeryNear}, 
in any fixed point $(b'',\mu'')$ of $C$, we know that $b''$ is 
both a $\epsilon^*$-PE and
a $\delta/2$-approximation (in 
$\ell_\infty$-distance) to a PE $b^*$ of $\calG$, and furthermore that $\mu''$ is
the  unique belief system generated by $b''$, and that $\mu''$ 
is a $\delta/2$-approximation (in $\ell_\infty$-distance) of a belief system 
$\mu^*$
such that $(b^*,\mu^*)$ is a SE of $\calG$. Finally, by the triangle 
inequality, any $\delta'=\delta/2$-approximation $(b',\mu')$ to a fixed point $(b'',\mu'')$ of $C$ on 
$P$ is a $\delta/2 + \delta/2 
= \delta$ approximation (in $\ell_\infty$) of some pair $(b^*,\mu^*)$,
such that $b^*$ is a PE of $\calG$ and $(b^*,\mu^*)$ is a SE 
 of $\calG$.   We have thus established Theorem \ref{thm:main-fixpa-result}, parts (1.) and (3.).
\qed

\noindent Next, we want to prove something analogous to Lemma \ref{LEM:AlmostVeryNear}, but for QPEs.
In order to do so, we first need the following:

\begin{proposition}
\label{prop:prop-formula-MU}
For any EFGPR, $\calG$, with $i \in [n]$, $j \in [d_i]$, and 
any $a , a' \in \calA_{i,j}$, 
the inequality $\MU^{j,a}_{i}(x) < \MU^{j,a'}_i(x)$  
can be expressed as formula, $\Phi^{i,j,a,a'}_\calG(x) \equiv \exists y 
\Psi^{i,j,a,a'}_\calG(y, x)$, in the existential theory of reals,
where $\Psi^{i,j,a,a'}_\calG(y, x)$ is quantifier free, 
where the total degree of all polynomials involved in
$\Psi^{i,j,a,a'}_\calG(y, x)$ is 2,
where the
encoding size of $\Phi^{i,j,a}_{\calG}(x)$ is polynomial in $|\calG|$, 
and such that for 
all $b \in B^{> 0}$, $\Phi^{i,j,a,a'}_\calG(b)$ holds true
iff $\MU^{j,a}_{i}(b) < \MU^{j,a'}_i(b)$.
\end{proposition}
\begin{proof}
Note that $\MU^{j,a}_i(x)
< \MU^{j,a}_i(x)$ is an inequality between two
$\{+,-,*,/,\max \}$-formulas (over the variables $x$)
of encoding size polynomial in $|\calG|$.  
We will show that any such inequality, over any subset of
Euclidean space where the formula is always well-defined (i.e., involves no
division by $0$),  can be expressed by an existential
theory of reals formula whose encoding size is polynomial in the original
inequality (and thus polynomial in $|\calG|$).

Specifically, suppose $x$ is an $m$-vector of variables.
By induction on the depth of any $\{+,-,*,/,\max \}$-formula,
$\zeta(x)$, which is well-defined over the domain $B^{> 0}$
(i.e., which involves no sub formula that performs a 
 division by $0$, when $x$ is anywhere in that domain),
we prove that there is a existential theory
of reals formula 
$\Psi_{\zeta}(y_0, y, x)$, of size linear in the size of $\zeta$, 
with auxiliary variable $y_0$ and a vector of auxiliary variables $y$,
such that  for all $x \in B^{> 0}$, 
$\{ y_0 \in \real \mid \exists y  \Psi_{\zeta}(y_0, y, x) \} = 
\{ \zeta(x) \}$.
In other words, for the values $x$ in the domain $B^{> 0}$,
the formula $\exists y  \Psi_{\zeta}(y_0, y, x)$
``expresses'' a unique value, $y_0 \in \real$, 
which is the same value as $\zeta(x)$.

The base case, when $\zeta(x)$ is a variable from $x$,
or a rational
constant, is trivial.

Inductively, suppose $\zeta(x) :=  \zeta_1(x) \odot
\zeta_2(x)$,   where $\odot \in \{ + , -, *, /, \max \}$.
By the inductive hypothesis,  there is a formula  
$\exists y \Psi_{\zeta_1}(y_0,y, x)$  
using which $y_0$ expresses
$\zeta_1(x)$, and which has size linear in that of $\zeta_1$, and likewise 
there is a formula 
$\exists  y' \Psi_{\zeta_2}(y'_0,y', x)$ using which $y'_0$
expresses 
$\zeta_2(x)$, and which has size linear in that of $\zeta_2$.

We construct a new formula 
$\exists y_0, y'_0, y, y' \Psi_{zeta}(y''_0,y_0,y'_0,y,y',x)$
as follows.  If $\odot \in \{+,*, -\}$, then
$\Psi_{\zeta}(y''_0,y_0, y'_0, y,y',x) := 
(y''_0 =  y_0 \odot y'_0   \wedge \Psi_{\zeta_1}(y_0,y,x) 
\wedge \Psi_{\zeta_2}(y'_0,y',x))$.

If $\odot  \doteq / $,  then
$\Psi_{\zeta}(y''_0,y_0, y'_0, y,y',x) := 
(y''_0 * y'_0 =  y_0   \wedge \Psi_{\zeta_1}(y_0,y,x) 
\wedge \Psi_{\zeta_2}(y'_0,y',x))$.

If $\odot  \doteq  \max $, then 
$\Psi_{\zeta}(y''_0,y_0, y'_0, y,y',x) := 
( y''_0 \geq y_0 \wedge y''_0 \geq y_0 \wedge (y''_0 \leq y_0 \vee 
y''_0 \leq y'_0) \wedge  \Psi_{\zeta_1}(y_0,y,x) 
\wedge \Psi_{\zeta_2}(y'_0,y',x))$.
(The case with $\odot  \doteq \min $ is entirely similar
and symmetric
 to the $\max$ case. )

Note that, by induction, the new formula 
$\exists y_0, y'_0, y, y' \Psi_{\zeta}(y''_0,y_0,y'_0,y,y',x)$
again has encoding size linear in the encoding size of $\zeta(x)$,
and furthermore note that the total degree of all polynomials
in $\Psi_{\zeta}(y''_0,y_0,y'_0,y,y',x)$ remains $2$.

Finally, for $x$ in the domain $B^{> 0}$, let $\MU^{j,a}_i(x)$ be
expressed by  $\exists y \Psi_{\MU^{j,a}_i}(y_0,y,x)$,
and let $\MU^{j,a'}_i(x)$ be
expressed by $\exists y' \Psi_{\MU^{j,a'}_i}(y'_0,y',x)$.
We can express the inequality 
$\MU^{j,a}_i(x) < \MU^{j,a'}_i(x)$
using the following existential theory of reals formula:
$$\Phi^{i,j,a,a'}_\calG(x)  := 
\exists y_0 , y'_0 , y, y'   \ ( \ y_0 < y'_0  \wedge
\Psi_{\MU^{j,a}_i}(y_0,y,x) \wedge 
\Psi_{\MU^{j,a'}_i}(y'_0,y',x) \ ) .$$
 \qed
\end{proof}

\begin{lemma}
\label{LEM:AlmostVeryNear-QPE}
There is a polynomial $q(\cdot)$, such that, 
for any EFGPR, $\calG$, and any $\delta = 2^{-k} > 0$,
where $k$ is a positive integer,
for any $\epsilon \leq  \frac{1}{2^{2^{q(|\calG| + k)}}}$,
any $\epsilon$-QPE of $\calG$ is $\delta$-close (in $\ell_\infty$)
to a QPE.
\end{lemma}

\begin{proof}  
The proof is entirely analogous to that of Lemma \ref{LEM:AlmostVeryNear}.
We spell out the details for completeness.

\paragraph{First-order formula for $\epsilon$-quasi-perfect equilibrium:} 

Let $\epsQPE(x,\epsilon)$ be the first-order formula (a
{\em universal} formula in the theory of reals), with free
variables $x \in \RR^m$ and  $\epsilon \in \RR$, defined by the conjunction
of the following formulas, which together express the fact 
that $x \in B^{> 0}$ 
is
a behavior profile that is
an extensive form $\epsilon$-QPE of the given EFGPR, $\calG$:
\begin{gather*}
x_{i,j,a} > 0,  \quad \text{for } i \in [n] , j \in [d_i] \text{, and }  a \in \calA_{i,j} \enspace ,\\
\sum_{a \in \calA_{i,j}} x_{i,j,a} = 1, \quad \text{for } i \in [n] \text{ and }  j \in [d_i] \ \enspace ,\\
(\neg \Phi^{i,j,a,a'}_i(x)) \vee \left(x_{i,j,a} \leq \epsilon \right) ,
\quad \text{for } i \in [n] \text{, } j \in [d_i], \text{ and }   a, a' \in \calA_{i,j} \enspace .
 \enspace 
\end{gather*}

\noindent Note that by Proposition 
\ref{prop:prop-formula-MU},
$\Phi^{i,j,a,a'}_i(x)$
is expressible as a 
existential formula in the theory of reals, whose
size is polynomial in $|\calG|$.  Thus, the conjunction
$\epsQPE(x,\epsilon)$ of all of the above formulas
is expressible as a universal formula in the theory of reals. 

\paragraph{First-order formula for  quasi-perfect equilibrium:}
Let $\QPE(x)$ denote the following first-order formula with free
variables $x\in \RR^m$, 
expressing that $x$ is a behavior profile that is  a QPE of $\calG$:
\begin{gather*}
\forall \epsilon >0 \: \exists x' \in \RR^m
: \epsQPE(x',\epsilon) \wedge \norm{x-x'}^2 < \epsilon \enspace .
\end{gather*}

\paragraph{First-order formula for ``almost implies near'' statement:}

Given a fixed $\delta > 0$, 
let $\QPEbound(\epsilon)$ denote the following
first-order formula with free variable $\epsilon\in\RR$, denoting that
any $\epsilon$-quasi-perfect equilibrium, $x$, of $\calG$ is $\delta$-close to a
QPE:
\begin{gather*}
\forall x \in \RR^m  \; \exists x^* \in \RR^m : \\
(\epsilon>0) \wedge \left(\neg \epsQPE(x,\epsilon) \vee 
\left(\QPE(x^*) \wedge \norm{x-x^*}^2 < \delta^2 \right)\right) \enspace .
\end{gather*}

Suppose $\delta^2 = 2^{-k}$, for some positive integer $k$, and let 
$q'(\cdot)$ be some fixed polynomial such that
$\tau = q'(|\calG|) + k$ 
is at least the maximum encoding size of any coefficient in 
any of the polynomials involved in   $\QPEbound(\epsilon)$.
(We know that such an explicit polynomial $q'(\cdot)$  exists, given  
 the polynomial bounds as a function of $\calG$ on the encoding size of 
 the various parts of the
 formula $\QPEbound(\epsilon)$.)   

\begin{itemize}
\item The total degree of all involved polynomials is at most $2$.
\item The bitsize of coefficients is at most $\tau$.
\item The number of free variables is $1$.
\item 
Converting to prenex normal form, the formula has 5 blocks of
  quantifiers, of sizes at most $m$, $m$, $1$, $m$, 
and $q''(|\calG|)$, for some fixed polynomial $q''(\cdot)$, respectively.
\end{itemize}
 
We now apply quantifier elimination
\cite[Algorithm 14.6, page 555]{BasuPollackRoy2011} to the formula $\QPEbound(\epsilon)$, converting it into an
equivalent quantifier free formula $\QPEbound'(\epsilon)$ with a single free variable $\epsilon$. 
This yields Boolean formula whose atoms are sign conditions on various polynomials in $\epsilon$. 
Since $m \leq | \calG|$,
the bounds given by 
\cite{BasuPollackRoy2011}
in association with Algorithm 14.6 imply that, for some fixed polynomial $q'''(\cdot)$, we have that in this formula:
\begin{itemize}
\item The degree of all involved polynomials (which are univariate polynomials in $\epsilon$) 
is at most $2^{q'''(|\calG| + k)}$.
\item The bitsize of all coefficients is at most: $2^{q'''(|\calG| + k)}$.
\end{itemize}

By Lemma~\ref{LEM:AlmostNear}, we know that there exists an $\epsilon>0$
so that the formula $\QPEbound'(\epsilon)$ is true. We now apply 
Theorem 13.14 of \cite[Page 521]{BasuPollackRoy2011} 
to the set of polynomials that are atoms of $\QPEbound'(\epsilon)$  and conclude that $\QPEbound'(\epsilon^*)$ is 
true for some
$\epsilon^* \geq 2^{-2^{q'''(|\calG| + k)^2}}$. By the semantics of the formula $\QPEbound(\epsilon)$, we also have that 
$\QPEbound(\epsilon')$ is true for all positive $\epsilon' \leq \epsilon^*$, and the statement of the lemma follows. \qed 
\end{proof}

\noindent {\bf Proof of Theorem \ref{thm:main-fixpa-result}, part (2.)}  
The proof is completely analogous to the proof of parts (1.) and (3.).
We use the algebraically defined functions $H^{\epsilon}_\calG : B \rightarrow B^{\epsilon}$,
which are parametrized by an input variable $\epsilon$.
We ``instantiate'' $\epsilon$ with $\epsilon^* =  2^{-2^{q'''(|\calG| + k)^2}}$,  
where $k = \lceil - \log((\delta/2)^2) \rceil$.
We know we can define $\epsilon^*$
using an algebraic circuit having encoding size $q'''(|\calG| + k)^2$.
We thus can construct an $\{+, - ,*,/,\max,\min \}$-circuit $C(x)$, 
having encoding size polynomial in $|\calG|$ and $\size(\delta)$,
which defines the
function $H^{\epsilon^*}_\calG: B \rightarrow B^{\epsilon^*}$ on the domain $B$, 
and such that  every fixed point of $H^{\epsilon^*}_\calG$
is a $\epsilon^*$-QPE of $\calG$, 
which by Lemma \ref{LEM:AlmostVeryNear-QPE} is also $(\delta/2)$-close (in $\ell_\infty$) to 
an actual QPE.
Thus, applying the triangle inequality, if we
approximate a fixed point of $H^{\epsilon^*}_\calG$ within $\ell_\infty$ distance
$(\delta/2)$, we will have approximated a QPE of $\calG$ within $\ell_\infty$ distance $\delta$.
This shows that $\delta$-approximating a QPE, given $\calG$ and given $\delta > 0$, is in $\FIXPA$.
\qed

\section{Computing a 
$\delta$-almost-$\epsilon$-PE
\& $\delta$-almost-$\epsilon$-QPE is \PPAD-complete.}

\label{sec:delta-almost-epsilon-PE}

In this section we again exploit the functions $F^{\epsilon}_{\calG}(x)$
and $H^{\epsilon}_\calG(x)$,
for a EFGPR, $\calG$, 
devised in Section \ref{sec:epsilon-PE-char}  for $\epsilon$-PEs
and $\epsilon$-QPEs.
This time we do so in order to show that computing
a 
$\delta$-almost-$\epsilon$-PE, 
given $\calG$ and 
$\delta > 0$ and $\epsilon > 0$,
is \PPAD-complete.
We also show that the notion of $\delta$-almost-$\epsilon$-PE suitably 
``refines'' $\delta$-almost-SGPE (and thus also $\delta$-almost-NE), and that 
as a consequence computing a $\delta$-almost-SGPE (or a $\delta$-almost-NE),
given $\calG$ and given $\delta > 0$, is \PPAD-complete 
(\cite{DFP06}).
Furthermore, we also show computing a $\delta$-almost-$\epsilon$-QPE,
given $\calG$, and given $\delta > 0$ and $\epsilon > 0$ is 
PPAD-complete.

We have not yet actual defined the ``almost'' relaxation
for QPE, which we call 
$\delta$-almost-$\epsilon$-QPE.  We do so now:
for $\delta \geq 0$,
a behavior profile $b \in B$ is called a 
{\em $\delta$-almost $\epsilon$-quasi-perfect equilibrium} 
($\delta$-almost-$\epsilon$-QPE) of $\calG$, if it is (a): {\em fully mixed}, 
$b \in B^{> 0}$, and (b):  for all players $i$, 
all $j \in [d_i]$, and all
actions $a, a' \in {\mathcal A}_{i,j}$, 
if  $\MU^{j,a}_i(b) <  \MU^{j,a'}_i(b)  - \delta$   
then  $b_{i,j}(a) \leq \epsilon$.
Note that when $\delta = 0$ this definition is equivalent to
$\epsilon$-QPE (this is because for a {\em fully mixed} profile $b$,
$\MU^{j,a}_i(b) < \MU^{j,a'}_i(b)$ holds {\em if and only if} 
$\max_{b'_i \in B_i} U_i(b \mid_{(i,j)} (b'_i \mid \pi^{a}_{i,j}))  < 
 \max_{b''_i \in B_i} U_i(b \mid_{(i,j)} (b''_i \mid \pi^{a'}_{i,j}))$.\footnote{In fact, as noted earlier, 
van Damme \cite{vanDamme84} defines QPE using the strict inequalities $\MU^{j,a}_i(b) < \MU^{j,a'}_i(b)$
instead of  $\max_{b'_i \in B_i} U_i(b \mid_{(i,j)} (b'_i \mid \pi^{a}_{i,j}))  < 
 \max_{b''_i \in B_i} U_i(b \mid_{(i,j)} (b''_i \mid \pi^{a'}_{i,j}))$.}
Thus, our definition is a reasonable 
``almost'' relaxation of $\epsilon$-QPE.

We will make crucial use of some results and definitions from
\cite{EY07}, which we now recall.
Note that the circuit defining $F^{\epsilon}_{\calG}(x)$  associates a function  $F^{\epsilon}_{\calG}:
B^{\epsilon} \rightarrow B^{\epsilon}$ with each given pair $\langle \calG, \epsilon \rangle$,
where the rational value $\epsilon > 0$ is given in binary 
as part of the input.\footnote{In this section it will be more convenient to
view the domain of the function $F^{\epsilon}_{\calG}$ as $B^{\epsilon}$, rather than $B$, because $\epsilon > 0$
will be explicitly given.} 
Thus $|\calG| + \size(\epsilon)$ is the encoding size of the input from 
which  the algebraic circuit for $F^{\epsilon}_{\calG}(x)$ is generated.

Following  \cite{EY07}, we call 
the family of functions $\langle F^{\epsilon}_{\calG}(x) \rangle_{\{\langle \calG, \epsilon
\rangle\}}$, associated
with input pairs $\langle \calG, \epsilon \rangle$,  {\em polynomially continuous}
in their domain $B^{\epsilon}$,
if there is a polynomial $q(z)$ such that for all input pairs 
$\langle \calG, \epsilon \rangle$,  for every rational $\epsilon_1 > 0$, there is
a rational $\delta_1 > 0$, such that $\size(\delta_1) \leq q(|\calG| + \size(\epsilon) + \size(\epsilon_1))$
and such that for all $b,b' \in B^{\epsilon}$:
$$ \| b - b' \|_\infty < \delta_1  \;  \Longrightarrow \;  
\| F^{\epsilon}_{\calG}(b) - F^{\epsilon}_{\calG}(b') \|_\infty < 
\epsilon_1.$$
Again following \cite{EY07}, we call 
the family of functions  $\langle F^{\epsilon}_{\calG}(x) \rangle_{\{\langle \calG, \epsilon
\rangle\}}$ associated
with input instances $\langle \calG, \epsilon \rangle$,
{\em polynomially computable}  if (a):
the domain $B^{\epsilon}$ of the functions $F^{\epsilon}_{\calG}: B^{\epsilon} \rightarrow B^{\epsilon}$ 
is
a convex polytope described by a set of linear inequalities
with rational coefficients that
can be computed from the input $\langle \calG, \epsilon \rangle$ in
polynomial time (note that this is clearly always the case for $B^\epsilon$, 
because $\epsilon > 0$ is part of the input),
and (b): there is a polynomial $q(z)$ such that there is an algorithm 
that given $\langle \calG, \epsilon \rangle$, and given
a rational vector $b \in B^{\epsilon}$,
computes $F^{\epsilon}_{\calG}(b)$ 
(which is of course also a rational vector)
in time $q(|\calG| + \size(\epsilon) + \size(b))$.
We need the following Lemma:

\begin{lemma}
\label{lem:poly-con-poly-comp}
The family of functions 
$\langle F^{\epsilon}_{\calG}(x) \rangle_{\{\langle \calG, \epsilon
\rangle\}}$ for EFGPRs defined in Section \ref{approx-of-PE} 
(equation (\ref{eq:fixed-point-functions}))
is both (a.) polynomially computable and (b.) polynomially continuous.
\end{lemma}

\begin{proof}
\mbox{}

\noindent (a.): First, we observe
that the family of functions 
$\langle F^{\epsilon}_{\calG}(x) \rangle_{\{\langle \calG, \epsilon
\rangle\}}$ for EFGPRs
is  polynomially computable.
This follows easily from the definition of 
$F^{\epsilon}_{\calG}(x)$ given
Section \ref{approx-of-PE}  and 
in equations (\ref{eq:fixed-point-functions})
and (\ref{eq:t-i-j}).
Specifically, given a rational vector $b \in B^{\epsilon}$,
to compute $F^{\epsilon}_{\calG}(b)$,  we must first compute
a vector $y := h(b) := b + v(b)$,  where 
$v(b)_{i,j,a} := U_{i,j}(b \mid \pi^a_{i,j})
= U_i(b \mid \pi^a_{i,j})$.
Note that, given a rational vector $b \in B^{\epsilon}$,  each value 
$y_{i,j,a} = h(b)_{i,j,a} = b_{i,j,a} +  U_i(b \mid \pi^a_{i,j})$ is 
clearly computable in P-time,
because $U_i(x \mid \pi^a_{i,j})$ is given by a polynomial
in $x$ whose encoding size, as a sum of multilinear monomials, 
is polynomial in $| \calG | + \size(\epsilon)$.
Note also that the encoding size of the resulting rational vector $y$ is
clearly polynomial in $| \calG | + \size(\epsilon) + \size(b)$.
Next, having computed the vector $y$,  we must sort each subvector $y_{i,j}$,
associated with agent $(i,j)$, into a non-increasing sequence:
$z_{i,j} = (z_{i,j,a_{1}}, z_{i,j,a_{2}}, \ldots, z_{i,j,a_{{|\calA_{i,j}|}}})$.
We can clearly do so in P-time.
Next, for each agent $(i,j)$, we can clearly compute $t_{i,j}(b)$ in P-time 
using  the simple $\{ \max, + \}$ formula over the sorted vector
of inputs $z_{i,j}$
given in  equation (\ref{eq:t-i-j}).
Finally, having computed $t_{i,j}(b)$ and $y = h(b)$ in P-time,  
we have from equation
(\ref{eq:fixed-point-functions})
that $F^{\epsilon}_{\calG}(b)_{i,j,a} = \max ( h_{i,j,a}(b) - t_{i,j}(b) , \epsilon)$.
Thus we can compute $F^{\epsilon}_{\calG}(b)$ in time
polynomial in $|\calG| + \size(\epsilon) + \size(b)$,
given $\calG$, $\epsilon > 0$, and any rational vector $b \in B^{\epsilon}$.

\noindent (b.): Next, we want to show that the function family 
$\langle F^{\epsilon}_{\calG}(x) \rangle_{\{\langle \calG, \epsilon
\rangle\}}$ for EFGPRs
is  polynomially continuous.
We will in fact show that 
in the domain $B^\epsilon$ the function
$F^{\epsilon}_{\calG}(x)$ is
Lipschitz continuous 
with Lipschitz constant 
$2^{q(|\calG| + \size(\epsilon))}$  (with respect to the $\ell_\infty$ norm),
for some polynomial $q(\cdot)$.  In other words,
for all $b, b' \in B^{\epsilon}$,   we have:
\begin{equation}
\label{eq:lip-of-F}
\|F^{\epsilon}_{\calG}(b) -
F^{\epsilon}_{\calG}(b') \|_\infty  \leq 
2^{q(|\calG| + \size(\epsilon))} \cdot \|b - b' \|_\infty .
\end{equation}
Of course, it immediate follows from (\ref{eq:lip-of-F})
is that the family of functions
$\langle F^{\epsilon}_{\calG}(x) \rangle_{\{\langle \calG, \epsilon
\rangle\}}$ is polynomially continuous: in the definition
of polynomially continuity, take $\delta_1 := 
\frac{1}{2^{q(|\calG| + \size(\epsilon))}} \cdot \epsilon_1$,
it then follows from (\ref{eq:lip-of-F}) that 
for all $b, b' \in B^{\epsilon}$,   $\| b - b' \|_\infty < \delta_1
\Longrightarrow   \| F^{\epsilon}_{\calG}(b) - F^{\epsilon}_{\calG}(b') 
\|_\infty < \epsilon_1$.
Furthermore, clearly
$\size(\delta_1)  \leq q^*(|\calG| + \size(\epsilon) + \size(\epsilon_1))$,
for some fixed polynomial $q^*(\cdot)$.
So, we only need to establish (\ref{eq:lip-of-F}).

Consider any $b, b' \in B^{\epsilon}$.
First, let us bound $\| h(b) - h(b') \|_\infty$.
Recall that $h_{i,j,a}(x) = x_{i,j,a} + U_i(x \mid \pi^a_{i,j})$.
Moreover, we know by Proposition
\ref{prop:expected-poly}  that $U_i(x \mid \pi^a_{i,j})$ is given by 
an explicit polynomial (a weighted sum of multilinear monomials) 
in the variables $x$,
with degree bounded by the height $h^\calG$ of the game tree $T$, and 
with encoding size polynomial in $|\calG|$.

First, consider any monomial $f(x) = \alpha \cdot x_{i_1} \ldots x_{i_k}$.
Note that in the domain $B^\epsilon \subseteq [0,1]^d$ (for a suitable dimension
$d$), the monomial
$f(x)$ is Lipschitz continuous 
with Lipschitz constant $|\alpha| k$
(with respect to the $\ell_\infty$ norm).
To see this simple fact, note that for $b, b' \in B^\epsilon$, we have
$|f(b) - f(b')|  \leq 
|\alpha| |b_{i_1} \ldots b_{i_k} - b'_{i_1} \ldots b'_{i_k} |$.
Furthermore, by induction on $k \geq 1$, we have that for 
$b, b' \in [0,1]^k$,
$| b_{1} \ldots b_{k} - b'_{1} \ldots b'_{k} | \leq k \| b - b' \|_\infty$.
The base case, $k=1$, is trivial.
For the inductive case, we have:
\begin{eqnarray*}
|b_{1} \ldots b_{k} - b'_{1} \ldots b'_{k} | 
& = &  | b_1 \ldots b_k - b_1 b'_2 \ldots b'_k  + b_1 b'_2 \ldots b'_k 
- b'_1 \ldots b'_k |\\
& \leq &  | b_1 \ldots b_k - b_1 b'_2 \ldots b'_k |  + | b_1 b'_2 \ldots b'_k 
- b'_1 \ldots b'_k |\\
& = &  |b_1| \cdot |b_2 \ldots b_k - b'_2 \ldots b'_k|  +  
|b'_2 \ldots b'_k| \cdot |b_1 - b'_1|\\
& \leq &  |b_1| \cdot (k-1) \|b - b' \|_\infty  +
|b_1 - b'_1| \cdot |b'_2 \ldots b'_k|    \quad \quad \mbox{(by inductive hypothesis)}\\
& \leq & (k-1) \|b - b' \|_\infty +  |b_1 - b'_1|   \quad \quad
\mbox{(because $|b_1| \in [0,1]$ and $|b'_2 \ldots b'_k| \in [0,1]$)}\\
& \leq & k \| b - b' \|_\infty.
\end{eqnarray*}

Now suppose that the polynomial $h_{i,j,a}(x) = x_{i,j,a} + U_i(x \mid \pi^a_{i,j})$
is the sum of $M_{i,j,a}$ weighted monomials, and that the maximum
absolute value of a coefficient of any of the monomials is $A^{\max}_{i,j,a}$.
Then by the above, for any $b , b' \in B^\epsilon$, we have
$|h_{i,j,a}(b) - h_{i,j,a}(b') |  \leq M_{i,j,a} \cdot A^{\max}_{i,j,a} \|b - b'\|_\infty$.
Let $M^{\max} = \max_{i,j,a} M_{i,j,a}$ , and let $A^{\max}= \max_{i,j,a} A^{\max}_{i,j,a}$.
Then we have $\| h(b) - h(b')\|_\infty \leq M^{\max} \cdot A^{\max} \cdot \|b - b' \|_\infty$.   Thus, clearly $h(x)$ is Lipschitz continuous in domain $B^\epsilon$,
with Lipschitz constant $M^{\max} \cdot A^{\max}$, which is clearly upper bounded
by $2^{q(|\calG|)}$ for some polynomial $q(\cdot)$.

Next, note that the sort function has Lipschitz constant $1$,  with respect to 
$\ell_\infty$.  
In other words,  if $\sort(y)$ is a function that
takes a vector $y \in \real^k$ as input, and yields its (non-increasing)
sort, $\sort(y) \in \real^k$, then for all $y, y' \in \real^k$,
$\| \sort(y) - \sort(y') \|_\infty \leq \| y - y' \|_\infty$.

For completeness, we provide a proof of this easy fact.
Suppose for contradiction that $|\sort(y)_{i^*} - \sort(y')_{i^*} | =
\| \sort(y) - \sort(y') \|_\infty >  \| y - y' \|_\infty$,
for some index $i^* \in [k]$.
Define the permutations $\pi$ and $\pi'$ of $[k]$,
such that
for all $i \in [k]$, $\sort(y)_i = y_{\pi(i)}$  and $\sort(y')_i =  y_{\pi'(i)}$. 
Suppose, wlog, that $y_{\pi(i^*)} = \sort(y)_{i^*} < \sort(y')_{i^*} = y'_{\pi'(i^*)}$.
Since $|\{ \pi(1),\ldots, \pi(i^*) \}| = i^* >  i^*-1 = |\{\pi'(1),\ldots, \pi'(i^*-1)\}|$, there must exist
an $r \in \{1, \ldots, i^*\}$ such that
$\pi(r) \in \{\pi'(i^*), \pi'(i^*+1), \ldots, \pi'(k) \}$.
In other words,   
$y_{\pi(r)} \leq y_{\pi(i^*)} = \sort(y)_{i^*} < \sort(y')_{i^*} = y'_{\pi'(i^*)} \leq
y'_{\pi(r)}$.   Thus $\| \sort(y) - \sort(y') \|_\infty =
| \sort(y)_{i^*} - \sort(y')_{i^*} | \leq  |y'_{\pi(r)} - y_{\pi(r)}| \leq \| y' - y \|_\infty$. 

Note also that the composition $f_1(f_2(x))$ of 
Lipschitz continuous functions $f_1(y)$ and $f_2(x)$,
where $f_1(y)$ has Lipschitz constant $\beta_1$ and 
$f_2(x)$ has Lipschitz constant $\beta_2$
(both with respect to the $\ell_\infty$ norm),
is Lipschitz continuous with constant
$\beta_1 \cdot \beta_2$  (again with respect to $\ell_\infty$).

Now, consider $t_{i,j}(x)$ as defined by equation (\ref{eq:t-i-j}).
The expression defining $t_{i,j}(x)$ is a maximum over linear (affine) expressions
(using $\epsilon$ as a constant)
with at most $| \calA_{i,j} |$ terms over
the sorted vector of variables $z_{i,j}$.  
Since the $\max$ function has Lipschitz constant $1$
(it is just a component of the sort function), 
it follows that
for all $b , b' \in B^\epsilon$, we have $\| t_{i,j}(b) - t_{i,j}(b') \|_\infty \leq 
2^{q'(|\calG| + \size(\epsilon))} \cdot \| b - b' \|_\infty$ for some polynomial
$q'(\cdot)$.

Finally, since we have $F^\epsilon_\calG(x)_{i,j,a} = \max (h_{i,j,a}(x) - t_{i,j}(x),\epsilon)$,
and since $\max$ has Lipschitz constant $1$, and since the sum of 
two Lipschitz functions with Lipschitz constant $\beta_1$ and $\beta_2$
is a Lipschitz function with Lipschitz constant $\leq \beta_1 + \beta_2$,
we are done:  there is a polynomial $q(\cdot)$ such that for all $b,b' \in B^\epsilon$,
$$\| F^{\epsilon}_\calG(b) - F^{\epsilon}_{\calG}(b') \|_\infty 
\leq 2^{q(|\calG| + \size(\epsilon))} \cdot \| b - b'\|_\infty.$$

\vspace*{-0.27in}

\qed
\end{proof}

In fact, let us remark that 
Lemma \ref{lem:poly-con-poly-comp} is a special case
of a more general fact, 
namely that function families defined
by  $\{+,*,\max, \sort \}$-{\em formulas} whose
encoding size is polynomial in the input instance,
over a bounded domain such as $B^\epsilon$, are
necessarily polynomially computable and polynomially continuous.
The proof of the next lemma will argue this more explicitly.

\begin{lemma}
\label{lem:poly-con-poly-comp-qpe}
The family of functions 
$\langle H^{\epsilon}_{\calG}(x) \rangle_{\{\langle \calG, \epsilon
\rangle\}}$ for EFGPRs defined in Section \ref{approx-of-PE} 
(equation (\ref{eq:qpe-fixed-point-functions}))
is both (a.) polynomially computable and (b.) polynomially continuous.
\end{lemma}

\begin{proof}
\noindent (a.): First, we again observe
that 
the family of functions 
$\langle H^{\epsilon}_{\calG}(x) \rangle_{\{\langle \calG, \epsilon
\rangle\}}$ for EFGPRs
is  polynomially computable over the corresponding domain $B^{\epsilon}$.
This again follows easily from the definition of 
$H^{\epsilon}_{\calG}(x)$ given
in Section \ref{approx-of-PE}, 
in equations (\ref{eq:qpe-fixed-point-functions})
and  in the dynamic program (\ref{eq:dyn-program}) defining
$\MU^{j,a}_i(x)$.
Specifically, given a rational vector $b \in B^{\epsilon}$,
to compute $H^{\epsilon}_{\calG}(b)$,  
noting that $\Normalize^{\epsilon}(b) = b$, 
we must first compute
a vector $y' := h'(b) := b + v'(b)$,  where 
$v'(b)_{i,j,a} := \MU^{j,a}_i(b)$.
We know from the dynamic program given in (\ref{eq:dyn-program}) that 
given $\calG$ and $b \in B^{\epsilon}$, we can compute
$v'(b)_{i,j,a}$ in time polynomial in $|\calG| + \size(\epsilon) + \size(b)$, for all $i$, $j$, and $a$.  
In particular, it is important to emphasize that $\size(\MU^{j,a}_i(b))$ remains
polynomial in 
$|\calG| + \size(\epsilon) + \size(b)$, and so do the sizes of all the 
intermediate rational numbers computed by subformulas of $\MU^{j,a}_i(b)$. 
This is not only because the formula has only polynomial size, but also because,
importantly, the special kind of $\{+,-,/,\max,\min, \sort \}$-formula
defining $K^{j,a}_i(b)$, given in (\ref{eq:dyn-program}), 
has the property that the only occurrences of {\em division}  in the formula occur when 
the denominator of the division operation evaluates to $\Prob_b(I_{i,j})$ for some information set $I_{i,j}$.
But the probability $\Prob_b(I_{i,j})$, for any $b \in B^{\epsilon}$ is
at least $\epsilon^{h^\calG}$. Note that $\size(\epsilon^{h^{\calG}}) \leq h^{\calG} \cdot \size(\epsilon)$. 
This ensures that the rational values arising as the result of such division gates
in the formula for $K^{j,a}_i(b)$ always have an encoding size that is polynomial in 
$|\calG| + \size(\epsilon) + \size(b)$.  
It follows, by an easy induction on the size $m$ of a subformula, 
that the encoding size of the value computed by a subformula of size $m$ has encoding 
size polynomial in $m \cdot (|\calG| + \size(\epsilon) + \size(b))$.  
Since $m$
itself is bounded by a polynomial in $|\calG| + \size(\epsilon) + \size(b)$, 
this means all values computed in the formula have encoding size
bounded by a polynomial in $|\calG| + \size(\epsilon) + \size(b)$.
We can thus also compute  $h'(b)_{i,j,a}$ in time polynomial in   $|\calG| + \size(\epsilon) + \size(b)$.
Likewise, computing $t'_{i,j}(b)$ is easily done in time polynomial in $|\calG| + \size(\epsilon) + \size(b)$,
using sorting.  Thus $H^{\epsilon}_\calG(b)$ can be computed in time polynomial in
$|\calG| + \size(\epsilon) + \size(b)$.
Thus we can compute $H^{\epsilon}_{\calG}(b)$ in time
polynomial in $|\calG| + \size(\epsilon) + \size(b)$,
given $\calG$, $\epsilon > 0$, and any rational vector $b \in B^{\epsilon}$.

(b.)  We now argue that 
the family of functions 
$\langle H^{\epsilon}_{\calG}(x) \rangle_{\{\langle \calG, \epsilon
\rangle\}}$  is polynomially continuous over the domain $B^{\epsilon}$.  
We will again actually show that the
functions $H^{\epsilon}_{\calG}(x)$
are Lipschitz continuous, with a Lipschitz constant of the form
$2^{q(|\calG| + \size(\epsilon))}$, for some polynomial $q(\cdot)$,
over domain $B^{\epsilon}$.
Just as in Lemma \ref{lem:poly-con-poly-comp}, this implies
polynomial continuity.

The proof is again similar to the case $F^{\epsilon}_{\calG}(x)$. 
We noted already, after the proof of Lemma \ref{lem:poly-con-poly-comp}, that 
an adaptation of that proof shows that {\em any} such
function that can be defined by a $\{ + , *, \max , \sort \}$-{\em formula} 
and has encoding size polynomial in $|\calG| + \size(\epsilon)$ is polynomially continuous
over the domain $B^{\epsilon}$.   We will establish a more direct
version of this fact here.
$H^{\epsilon}_\calG(x)$ is defined by a $\{+,*,/, \max, \sort \}$-formula,
meaning it also involves division.
However,
in the case of $H^{\epsilon}_\calG(x)$ we furthermore have the fact 
that the only use of division is inside subformulas which compute
$\Prob_b(u \mid I_{i,j}) = \frac{\Prob_b(u)}{\Prob_b(I_{i,j})}$,
for some information set $I_{i,j}$ and some node $u \in I_{i,j}$.
Furthermore, we also see easily 
by inspection of $H^{\epsilon}_\calG(x)$ that, for all $b \in B^{\epsilon}$,
and for {\em every} subformula $f_1(x)$ of the formula for 
 $H^{\epsilon}_\calG(x)$,  we have 
$\max_{b \in B^{\epsilon}} |f_1(b)| \leq 2^{q''(|\calG| + 
\size(\epsilon))}$ for some fixed polynomial  $q''(\cdot)$ 
which is also independent of the subformula.
We will use both of these facts.

Now, for any two subformulas $f_1(x)$ and $f_2(x)$
of $H^{\epsilon}_\calG(x)$,
suppose $f_1(x)$  ($f_2(x)$) has Lipschitz constant $\beta_1$
($\beta_2$),
with respect to the $\ell_\infty$ norm, i.e.,
that for $k \in \{1,2\}$,
if for all $b, b' \in B^{\epsilon}$ we have 
$|f_k(b) - f_k(b') | < \beta_k \| b - b' \|_\infty$, then:

\begin{enumerate}
\item  $f_1(x) \cdot f_2(x)$ has Lipschitz constant at most
$2^{q''(|\calG| + 
\size(\epsilon))} \cdot (\beta_1 + \beta_2)$.  
To see this, note that for all
$b, b' \in B^{\epsilon}$ we have:
\begin{eqnarray*}
| f_1(b) \cdot f_2(b) - f_1(b') \cdot f_2(b') | & = & 
|  f_1(b) \cdot ( f_2(b) - f_2(b') ) +  f_2(b') ( f_1(b) - f_1(b')) |\\
& \leq & |f_1(b) | \cdot |f_2(b) - f_2(b') | + |f_2(b)| \cdot
|f_1(b) - f_1(b')| \\
& \leq &   2^{q''(|\calG| + \size(\epsilon))} (\beta_1 + \beta_2) \cdot 
\| b - b' \|_\infty
\end{eqnarray*}
\item
 $f_1(x) + f_2(x)$ has Lipschitz constant at most
$\beta_1 + \beta_2$.
(This is obvious.)

\item $\max( f_1(x), f_2(x) )$ has
Lipschitz constant at most $\max(\beta_1, \beta_2)$.
This follows immediately from the more general fact 
(established in the proof of Lemma \ref{lem:poly-con-poly-comp}) that the $\sort$ function has Lipschitz constant $1$
(under the $\ell_\infty$ norm), since $\sort(y)_1 = \max_i y_i$.
More directly (and repeating some the same arguments),  we have:\\
$| \max(f_1(b), f_2(b)) - \max( f_1(b') , f_2(b') ) | 
\leq \max( |f_1(b) - f_1(b') |, |f_2(b) - f_2(b') |) 
\leq \max (\beta_1, \beta_2) \cdot \| b - b'\|_\infty$.
To see why the first inequality holds, assume w.l.o.g. that
$\max(f_1(b), f_2(b)) \geq \max(f_1(b'), f_2(b'))$, and that
$f_1(b) \geq  f_2(b)$.   Then, if $f_1(b') \geq f_2(b')$ we have
$| \max(f_1(b), f_2(b)) - \max( f_1(b') , f_2(b') ) | = | f_1(b) - f_1(b')|$.
Otherwise, if $f_1(b') < f_2(b')$, then 
$| \max(f_1(b), f_2(b)) - \max( f_1(b') , f_2(b') ) | = | f_1(b) - f_2(b')|
< | f_1(b) - f_1(b')|$, since $f_1(b) \geq f_2(b')$.
Thus, w.l.o.g., in all cases,
$| \max(f_1(b), f_2(b)) - \max( f_1(b') , f_2(b') ) | 
\leq \max( |f_1(b) - f_1(b') |, |f_2(b) - f_2(b') |) $.
\end{enumerate}

Next, observe that for $x$ in the domain $B^{\epsilon}$ 
the functions $\frac{\Prob_x(u)}{\Prob_x(I_{i,j})}$ 
have  Lipschitz constant 
$2^{q'(|\calG| + \size(\epsilon))}$, for some fixed polynomial $q'(\cdot)$.
This holds because 
for all $i,j$ and $u$, and 
for all $b, b' \in B^{\epsilon}$, we have:

\begin{eqnarray*}
| \frac{\Prob_b(u)}{\Prob_b(I_{i,j})} - 
\frac{\Prob_{b'}(u)}{\Prob_{b'}(I_{i,j})} | & = &  
| \frac{\Prob_{b'}(I_{i,j}) \cdot \Prob_b(u) - 
   \Prob_{b}(I_{i,j})  \cdot \Prob_{b'}(u)}{\Prob_{b}(I_{i,j})
\cdot \Prob_{b'}(I_{i,j})} | \\
& \leq &  \frac{1}{| \Prob_{b}(I_{i,j}) \cdot
\Prob_{b'}(I_{i,j}) |} \cdot  | \Prob_{b'}(I_{i,j}) \cdot \Prob_b(u) - 
   \Prob_{b}(I_{i,j})  \cdot \Prob_{b'}(u) | \\
& \leq & \frac{1}{\epsilon^{2 \cdot h^\calG}} \cdot   | \Prob_{b'}(I_{i,j}) \cdot \Prob_b(u) - 
   \Prob_{b}(I_{i,j})  \cdot \Prob_{b'}(u) |  \quad \quad
\mbox{(because for all $b'' \in B^{\epsilon}$,   $\Prob_{b''}(I_{i,j}) \geq \epsilon^{h^{\calG}}$)} \\
& \leq & 
2^{q''(| \calG | + \size(\epsilon))} 
\cdot   | \Prob_{b'}(I_{i,j}) \cdot \Prob_b(u) - 
   \Prob_{b}(I_{i,j})  \cdot \Prob_{b'}(u) |   \\
& = &   
2^{q''(| \calG | + \size(\epsilon))} 
|   \Prob_{b'}(u) (  \Prob_{b}(I_{i,j}) - \Prob_{b'}(I_{i,j})) + 
      \Prob_{b'}(I_{i,j})) ( \Prob_{b'}(u) - \Prob_b(u)) | \\
& \leq &  2^{q''(| \calG | + \size(\epsilon))}  ( |(\Prob_{b}(I_{i,j}) - \Prob_{b'}(I_{i,j}))| + 
| \Prob_{b'}(u) - \Prob_b(u)) | ) \\
& \leq & 2^{q''(| \calG | + \size(\epsilon)) + 
q'(| \calG | + \size(\epsilon)) + 1} \| b - b' \|_\infty
\end{eqnarray*}

Thus, by induction on the size $s$ of any 
subformula of 
$f(x)$ of  
$H^{\epsilon}_\calG(x)$,  
which is either a $\{+, *, \max, \sort \}$-formula
or of the form $\frac{\Prob_x(u)}{\Prob_x(I_{i,j})}$,
we have that
for all
$b , b' \in B^{\epsilon}$, 
$| f(b) - f(b')| \leq 2^{(q'(|\calG| + \size(\epsilon)) + 1) \cdot s} \leq
2^{q(|\calG| + \size(\epsilon))}$, for some fixed polynomial
$q(\cdot)$.
Thus,  $H^{\epsilon}_\calG(x)$
is polynomially continuous over the domain $B^{\epsilon}$.
\qed
\end{proof}

We now define a search problem called
 the {\em almost fixed point approximation problem}, 
called  the {\em weak (fixed point) approximation problem} in \cite{EY07}, 
specialized to the case of 
the fixed point functions $F^{\epsilon}_\calG: B^\epsilon \rightarrow B^\epsilon$.
Namely, given as input $\langle \calG, \epsilon \rangle$, and a rational $\delta_1 > 0$,  
compute a rational vector $b' \in B^{\epsilon}$, such that  
$\| F^{\epsilon}_{\calG}(b') - b' \|_\infty < \delta_1$.
We shall make crucial use of the following fact, 
which was established in \cite{EY07} by employing 
Scarf's \cite{Scarf67} algorithm, and Kuhn's \cite{Kuhn68} related algorithm,
for weak (i.e., almost) fixed point approximation:

\begin{proposition}[\cite{EY07}, Prop. 2.2 (part 2.)]
\label{prop:weak-approx-ppad}
If the family of fixed point functions 
$\langle F^{\epsilon}_{\calG}(x) \rangle_{\{\langle \calG, \epsilon
\rangle\}}$,
associated with input instances $\langle \calG, \epsilon \rangle$,
is polynomially continuous and polynomially computable, then
the almost (weak) fixed point approximation problem for $F^{\epsilon}_{\calG}(x)$,
given input $\langle \calG, \epsilon \rangle$,
is in \PPAD.
\end{proposition}

\noindent The following Lemma is the key to this section:

\begin{lemma}
\label{lem:almost-fixed-is-almost-equilib}
For any EFGPR, $\calG$, and $\epsilon > 0$:
\begin{enumerate}

\item For any $\delta > 0$, 
if $b \in B^{\epsilon}$ satisfies 
$\| b - F^{\epsilon}_{\calG}(b) \|_{\infty} < \delta$,
then $b$ is a $(3 \cdot \delta)$-almost-$(\delta + \epsilon)$-PE of $\calG$.

\item For any $\delta > 0$,
if $b \in B^{\epsilon}$ satisfies 
$\| b - H^{\epsilon}_{\calG}(b) \|_{\infty} < \delta$,
then $b$ is a $(3 \cdot \delta)$-almost-$(\delta + \epsilon)$-PE of $\calG$.

\item 
For any $\delta > 0$, let $\epsilon(\calG,\delta) :=
  \frac{p^\calG_{0,\min}}{12 \cdot (\height^\calG +1) \cdot M_\calG \cdot |\calG|} \cdot \delta$.

If $b \in B^{\epsilon(\calG,\delta)}$ 
is a
$\frac{1}{(\height^\calG+1)} \cdot \epsilon(\calG,\delta)^{(\height^\calG + 1)}$-almost-$(2 \cdot 
\epsilon(\calG,\delta))$-PE,
then $b$ is a $\delta$-almost-SGPE.
\end{enumerate}
\end{lemma}
\begin{proof}
\mbox{}\\
\noindent (1.) 
Suppose that for $b \in B^{\epsilon}$, we have  $\| F^{\epsilon}_\calG(b) - b \|_\infty \leq 
\delta$.\\
Then $| b_{i,j,a} - \max (b_{i,j,a}+ v(b)_{i,j,a} -t_{i,j}(b),\epsilon)| \leq \delta$
for all $(i,j,a)$. \\
Recall that $v(b)_{i,j,a}= U_i(b \mid \pi^a_{i,j}) = U_{i,j}(b \mid \pi^a_{i,j})$. \\
Now note that
$| b_{i,j,a} - \max (b_{i,j,a}+ U_i(b \mid \pi^a_{i,j}) -t_{i,j}(b),\epsilon)| \leq \delta$
implies the following, by case splitting based on the value of $b_{i,j,a}$:
\begin{enumerate}
\item  If $b_{i,j,a} > \epsilon + \delta$,  then 
$|b_{i,j,a} - (b_{i,j,a}+ U_i(b \mid \pi^a_{i,j})-t_{i,j}(b))| \leq \delta$,
and thus  $|U_i(b \mid \pi^a_{i,j})-t_{i,j}(b))| \leq \delta$.
Thus, in this case $t_{i,j}(b) + \delta \geq 
U_i(b \mid \pi^a_{i,j}) \geq t_{i,j}(b) - \delta$.

\item If $\epsilon \leq b_{i,j,a} \leq  \epsilon + \delta$, 
then $b_{i,j,a}+ U_i(b \mid \pi^a_{i,j})-t_{i,j}(b)  \leq \epsilon + 2 \cdot \delta$,
and thus $U_i(b \mid \pi^a_{i,j})-t_{i,j}(b) \leq 2 \cdot \delta$,
and so $U_i(b \mid \pi^a_{i,j}) \leq t_{i,j}(b) + 2 \cdot \delta$.
\end{enumerate}
Thus, for all $(i,j,a)$, we have 
 $U_i(b \mid \pi^a_{i,j})  \leq t_{i,j}(b) + 2 \cdot \delta$,
and for all $(i,j,a)$ where $b_{i,j,a} > \epsilon + \delta$,
we have $U_i(b \mid \pi^a_{i,j}) \geq t_{i,j}(b) - \delta$.
Thus,  if $b_{i,j,a'} > \epsilon + \delta$,  then 
$(\max_a U_i(b \mid \pi^a_{i,j})) - U_i(b \mid \pi^{a'}_{i,j}) \leq 3 \delta$.
In other words,  $b$ is a $(3 \cdot \delta)$-almost-$(\epsilon + \delta)$-PE.
This completes the proof of Part (1.). of Lemma
\ref{lem:almost-fixed-is-almost-equilib}.\\

\noindent (2.):  the proof of part (2.) is  actually identical to the proof of 
part (1.), except that instead of $v(b)_{i,j,a} = U_i(b \mid \pi^a_{i,j})$,
we have to use $v'(b)_{i,j,a} = \MU^{j,a}_i(b)$, and instead of
$t_{i,j}(b)$ we have $t'_{i,j}(b)$.   If we systematically replace
occurrences of $U_i(b \mid \pi^a_{i,j})$ by $\MU^{j,a}_i(b)$ in the proof, 
and likewise replace $U_i(b \mid \pi^{a'}_{i,j})$ by $\MU^{j,a'}_i(b)$, 
and replace $t_{i,j}(b)$ by $t'_{i,j}(b)$, then the 
proof remains unchanged.
Note, in particular, that for $b \in B^{\epsilon}$,  we have $b = \Normalize^{\epsilon}(b)$,
and thus we can ignore the applications of $\Normalize^{\epsilon}(x)$ in the
definition of $H^{\epsilon}_\calG$, because here we are explicitly given $\epsilon > 0$ and
ae can view the function as $H^{\epsilon}_\calG: B^{\epsilon} \rightarrow B^{\epsilon}$.\\

\noindent (3.):

Recall that (w.l.o.g.)  
the payoff functions $r_i: \Leaves \rightarrow \nat_{> 0}$
are 
positive integer-valued 
for every player in $\calG$, and that $M_\calG$ denotes the maximum such value.
Also recall that $\height^\calG$ denotes the {\em height} of the game
tree $T = (V,E)$ of $\calG$, and that for any node $u \in V$,
$\height^\calG_u$ denotes the height of the subtree rooted at $u$.

Note that
for any profile $b \in B^{\epsilon'}$ for any $\epsilon' > 0$,
for any player $i$, any information set $j \in [d_i]$, and
for any node $u \in I_{i,j}$, 
the conditional 
probability
$\Prob_b(u \mid I_{i,j})$
of the play reaching node $u$ conditioned on the event 
of reaching information 
set $I_{i,j}$, under profile $b$,
is well defined.
Furthermore, importantly, again note
that the conditional probability $\Prob_b(u \mid I_{i,j})$  
is {\em independent of}  $b_i$.  
It only depends on 
the behavior strategies of players other than $i$, because, by perfect recall, 
for all nodes $u \in I_{i,j}$
the visible history for player $i$ at node $u$ is the same: it is $Y_{i,j}$.

For $i \in [n]$,
and for $j \in [d_{i'}]$, we use $U^{j}_{i}(b)$ to denote the 
conditional expected payoff to player $i$,
conditioned on the event of reaching information set $I_{i,j}$.  

We are now ready to prove (3.).
By assumption,
$b \in B^{\epsilon(\calG,\delta)}$, and $b$ is a\\  
\centerline{$\frac{1}{(\height^\calG+1)} \cdot \left( \frac{p^\calG_{0,\min}}{12 \cdot 
(\height^\calG + 1) \cdot M_\calG \cdot |\calG|} \cdot \delta \right)^ 
{(\height^\calG + 1)}$-almost-$(\frac{p^\calG_{0,\min}}{6 \cdot 
(\height^\calG + 1) \cdot M_\calG \cdot
|\calG|} \cdot \delta)$-PE.}
We will show that any such $b$ is also a $\delta$-almost-SGPE
of $\calG$.
Consider $b$ from the point of view of a single player $i$.
We need to show that behavior strategy $b_i$ is a $\delta$-almost best response to $b$,
i.e., that $U_i(b) \geq U_i(b \mid \pi^{c}_i) - \delta$, for any pure strategy $c \in S_i$.
Recall that a pure strategy $c: [d_i] \rightarrow \Sigma$  
for player $i$ maps information sets $j \in [d_i]$ to
available actions $c(j) \in \calA_{i,j}$.\\

\begin{claim}
\label{claim:delta-1}
For every player $i$,
every $j \in [d_i]$,
and every action $a \in \calA_{i,j}$ such that\\ 
$b_{i,j,a} > (\frac{p^\calG_{0,\min}}{6 \cdot 
(\height^\calG + 1) \cdot M_\calG \cdot                  
|\calG|} \cdot \delta)$, we have for any $a' \in \calA_{i,j}$:
$$U^j_i(b \mid \pi^a_{i,j}) \geq U^j_i(b \mid \pi^{a'}_{i,j}) - \frac{1}{3 \cdot (\height^{\calG} + 1)} \cdot \delta$$
\end{claim}
\begin{proof}
Since $b \in B^{\epsilon(\calG,\delta)}$
is a $\frac{1}{(\height^\calG + 1)} \left( \frac{p^\calG_{0,\min}}{12 \cdot (\height^\calG + 1) \cdot M_\calG \cdot |\calG|} \cdot \delta \right)^ 
{(\height^\calG + 1)}$-almost-$(\frac{p^\calG_{0,\min}}{6 \cdot (\height^\calG + 1) \cdot M_\calG \cdot
|\calG|} \cdot \delta)$-PE,
for any  $a 
\in \calA_{i,j}$ such that  
$b_{i,j,a} > (\frac{p^\calG_{0,\min}}{6 \cdot (\height^\calG + 1) \cdot M_\calG \cdot                  
|\calG|} \cdot \delta)$,
and any $\pi^{a'}_{i,j}$, we 
know that:

\begin{equation}
\label{eq:bound-on-almost-expectation}
U_i(b \mid \pi^a_{i,j} ) \geq U_i(b \mid \pi^{a'}_{i,j}) - 
\frac{1}{(\height^\calG + 1)} \left( \frac{p^\calG_{0,\min}}{12 \cdot (\height^\calG + 1) \cdot M_\calG \cdot |\calG|} \cdot \delta \right)^ 
{(\height^\calG + 1)}
\end{equation}
Note that, for any
$b' \in B^{\epsilon(\calG,\delta)}$, we have  
$$\Prob_{b'}(I_{i,j}) \geq \epsilon(\calG,\delta)^{\height^\calG} = 
\left( \frac{p^\calG_{0,\min}}{6 \cdot (\height^\calG + 1) \cdot M_\calG 
\cdot |\calG|} \cdot \delta \right)^{\height^\calG }$$
This follows because $\epsilon(\calG, \delta) \leq
p^\calG_{0,\min}$, and thus under profile $b' \in B^{\epsilon(\calG,\delta)}$,
every ``edge'' of the game tree will have probability
at least $\epsilon(\calG,\delta)$.
Thus already for every node $u \in I_{i,j}$,  $\Prob_{b'}(u) \geq  
\epsilon(\calG,\delta)^{\height^\calG}$, and 
so $\Prob_{b'}(I_{i,j}) \geq
\Prob_{b'}(u) \geq \epsilon(\calG,\delta)^{\height^\calG}$.

Now note that, for any profile $b' \in B^{\epsilon(\calG,\delta)}$, 
the expected payoff $U_i(b')$ can be expressed
as a sum $U_i(b') = U^j_i(b') \Prob_{b'}(I_{i,j}) +  U^{\neg j}_i(b') \Prob_{b'}(\neg I_{i,j})$,
where $U^{\neg j}_i(b')$ denotes the expected payoff to player $i$ conditioned
on {\em not} reaching information set $I_{i,j}$, and 
$\Prob_{b'}(\neg I_{i,j}) \doteq (1- \Prob_{b'}(I_{i,j}))$
denotes the probability of not reaching information set $I_{i,j}$.

Note that, if in any such profile $b'$ we change only the local 
strategy $b'_{i,j}$ to a new strategy $b''_{i,j}$ then 
this does not effect the probabilities
$\Prob_{b'}(I_{i,j})$ and
$\Prob_{b'}(\neg I_{i,j})$, 
nor does it effect the conditional expectation $U^{\neg j}_i(b')$.
In other words, for any behavior profile $b' \in B^{\epsilon(\calG,\delta)}$ and 
any local strategy $b''_{i,j} \in B_{i,j}$, we have:
\begin{equation}
\label{eq:decomp-of-condition-expect}
U_i(b' \mid b''_{i,j}) = U^j_i(b' \mid b''_{i,j}) \cdot \Prob_{b'}(I_{i,j})
+  U^{\neg j}_i(b') \cdot \Prob_{b'}(\neg I_{i,j})
\end{equation}
Now suppose, {\em for contradiction}, that 
for some $\pi^{a'}_{i,j}$, we have: 
$$U^j_i(b \mid \pi^a_{i,j}) < 
U^j_i(b \mid \pi^{a'}_{i,j}) - \frac{1}{3(\height^\calG + 1)} \cdot \delta.$$
But then, by applying
equation (\ref{eq:decomp-of-condition-expect})
with $b' := (b \mid \pi^{a'}_{i,j})$ and $b''_{i,j} := \pi^a_{i,j}$,
 we have: 
\begin{eqnarray*}
U_i(b \mid \pi^a_{i,j}) & = &  U^j_i(b \mid \pi^a_{i,j}) 
\cdot \Prob_{(b \mid \pi^{a'}_{i,j})}(I_{i,j})
+  U^{\neg j}_i(b \mid \pi^{a'}_{i,j}) \cdot \Prob_{(b \mid \pi^{a'}_{i,j})}(\neg I_{i,j})\\
& < & (U^j_i(b \mid \pi^{a'}_{i,j}) - \frac{1}{3 (\height^\calG + 1)} \cdot \delta)
\cdot 
       \Prob_{(b \mid \pi^{a'}_{i,j})}(I_{i,j})
+  U^{\neg j}_i(b \mid \pi^{a'}_{i,j}) \cdot \Prob_{(b \mid \pi^{a'}_{i,j})}(\neg I_{i,j}) \\
& \leq &   U_i(b \mid \pi^{a'}_{i,j}) - \frac{1}{3(\height^\calG + 1)}
\cdot \delta \cdot
 \Prob_{(b \mid \pi^{a'}_{i,j})}(I_{i,j})\\
& \leq &   U_i(b \mid \pi^{a'}_{i,j})  - 
 \frac{1}{3 (\height^\calG + 1)} \cdot \delta \cdot
\left( \frac{p^\calG_{0,\min}}{6 \cdot (\height^\calG + 1) \cdot M_\calG 
\cdot |\calG|} \cdot \delta \right)^{\height^\calG }
\end{eqnarray*}

\noindent Thus\footnote{Noting
that  
$\frac{p^\calG_{0,\min}}{12 \cdot (\height^\calG + 1) \cdot M_\calG 
\cdot |\calG|} \cdot \delta < \frac{\delta}{3}$.},
 $U_i(b \mid \pi^a_{i,j}) <  U_i(b \mid \pi^{a'}_{i,j})  - 
 \frac{1}{(\height^\calG + 1)}
\left( \frac{p^\calG_{0,\min}}{12 \cdot (\height^\calG + 1) \cdot M_\calG 
\cdot |\calG|} \cdot \delta \right)^{\height^\calG +1}$.
But this contradicts inequality
(\ref{eq:bound-on-almost-expectation}).
Thus, we must have
$U^j_i(b \mid \pi^a_{i,j}) \geq 
U^j_i(b \mid \pi^{a'}_{i,j}) - \frac{1}{3 (\height^\calG + 1)} \cdot \delta$.
\qed
\end{proof}

Again, let $b \in B^{\epsilon(\calG,\delta)}$ be a  
$\frac{1}{(\height^\calG + 1)} \left( \frac{p^\calG_{0,\min}}{12 \cdot (\height^\calG + 1) \cdot M_\calG \cdot |\calG|} \cdot \delta \right)^ 
{(\height^\calG + 1)}$-almost-$(\frac{p^\calG_{0,\min}}{6 \cdot (\height^\calG + 1) \cdot M_\calG \cdot
|\calG|} \cdot \delta)$-PE.

\begin{claim}
\label{claim:-delta2}
For every player $i$,  
for every integer $m$ where $0 \leq m \leq \height^{\calF_i}$,
for every information set $I_{i,j}$ such that $\height^{\calF_i}_j = m$,
and for every pure strategy $\pi^{c}_i \in B_i$ for player $i$:
$$U^j_i(b) \geq U^j_i(b \mid_{m} \pi^{c}_i) -  \frac{m+1}{(\height^{\calF_i}+1)} 
\cdot \delta$$
\end{claim}

\begin{proof}
The proof is by induction on $m$, using Claim \ref{claim:delta-1},
starting with base case $m = 0$.

\noindent {\em Base case:}  For $m=0$ 
consider an information set $I_{i,j}$
such that $\height^{\calF_i}_j \geq  0$.  
This means that $j$ is a leaf node in the 
directed information set forest $\calF_i$.  
So, for any pure strategy $\pi^c_i$,
suppose the local pure strategy (i.e., local action) 
chosen at $I_{i,j}$ within the pure strategy $\pi^c_i$
is $a' \in \calA_{i,j}$. 
Note that we then  have 
$U^j_i(b \mid_m \pi^c_i) = U^j_i(b \mid \pi^{a'}_{i,j})$.
Thus, we have to show that 
$U^j_i(b)  \geq U^j_i(b \mid \pi^{a'}_{i,j}) -  
\frac{1}{(\height^{\calF_i}+1)} \cdot \delta$.

For the local strategy $b_{i,j}$,  and for $\eta \geq 0$,
let $b^{> \eta}_{i,j} = \sum_{\{a \in \calA_{i,j} \mid b_{i,j,a} > \eta\}} b_{i,j,a}$.
Likewise, let 
$b^{\leq \eta}_{i,j} = \sum_{\{a \in \calA_{i,j} \mid b_{i,j,a} \leq \eta\}} b_{i,j,a}$.
For $\bigtriangledown \in \{ > , \leq \}$,
for $\epsilon > 0$, 
for a behavior profile $b \in B^\epsilon$, 
and for $\eta \geq 0$, 
let $U^{j,\bigtriangledown \eta}_i(b)$ denote the conditional expected payoff
to player $i$, under profile $b$, conditioned on the event
that the play both reaches information set $I_{i,j}$, and thereupon plays
some action in the set $\{a \in \calA_{i,j} \mid b_{i,j,a} \bigtriangledown 
\eta\}$.
Note that, for the profile $b \in B^{\epsilon(\calG,\delta)}$,
the conditional expected payoff $U^j_i(b)$ can be written as:
\begin{equation}
\label{eq:conditional-decomp-big-small}
U^j_i(b)  =  U^{j,> \epsilon(\calG,\delta)}_i(b) \cdot b^{> 
\epsilon(\calG,\delta)}_{i,j} + 
U^{j, \leq \epsilon(\calG,\delta)}_i(b) \cdot b^{\leq \epsilon(\calG,\delta)}_{i,j}
\end{equation}
But then, for any $a' \in \calA_{i,j}$, we have
\begin{eqnarray*}
U^j_i(b) & = & U^{j,> \epsilon(\calG,\delta)}_i(b) \cdot b^{> 
\epsilon(\calG,\delta)}_{i,j} + 
U^{j, \leq \epsilon(\calG,\delta)}_i(b) \cdot b^{\leq \epsilon(\calG,\delta)}_{i,j}\\
& \geq &  U^{j,> \epsilon(\calG,\delta)}_i(b) \cdot b^{> 
\epsilon(\calG,\delta)}_{i,j} \\
& \geq & (U^j_i(b \mid \pi^{a'}_{i,j}) - \frac{1}{3 \cdot (\height^\calG + 1)} \cdot \delta)
\cdot b^{> 
\epsilon(\calG,\delta)}_{i,j}  \quad \quad \mbox{(by Claim \ref{claim:delta-1})}\\
& \geq &  (U^j_i(b \mid \pi^{a'}_{i,j}) - \frac{1}{3 \cdot (\height^\calG + 1)} \cdot \delta)
\cdot (1- |\calA_{i,j}| \cdot \epsilon(\calG,\delta))\\
& = &     (U^j_i(b \mid \pi^{a'}_{i,j}) - \frac{1}{3 \cdot (\height^\calG + 1)} \cdot \delta)
\cdot (1- |\calA_{i,j}| \cdot   \frac{p^\calG_{0,\min}}{12 \cdot (\height^\calG + 1) \cdot M_\calG \cdot |\calG|} \cdot \delta )\\ 
& \geq &  U^j_i(b \mid \pi^{a'}_{i,j})  - \frac{1}{3 \cdot (\height^\calG + 1)} \cdot \delta 
- U^j_i(b \mid \pi^{a'}_{i,j}) \cdot (|\calA_{i,j}| \cdot   \frac{p^\calG_{0,\min}}{12 \cdot (\height^\calG + 1) \cdot M_\calG \cdot |\calG|} \cdot \delta)\\
& \geq &   U^j_i(b \mid \pi^{a'}_{i,j})  - \frac{1}{3 \cdot (\height^\calG + 1)} \cdot \delta  
-   \frac{1}{3 \cdot (\height^\calG + 1)} \cdot \delta \\
& &    \quad \quad \quad \quad 
\mbox{(because  $U^j_i(b \mid \pi^{a'}_{i,j}) \leq M_\calG$ and  $|\calA_{i,j}| \leq |\calG|$ and
$p^\calG_{0,\min} \leq 1$)}\\
& \geq &  U^j_i(b \mid \pi^{a'}_{i,j})  -  \frac{1}{(\height^{\calF_i} + 1)} \cdot \delta
\quad \quad \mbox{(because $\height^{\calF_i} \leq \height^{\calG}$, for all $i$, 
and $\frac{2}{3} \leq 1$.)}
\end{eqnarray*}

Thus $U^j_i(b) \geq  U^j_i(b \mid \pi^{a'}_{i,j})  -  \frac{1}{(\height^{\calF_i} + 1)} \cdot \delta$,
which completes the proof of the base case.\footnote{Let us remark that
we could have opted for a proof that renders the base case trivial,
and ``swallows'' it into the inductive
case, but we felt this would have come at the expense of clarity.}

\noindent {\em Inductive case:}  Assume 
the claim is true for $m-1$ such that  $0 \leq m-1 < \height^{\calF_i}$.  
We want to show it holds for $m$.
Again, consider any pure strategy $\pi^c_i$ for player $i$,
and suppose that $\pi^c_i(j) = a'$.  In other words, in information
set $I_{i,j}$, the action chosen by $\pi^c_i$ is $a'$.

Let $J^{i}(j,a') = \{ j' \in [d_i] \mid  (j,a',j') \in E^{\calF_i} \}$ denote the set of children $j'$ of $j$ in the forest $\calF_i$, such
that the edge from $j$ to $j'$ is labeled by $a'$.
(In other words, $J^{i}(j,a')$ denotes the information
 sets belonging to player $i$ that could possibly be the
 next information set for that player which is reached, after 
 reaching information set $j$.)
For $j' \in J^{i}(j,a')$,  let 
$\Probb^i_{(b \mid \pi^{a'}_{i,j})}(j' \mid j)$
denote the conditional probability of reaching
information set $I_{i,j'}$,  conditioned on
event of reaching information set $I_{i,j}$ and thereupon
taking action $a' \in \calA_{i,j}$, under profile $b$.
Furthermore, let $\Probb^i_{(b \mid \pi^{a'}_{i,j})}( \neg J^{i}(j,a') 
\mid j)$ denote the conditional probability of 
not reaching any information set in $J^{i}(j,a')$,
conditioned on the event of reaching $I_{i,j}$ and thereupon taking
action $a'$.
Finally, let $U_i^{j, \neg J^{i}(j,a')}(b \mid \pi^{a'}_{i,j})$
denote the conditional expected payoff (under profile $b$), conditioned on
reaching $I_{i,j}$ and thereupon playing $a'$, but thereafter
{\em not} reaching any information set in $J^i(j,a')$.
Note that for all $b \in B^{\epsilon(\calG,\delta)}$,
and every $a' \in \calA_{i,j}$, 
we have:\\
$$(\sum_{j' \in J^i(j,a')} \Probb^i_{(b \mid \pi^{a'}_{i,j})}(j' 
\mid j)) + \Probb^i_{(b \mid \pi^{a'}_{i,j})}( \neg
 J^i(j,a')  \mid  j) = 1.$$

\noindent Note furthermore that:
\begin{eqnarray}
\label{eq:decomp-cond-expect-children-in-info-forest}
U^j_i(b \mid \pi^{a'}_{i,j}) = 
\left( \sum_{j' \in J^{i}(j,a')} U^{j'}_i(b \mid \pi^{a'}_{i,j})
\cdot \Probb^i_{(b \mid \pi^{a'}_{i,j})}(j' \mid j) \right) +
U_i^{j, \neg J^i(j,a')}(b \mid \pi^{a'}_{i,j}) \cdot
\Probb^i_{(b \mid \pi^{a'}_{i,j})}( \neg J^i(j,a')   
\mid  j ) .
\end{eqnarray}

\noindent We now use equation (\ref{eq:decomp-cond-expect-children-in-info-forest}),
the inductive hypothesis, and equation
(\ref{eq:conditional-decomp-big-small}), 
in order to establish
that for any pure strategy $\pi^c_i$ for player $i$, we have
$U^j_i(b) \geq U^j_i(b \mid_{m} \pi^{c}_i) -  
\frac{m+1}{(\height^{\calF_i}+1)} 
\cdot \delta$. 

Suppose that the pure strategy $\pi^c_i$ has $\pi^c_i(j) = a'$.
Observe that in this case: 
\begin{equation}
\label{eq:various-fixing-forms}
(b \mid_m \pi^c_i) = ((b \mid \pi^{a'}_{i,j}) 
\mid_{m-1} \pi^c_i)  = ((b \mid_{m-1} \pi^c_i) \mid
\pi^{a'}_{i,j})
\end{equation} 
Also observe that:
\begin{equation}
\label{eq:doesnt-change-expt-to-mod}
U_i^{j, \neg J^i(j,a')}(b \mid \pi^{a'}_{i,j}) =
U_i^{j, \neg J^i(j,a')}((b \mid \pi^{a'}_{i,j}) \mid_{m-1}
\pi^c_i)
\end{equation}
because this conditional expectation does not change when we
change the strategy $b_i$ in local strategies (at $J^i(j,a')$ 
and below) which we have conditioned
on not reaching.  
We thus have:

\begin{eqnarray*}
U^j_i(b) & = & U^{j,> \epsilon(\calG,\delta)}_i(b) \cdot b^{> 
\epsilon(\calG,\delta)}_{i,j} + 
U^{j, \leq \epsilon(\calG,\delta)}_i(b) \cdot b^{\leq \epsilon(\calG,\delta)}_{i,j}  \quad (by  (\ref{eq:conditional-decomp-big-small}) )\\
& \geq &  U^{j,> \epsilon(\calG,\delta)}_i(b) \cdot b^{> 
\epsilon(\calG,\delta)}_{i,j} \\
& \geq & (U^j_i(b \mid \pi^{a'}_{i,j}) - \frac{1}{3 \cdot (\height^\calG + 1)} \cdot \delta)
\cdot b^{> 
\epsilon(\calG,\delta)}_{i,j}  \quad \quad \mbox{(by Claim \ref{claim:delta-1})}\\
& \geq &  (U^j_i(b \mid \pi^{a'}_{i,j}) - \frac{1}{3 \cdot (\height^\calG + 1)} \cdot \delta)
\cdot (1- |\calA_{i,j}| \cdot \epsilon(\calG,\delta))\\
& = &     (U^j_i(b \mid \pi^{a'}_{i,j}) - \frac{1}{3 \cdot (\height^\calG + 1)} \cdot \delta)
\cdot (1- |\calA_{i,j}| \cdot   \frac{p^\calG_{0,\min}}{12 \cdot (\height^\calG + 1) \cdot M_\calG \cdot |\calG|} \cdot \delta )\\ 
& \geq &  U^j_i(b \mid \pi^{a'}_{i,j})  - \frac{1}{3 \cdot (\height^\calG + 1)} \cdot \delta 
- U^j_i(b \mid \pi^{a'}_{i,j}) \cdot (|\calA_{i,j}| \cdot   \frac{p^\calG_{0,\min}}{12 \cdot (\height^\calG + 1) \cdot M_\calG \cdot |\calG|} \cdot \delta)\\
& \geq &   U^j_i(b \mid \pi^{a'}_{i,j})  - \frac{1}{3 \cdot (\height^\calG + 1)} \cdot \delta  
-   \frac{1}{3 \cdot (\height^\calG + 1)} \cdot \delta \\
& &    \quad \quad \quad \quad 
\mbox{(because  $U^j_i(b \mid \pi^{a'}_{i,j}) \leq M_\calG$ and  $|\calA_{i,j}| \leq |\calG|$ and
$p^\calG_{0,\min} \leq 1$)}\\
& \geq &  U^j_i(b \mid \pi^{a'}_{i,j})  -  \frac{1}{(\height^{\calF_i} + 1)} \cdot \delta
\quad \quad \mbox{(because $\height^{\calF_i} \leq \height^{\calG}$, for all $i$)} \\
& = & 
(\sum_{j' \in J^{i}(j,a')} U^{j'}_i(b \mid \pi^{a'}_{i,j})
\cdot \Probb^i_{(b \mid \pi^{a'}_{i,j})}(j' \mid j))\\
&& 
 + \  U_i^{j, \neg J^i(j,a')}(b \mid \pi^{a'}_{i,j}) \cdot
\Probb^i_{(b \mid \pi^{a'}_{i,j})}( \neg J^i(j,a')   
\mid j) - \frac{1}{(\height^{\calF_i} + 1)} \cdot \delta
\quad \quad \mbox{(by equality 
(\ref{eq:decomp-cond-expect-children-in-info-forest}))}\\
& \geq & 
(\sum_{j' \in J^{i}(j,a')} (U^{j'}_i((b \mid \pi^{a'}_{i,j}) \mid_{m-1}
\pi^c_i) - \frac{m}{(\height^{\calF_i} + 1)} \cdot \delta) 
\cdot \Probb^i_{(b \mid \pi^{a'}_{i,j})}(j' \mid j))\\ 
&&
 + \  U_i^{j, \neg J^i(j,a')}((b \mid \pi^{a'}_{i,j}) \mid_{m-1} \pi^c_i)
- \frac{m}{(\height^{\calF_i} + 1)} \cdot \delta) \cdot                      
\Probb^i_{(b \mid \pi^{a'}_{i,j})}( \neg J^i(j,a')                             
\mid j) - \frac{1}{(\height^{\calF_i} + 1)} \cdot \delta\\
&& (\mbox{by inductive hypothesis, and by (\ref{eq:doesnt-change-expt-to-mod})})\\
& =  & U_i^{j}(b \mid_m \pi^c_i) - \frac{m}{(\height^{\calF_i} + 1)} \cdot \delta - \frac{1}{(\height^{\calF_i} + 1)} \cdot \delta  \quad \quad
(\mbox{by (\ref{eq:various-fixing-forms}) and
(\ref{eq:decomp-cond-expect-children-in-info-forest})})\\
& = & U_i^{j}(b \mid_m \pi^c_i) - 
\frac{m+1}{(\height^{\calF_i} + 1)} \cdot \delta
\end{eqnarray*}

Thus $U_i^j(b) \geq U_i^{j}(b \mid_m \pi^c_i) -                   
\frac{m+1}{(\height^{\calF_i} + 1)} \cdot \delta$. 
This completes the proof of Claim \ref{claim:-delta2}.
\qed
\end{proof}

Part (2.) of Lemma \ref{lem:almost-fixed-is-almost-equilib} 
now follows readily
from Claim \ref{claim:-delta2}.    To see this,
let $J^{\calF_i}$ denote the set of root vertices in the 
information set forest $\calF_i$.  
Let $\Probb^i_b(\neg J^{\calF_i})$ denote
the probability, under
profile $b$, of not reaching any information set in $J^{\calF_i}$.
Finally, let $U^{\neg J^{\calF_i}}_i(b)$ denote 
the conditional expected payoff to player $i$, 
under profile $b$, conditioned 
on the event of not reaching any of the information sets in
$J^{\calF_i}$, and if this event has probability zero, then
by definition we let $U^{\neg J^{\calF_i}}_i(b) := 0$.

Then, for any pure strategy $\pi^c_i$ for player $i$,
we have:

\begin{eqnarray*}
U_i(b) & =  &
(\sum_{j' \in J^{\calF_i}} U^{j'}_i(b)
\cdot \Prob_{b}(I_{i,j'})) \ \ + \ \    U_i^{\neg J^{\calF_i}}(b) \cdot
\Probb^i_{b}( \neg J^{\calF_i})\\
& \geq & (\sum_{j' \in J^{\calF_i}} (U^{j'}_i(b \mid \pi^c_i) - \delta) 
\cdot \Prob_{b}(I_{i,j'})) \ \  + \ \   (U_i^{\neg J^{\calF_i}}(b \mid \pi^c_i) - \delta) \cdot
\Probb^i_{b}( \neg J^{\calF_i})\\ & &  \quad \quad \quad \quad \quad  \quad \quad \mbox{(by applying Claim 
\ref{claim:-delta2}, and since 
$U_i^{\neg J^{\calF_i}}(b) = U^{\neg J^{\calF_i}}_i(b \mid \pi^c_i) )$} \\
&& \mbox{}\\
& =  & U_i(b \mid \pi^c_i) - \delta .
\end{eqnarray*}
Thus $U_i(b) \geq U_i(b \mid \pi^c_i) - \delta$,
which completes the proof of Part (2.) of Lemma 
\ref{lem:almost-fixed-is-almost-equilib}.
\qed
\end{proof}

\noindent Applying Lemma
\ref{lem:almost-fixed-is-almost-equilib},
Proposition 
\ref{prop:weak-approx-ppad}, and Lemma
\ref{lem:poly-con-poly-comp},
we obtain the 
main results of this section:

\begin{theorem}
\label{thm:almost-ppad-complete}
\mbox{}
\begin{enumerate}
\item The problem of computing, given a EFGPR, $\calG$, and given rationals $\delta > 0$ and $\epsilon >0$,
a $\delta$-almost-$\epsilon$-PE of $\calG$,  is \PPAD-complete.

Likewise, the problem of computing, given a EFGPR, $\calG$, and given rationals $\delta > 0$ and $\epsilon >0$,
a $\delta$-almost-$\epsilon$-QPE of $\calG$,  is \PPAD-complete.

\item (cf. \cite{DFP06})  The problem of computing, given a EFGPR, $\calG$, and given a rational $\delta > 0$,
a $\delta$-almost-SGPE of $\calG$ is \PPAD-complete.
\end{enumerate}
\end{theorem}

\begin{proof}
First, we establish containment in $\PPAD$ for all the problems:
\begin{enumerate}
\item 
The fact that computing a $\delta$-almost-$\epsilon$-PE, and computing a $\delta$-almost-$\epsilon$-QPE
for a given EFGPR, $\calG$, and given $\delta > 0$ and $\epsilon > 0$,
is in $\PPAD$ follows 
immediately from Lemma \ref{lem:almost-fixed-is-almost-equilib},
Parts (1.) and (2.), 
Proposition 
\ref{prop:weak-approx-ppad}, and Lemma
\ref{lem:poly-con-poly-comp}.

Specifically, by  Lemma \ref{lem:almost-fixed-is-almost-equilib},
Parts (1.),
for $0 < \delta < 1$ and $0 < \epsilon < 1$, 
a profile $b \in B^{\epsilon/2}$, such 
that $\| b - F^{\epsilon/2}_\calG(b) \|_\infty <  \frac{\epsilon \cdot \delta}{3}$,
is also a $\delta$-almost-$\epsilon$-PE.
Likewise,  profile $b \in B^{\epsilon/2}$, such 
that $\| b - H^{\epsilon/2}_\calG(b) \|_\infty <  \frac{\epsilon \cdot \delta}{3}$,
is a $\delta$-almost-$\epsilon$-QPE.

But by Proposition \ref{prop:weak-approx-ppad} and Lemma
\ref{lem:poly-con-poly-comp}, since the functions $F^\epsilon_{\calG}(b)$
and $H^{\epsilon}_{\calG}(b)$ are
polynomially computable and polynomially continuous
(with respect to the input $\langle G , \epsilon \rangle$),
the problem of computing such a profile $b$ is in $\PPAD$.

\item 
For $\delta > 0$, let $\epsilon(\calG,\delta) :=
  \frac{p^\calG_{0,\min}}{12 \cdot (\height^\calG +1) \cdot M_\calG \cdot |\calG|} \cdot \delta$.
Let $\delta' = \frac{1}{3 \cdot
    (\height^\calG + 1)} \cdot \epsilon(\calG,\delta)^{(\height^\calG +
    1)}$.

Since $\delta' <  \epsilon(\calG,\delta)$,
we have $(\delta' + \epsilon(\calG,\delta)) \leq (2 \cdot 
 \epsilon(\calG,\delta))$.
It thus follows from 
Lemma \ref{lem:almost-fixed-is-almost-equilib}, Part 1., that
if $b \in B^{\epsilon(\calG,\delta)}$ 
satisfies $\| b -
  F^{\epsilon(\calG,\delta)}_{\calG}(b) \|_{\infty} < \delta'$, then\\ $b$
is a
$\frac{1}{(\height^\calG+1)} \cdot  
\epsilon(\calG,\delta)^{(\height^\calG + 1)}$-almost-$(2 \cdot
\epsilon(\calG,\delta))$-PE.
But then Lemma 
\ref{lem:almost-fixed-is-almost-equilib}, Part 2.,
implies that $b$ is also a 
$\delta$-almost subgame perfect equilibrium of $\calG$.

Thus, the problem computing a $\delta$-almost-SGPE of $\calG$ is P-time reducible to the
problem of computing a $b \in B^{\epsilon(\calG,\delta)}$
such that 
$\| b -
  F^{\epsilon(\calG,\delta)}_{\calG}(b) \|_{\infty} < \delta'$.
But since both $\epsilon(\calG,\delta) > 0$
and $\delta' > 0$  
are rational numbers both of whose encoding size
(in binary) is polynomial in the encoding size of the
input $\langle \calG, \delta \rangle$,
by Proposition
\ref{prop:weak-approx-ppad}, computing a $\delta$-almost-SGPE
is in $\PPAD$.
\end{enumerate}

Finally, to see that both problems are $\PPAD$-hard,  recall
that Daskalakis, Goldberg, and Papadimitriou \cite{DasGP09}
established that computing a $\delta$-almost NE (a.k.a., a
$\delta$-NE, in the terminology they used), given a $n$-player 
{\em normal form} game, $\Gamma$, and given $\delta > 0$, is $\PPAD$-hard.  
Now recall that a $n$-player NFG, $\Gamma$, is trivially encodable
as a $n$-player EFGPR, $\calE(\Gamma)$, 
and note that a $\delta$-almost-SGPE of $\calE(\Gamma)$ is also a $\delta$-almost-NE of $\Gamma$.
\qed 
\end{proof}

A simple corollary of Theorem \ref{thm:almost-ppad-complete} is that
computing an $\delta$-almost-$\epsilon$-PE for a NFG
is also $\PPAD$-complete.

\begin{corollary}
The problem of computing, given a NFG, $\Gamma$, and given rationals $\delta > 0$ and $\epsilon >0$,
a $\delta$-almost-$\epsilon$-PE of $\Gamma$,  is \PPAD-complete.
\end{corollary}
\begin{proof}
This follows by applying Theorem \ref{thm:almost-ppad-complete} (
Part 1.) to the ``equivalent'' EFGPR,  $\calE(\Gamma)$, which 
we can easily construct from $\Gamma$,
and from  the fact 
that $\calE(\Gamma)$
has exactly the 
same $\delta$-almost-$\epsilon$-PEs  (in behavior strategies) 
as $\Gamma$ does (in mixed strategies).
This follows easily from the payoff-preserving one-to-one correspondence
between the mixed profiles of $\Gamma$ and the behavior profiles of $\calE(\Gamma)$.
\qed
\end{proof}

We have suggested that the notion of 
a $\delta$-almost-$\epsilon$-PE,  is a 
reasonable ``almost'' relaxation of ($\epsilon$-)PE, allowing
for its computation in $\PPAD$  (i.e., using path following
algorithms), in the same way that 
$\delta$-NE (= $\delta$-almost-NE) serves as a relaxation of NE. 

We have thusfar
not defined a ``almost'' relaxation for {\em sequential
equilibrium} (SE).  Since PE ``refines'' SE  (see Proposition 
\ref{prop:kreps-wilson}), 
a possible definition is this:  ``an
assessment $(b',\mu^{b'})$, where the behavior profile
$b'$ is 
a $\delta$-almost-$\epsilon$-PE, and where $\mu^{b'}$ is 
the belief system generated by $b'$ ''.
This is well-defined, because for $\epsilon > 0$, any $\delta$-almost-$\epsilon$-PE, $b'$,
is fully mixed, and thus 
the belief system $\mu^{b'}$ that it generates is uniquely
defined; and we can compute $\mu^{b'}$ efficiently, given $b'$ and $\calG$.
So, we can take this as our definition of a 
``almost'' relaxation of SE.
Theorem \ref{thm:almost-ppad-complete} then implies that
computing such an ``almost'' SE, 
given $\calG$, and 
given $\delta > 0$ and 
$\epsilon > 0$, is $\PPAD$-complete.

\section{Conclusions}

We have characterized the complexity of approximating 
various refinements of equilibrium,
and ``almost equilibrium'',
for extensive form games of perfect recall with $n \geq 3$ players.

Specifically, we have shown that the complexity of approximate 
(or almost) equilibrium 
computation for extensive form games of perfect recall, 
with $n \geq 3$ players, 
including for fundamental refinements such as sequential and (quasi-)perfect equilibrium,
is the same as that of approximate (or almost) Nash equilibrium 
computation for normal form games with 3 players.
Namely, these problems are, respectively, $\FIXPA$-complete and $\PPAD$-complete.

Although our results establish that approximating 
a  PE  for a $n$-player EFGPR,  is in $\FIXP_a$,  
our results {\em do not} imply 
that computing an actual ({\em real-valued}) 
PE  for an $n$-player EFGPR is in $\FIXP$. 
We leave this as an open question, although the more relevant
question, from the point of view of the standard (Turing) model of computation, 
is containment in $\FIXPA$ (in $\PPAD$) for approximation 
(respectively, ``almost'' computation), which we have 
established.

\noindent Some natural open questions suggest themselves:

\begin{enumerate}
\item
The complexity of approximating
a {\em proper equilibrium} for $n$-player NFGs.
Proper equilibrium,
defined by Myerson in \cite{Myerson78}, is an important refinement of PE for 
NFGs\footnote{Peter Bro Miltersen, in conversation with the author, 
has referred to proper equilibrium as 
{\em ``the mother of all''} refinements of equilibrium for NFGs.}
, which Myerson showed always exists for any NFG.

It is defined as follows:  for an NFG, $\Gamma$, and for $\epsilon > 0$,
a mixed strategy profile $\sigma = (\sigma_1, \ldots, \sigma_n)$
is called a {\em $\epsilon$-proper equilibrium} if it is (a.): fully mixed, and (b.):   for every two pure strategies
$c, c'$ of any player $i$,   if $U_i(\sigma \mid \pi^c_i) < U_i(\sigma \mid \pi^{c'}_i)$ then
$\sigma_i(c) \leq \epsilon \cdot \sigma_i(c')$.
A {\em proper equilibrium} is defined to be a limit point of a sequence of $\epsilon_k$-proper equilibria,
where $\epsilon_k > 0$ for all $k \in \nat$, and where  $\lim_{k \rightarrow \infty} \epsilon_k = 0$.

There are connections between proper equilibrium for NFGs and
QPEs of EFGPRs.
In particular,  van Damme \cite{vanDamme84} showed
that a proper equilibrium for an NFG, $\Gamma$, induces a QPE in every EFGPR whose 
(standard) normal form is $\Gamma$.   However, the other direction does not hold:
there are EFGPRs with a QPE (or PE) which is not induced by a proper equilibrium in
a corresponding normal form game. 
S{\o}rensen \cite{Sorensen12} has given a Lemke-like algorithm for computing a proper
equilibrium for 2-player NFGs.\footnote{For NFGs with $n \geq 3$ players, Yamamoto 
\cite{Yamamoto93}
outlined a procedure for approximating a proper equilibrium
based on a {\em continuous} homotopy path following approach,
but as indicated by S{\o}rensen in \cite{Sorensen12}, even for 2-player NFGs it is 
unclear under what conditions Yamamoto's procedure is guaranteed to converge to
an approximate
proper equilibrium.}
Can we approximate a proper equilibrium for  $n$-player NFGs in $\FIXPA$?

\item  One can adapt Myerson's definition of ($\epsilon$-)proper equilibrium
in a natural way,
to define a notion of {\em extensive form} ($\epsilon$-){\em proper equilibrium} (PropE)
as well as ($\epsilon$-){\em quasi-proper equilibrium} (QPropE) for EFGPRs.
PropE refines PE, and  likewise QPropE refines QPE, for EFGPRs.
Such refinements for EFGPRs were already alluded to briefly by van Damme in \cite{vanDamme84}\footnote{As 
van Damme remarks in \cite{vanDamme84}, his main result actually shows that every 
proper equilibrium of  
an NFG, $\Gamma$, induces a QPropE in every EFGPR which has $\Gamma$ as its 
(standard) normal form.},
but we are unaware of any subsequent study of them.
Myerson's existence proof of proper equilibrium for NFGs 
 can be suitably adapted to show 
existence of both a PropE and a QPropE for any EFGPR.
Can we approximate a PropE, and a QPropE, for $n$-player EFGPRs in $\FIXPA$?
\end{enumerate}

\noindent We believe the answer to both of the above questions is ``Yes''.\\
Even if the answers are ``yes'',  it is not entirely
clear what the suitable ``$\delta$-almost'' relaxations of (qausi-)proper equilibrium 
should be. We need such relaxations to place
the problems in $\PPAD$, i.e., to enable discrete path following
algorithms that compute a suitably refined ``almost equilibrium''.
One natural attempt is to define such a relaxation as follows:
a {\em $\delta$-almost-$\epsilon$-proper equilbrium} for NFGs is
a mixed strategy profile $\sigma$ which is (a.): fully mixed, and (b.):   for every two pure strategies
$c, c'$ for any player $i$,   if $U_i(\sigma \mid \pi^c_i) < U_i(\sigma \mid \pi^{c'}_i) - \delta$ then
$\sigma_i(c) \leq \epsilon \cdot \sigma_i(c')$.
It remains to be seen whether this definition is the ``right'' one, and in particular whether computing
such a $\delta$-almost relaxation can be placed in PPAD.

We want to again highlight that we believe our results
can potentially
provide a  ``practical'' computation method for computing a
``almost'' ($\epsilon$-perfect, and $\epsilon$-quasi-perfect) equilibrium for
EFGPRs, with $n \geq 3$ players, by applying Scarf-like
discrete path following algorithms
on the ``small'' algebraic
fixed point functions that we have developed for $n$-player EFGPRs.
We believe this is a promising computational approach 
that should be implemented and
explored experimentally.

\vspace*{0.1in}

\noindent{\bf Acknowledgements.}
Thanks to Peter Bro Miltersen for several helpful comments.
Thanks to Mihalis Yannakakis for our collaboration on \cite{EY07}, which 
provides both the perspective, and a number of the tools, used in this paper.
Thanks to my co-authors
 Peter Bro Miltersen,
 Kristoffer Hansen, and 
Troels S{\o}rensen on the recent paper \cite{EHMS14}, on which this paper 
builds directly.
Thanks to Costis Daskalakis for clarifications about the result
in \cite{DFP06} for extensive form games.

\bibliographystyle{plain}

\end{document}